\documentclass[12pt]{article}

\usepackage[colorlinks, linkcolor=blue, citecolor=blue]{hyperref}           
\usepackage{color}
\usepackage{graphicx,subfigure,amsmath,amssymb,amsfonts,bm,epsfig,epsf,url,dsfont}
\usepackage{amsthm}

\newtheorem{thm}{Theorem}[section]
\newtheorem{lem}[thm]{Lemma}
\newtheorem{assum}[thm]{Assumption}
\newtheorem{proposition}[thm]{Proposition}
\newtheorem{remark}[thm]{Remark}
\newtheorem{definition}[thm]{Definition}

\usepackage{tikz}
\usepackage{bbm}
\usepackage[small,bf]{caption}
\usepackage[top=1in,bottom=1in,left=1in,right=1in]{geometry}
\usepackage{fancybox}
\usepackage{hyperref}
\usepackage{algorithm}
\usepackage{verbatim}
\usepackage[titletoc,toc]{appendix}

\usepackage{mathtools}

\usepackage{scalerel,stackengine}
\stackMath
\newcommand\reallywidehat[1]{%
\savestack{\tmpbox}{\stretchto{%
  \scaleto{%
    \scalerel*[\widthof{\ensuremath{#1}}]{\kern-.6pt\bigwedge\kern-.6pt}%
    {\rule[-\textheight/2]{1ex}{\textheight}}
  }{\textheight}%
}{0.5ex}}%
\stackon[1pt]{#1}{\tmpbox}%
}

\makeatletter
\renewcommand{\paragraph}{%
  \@startsection{paragraph}{4}%
  {\z@}{1.25ex \@plus 1ex \@minus .2ex}{-1em}%
  {\normalfont\normalsize\bfseries}%
}
\makeatother

\newcommand*{\rom}[1]{\expandafter\@slowromancap\romannumeral #1@}

\newcommand{\abs}[1]{\left|#1\right|}

\newcommand{\EfN}{EfN }
\DeclareMathOperator*{\argmax}{arg\,max}

\newcommand{\calD}{{\cal D}}

\newcommand{\calF}{{\cal F}}

\newcommand{\calL}{{\cal L}}

\newcommand{\calN}{{\cal N}}

\newcommand{\calP}{{\cal P}}

\newcommand{\be}{\begin{equation}}
\newcommand{\ee}{\end{equation}}
\newcommand{\beqna}{\begin{eqnarray}}
\newcommand{\eeqna}{\end{eqnarray}}

\newcommand{\p}[1]{\left(#1\right)}
\newcommand{\pp}[1]{\left[#1\right]}
\newcommand{\ppp}[1]{\left\{#1\right\}}
\newcommand{\norm}[1]{\left\|#1\right\|}

\newcommand{\s}[1]{\mathsf{#1}}
\usepackage{xspace}

\numberwithin{equation}{section}

\allowdisplaybreaks
\usepackage[utf8]{inputenc}
\usepackage{authblk}
\usepackage{url}
\usepackage[left=1in,right=1in,top=1in,bottom=1in]{geometry}

\begin{document}

\title{Einstein from Noise: Statistical Analysis}
\author[1]{Amnon Balanov\thanks{Corresponding author: \url{amnonba15@gmail.com}}}
\author[1]{Wasim Huleihel}
\author[1]{Tamir Bendory}

\affil[1]{\normalsize School of Electrical and Computer Engineering, Tel Aviv University, Tel Aviv 69978, Israel}

\maketitle

\begin{abstract}
``Einstein from noise" (EfN) is a prominent example of the model bias phenomenon, where systematic errors in the statistical model lead to spurious but consistent estimates. In the \EfN experiment, one falsely believes that a set of observations contains noisy, shifted copies of a template signal (e.g., an Einstein image), whereas in reality, it contains only pure noise observations. To estimate the signal, the observations are first aligned with the template using cross-correlation and then averaged. Although the observations contain nothing but noise, it was recognized early on that this process produces a signal that resembles the template signal! This model bias pitfall was at the heart of a central scientific controversy about validation techniques in structural biology.

This paper provides a comprehensive statistical analysis of the \EfN phenomenon above. We show that the Fourier phases of the \EfN estimator (namely, the average of the aligned noise observations) converge to the Fourier phases of the template signal, thereby explaining the observed structural similarity. Additionally, we prove that the convergence rate of Fourier phases is inversely proportional to the number of noise observations and, in the high-dimensional regime, to the Fourier magnitudes of the template signal. Moreover, in the high-dimensional regime, the \EfN estimator converges to a scaled version of the template signal. 
This work not only deepens the theoretical understanding of the EfN phenomenon but also highlights potential pitfalls in template matching techniques and emphasizes the need for careful interpretation of noisy observations across disciplines in engineering, statistics, physics, and biology.
\end{abstract}

\newpage
\tableofcontents
\newpage

\section{Introduction}
Model bias is a fundamental pitfall arising across a broad range of statistical problems, leading to consistent but inaccurate estimations due to systematic errors in the model. 
This paper focuses on the \textit{Einstein from Noise} (EfN) experiment: a prototype example of model bias that appears in template matching techniques. Consider a scenario where scientists acquire observational data and genuinely believe their observations contain noisy, shifted copies of a known template signal. However, in reality, their data consists of pure noise with no actual signal present.

To estimate the (absent) signal, the scientists align each observation by cross-correlating it with the template and then average the aligned observations. Remarkably, empirical evidence has shown, multiple times, that the reconstructed structure from this process is structurally similar to the template, even when all the measurements are pure noise~\cite{henderson2013avoiding,shatsky2009method,sigworth1998maximum}. This phenomenon stands in striking contrast to the prediction of the unbiased model, that averaging pure noise signals would converge towards a signal of zeros, as the number of noisy observations diverges. Thus, the above \EfN estimation procedure is biased towards the template signal.

While the \EfN phenomenon has been analyzed in prior work (see Section~\ref{sec:cryoEM} for more details), a comprehensive theoretical understanding of the \EfN model remains limited. This work contributes to filling that gap by rigorously analyzing the relationship between the reconstructed signal and the underlying template.
The term 'Einstein from Noise' was popularized in~\cite{shatsky2009method}, where the authors illustrated the phenomenon using an image of Einstein as the template signal. However, the underlying effect had been observed earlier in the cryogenic electron microscope literature (see, for instance,~\cite{sigworth1998maximum}, and further details in Section \ref{sec:cryoEM}). In this work, we refer to the average of the aligned pure noise observations as the \emph{\EfN estimator}. A detailed formulation of the problem is provided in Section~\ref{sec:problem_formulation}.

Figure~\ref{fig:1} illustrates the \EfN process, which consists of two key stages. First, the observations are aligned with the template signal to achieve optimal alignment. Then, the aligned observations are averaged. The result is the \EfN estimator, which shares a structural resemblance to the template image, though it is not an identical reproduction.
 
\begin{figure}[t!]
    \centering
    \includegraphics[width=0.9\linewidth]{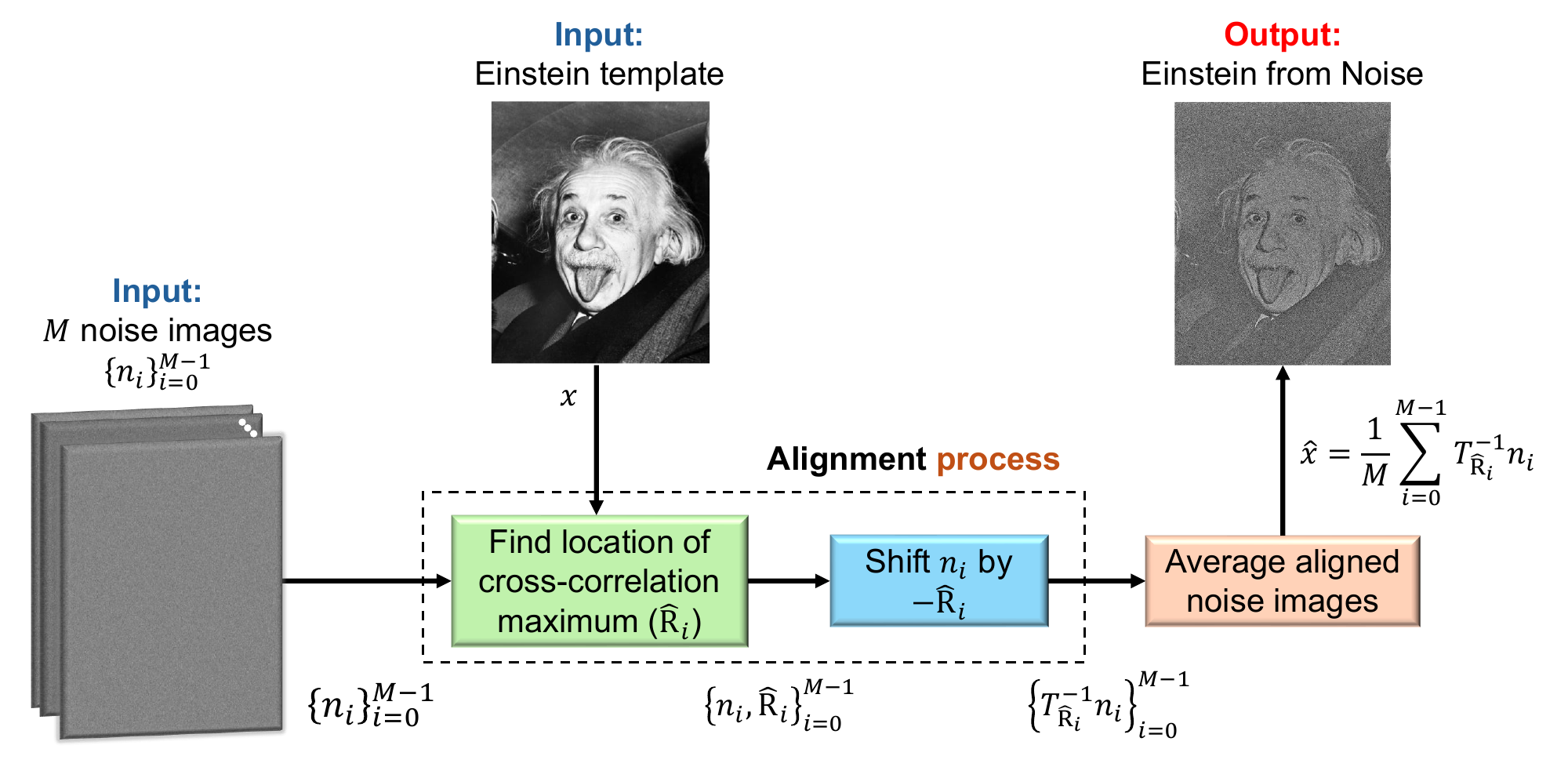}
    \caption{\textbf{Einstein from Noise.} The \EfN estimator consists of three stages: (1) finding the index of the maximum of the cross-correlation ($\hat{\s{R}}_i$) between the $i$-th noise signal ($n_i$) and the template signal (e.g., Einstein's image); (2) cyclically shifting the noise signal by $-\hat{\s{R}}_i$; (3) averaging the shifted noise signals. In this paper, we characterize the relationship between the output of this process---\emph{the \EfN estimator}---and the template signal.}
    \label{fig:1}
\end{figure}

\paragraph{Main results.} 

The central results of this work are as follows. Our first result, stated in Theorem~\ref{thm:1}, shows that the Fourier phases of the \EfN estimator converge to the Fourier phases of the template signal, as the number of noisy observations (denoted by $M$) converges to infinity. However, it is important to note that the \EfN estimator's \textit{Fourier magnitudes} do not necessarily converge to those of the template signal. We also show that the Fourier phases’ mean squared error (MSE) decays to zero with
a rate of $1/M$. Since the Fourier phases are responsible for the formation of geometrical image elements, such as contours and edges~\cite{oppenheim1981importance, shechtman2015phase}, this clarifies why the resulting \EfN estimator image exhibits a structural similarity to the template, but not necessarily a full recovery. 
Our second result, stated in Theorem~\ref{thm:2}, proves that in the high-dimensional regime, where the dimension of the signal diverges, the convergence rate of the Fourier phases is inversely proportional to the square of the Fourier magnitudes of the template signal. In this case, the Fourier magnitudes of the \EfN estimator converge to a scaled version of the template's Fourier magnitudes.

While Theorems~\ref{thm:1} and~\ref{thm:2} are proved under the assumption of white Gaussian noise, we also extend our analysis to more general noise models. 
In particular, we show that, although the convergence results in Theorems~\ref{thm:1} and~\ref{thm:2} do not necessarily hold under arbitrary noise statistics, several structural properties of the \EfN estimator persist.
First, in Proposition~\ref{prop:positiveCorrelation}, we show that the \EfN estimator remains positively correlated with the template for arbitrary noise statistics, even when the Fourier phases do not converge. Since the correlation between images often implies visual resemblance, this explains why the \EfN estimator still exhibits structural similarity to the template.
Second, in Theorem~\ref{thm:highDimentionalNoiseExtention}, we show that in the high-dimensional limit, if the noise signal is independent and identically distributed (i.i.d) (not necessarily Gaussian), then the same phase convergence behavior observed in the white Gaussian case still holds.
Finally, in Proposition~\ref{prop:circulantGauusianNoise}, we demonstrate that if the noise signal is Gaussian with circular symmetry, then the conclusions of Theorem~\ref{thm:1} remain valid, even though the noise is not white.

\paragraph{Organization.} 

The remainder of this paper is organized as follows. Section~\ref{sec:problem_formulation} provides a detailed formulation of the problem. Section \ref{sec:cryoEM} discusses the connection between the \EfN problem and single-particle cryo-electron microscopy (cryo-EM), the primary motivation for this work, and presents supporting empirical demonstrations. Our main theoretical results for white Gaussian noise observations, Theorems~\ref{thm:1} and~\ref{thm:2}, are stated in Section~\ref{sec:main_results}.
Extensions of these results to noise models beyond white Gaussian noise are presented in Section~\ref{sec:extenstionToOtherNoise}. Finally, we conclude with a discussion and outlook in Section~\ref{sec:outlook}.

\section{Problem Formulation and Notation}
\label{sec:problem_formulation}

This section outlines the probabilistic model behind the \EfN experiment and delineates our main mathematical objectives. 
Although the \EfN phenomenon is described typically for images, we will formulate and analyze it for one-dimensional signals, bearing in mind that the extension to two-dimensional images is straightforward (see Section~\ref{sec:outlook} for more details). 

\paragraph{Notations.} Throughout the rest of this paper, we use $\xrightarrow[]{\calD}$, $\xrightarrow[]{\calP}$, $\xrightarrow[]{\s{a.s.}}$, and $\xrightarrow[]{\calL^p}$, to denote the convergence of sequences of random variables in distribution, in probability, almost surely, and in $\calL^p$ norm, respectively.
Inner products in the Euclidean space between vectors $a$ and $b$ are  written as either $a^\top b$ or $\langle a, b \rangle$.

\paragraph{Problem formulation.}
Consider a scenario where scientists collect a series of observations under the belief that each observation is a noisy, randomly shifted version of a known \emph{template signal} $x \in \mathbb{R}^d$ (for example, an image of Einstein). Formally, the assumed postulated data model is given by:
\begin{align} \label{eqn:calQModel}
    \text{(Postulated model)} \quad y_i = \mathcal{T}_{\ell_i} \cdot x + n_i,
\end{align}
where $\mathcal{T}_\ell : \mathbb{R}^d \to \mathbb{R}^d$ is the cyclic shift operator defined by $[\mathcal{T}_\ell z]_r \triangleq z_{(r - \ell) \bmod d}$ for all $z \in \mathbb{R}^d$ and indices $0 \leq r \leq d-1$, and $n_i \sim \mathcal{N}(0, \sigma^2 I_{d \times d})$ are i.i.d. Gaussian noise vectors.
In reality, however, there is no underlying signal: the observations consist entirely of white Gaussian noise. That is, the true data-generating process follows the underlying model:
\begin{align} \label{eqn:calPModel}
    \text{(Underlying model)} \quad y_0, y_1, \ldots, y_{M-1} \stackrel{\text{i.i.d.}}{\sim} \mathcal{N}(\mathbf{0}, \sigma^2 I_{d \times d}),
\end{align}
where $M$ denotes the number of observations. Since the data consists purely of white Gaussian noise, we will explicitly write $y_i = n_i$ to emphasize this fact.

To estimate the (nonexistent) signal, the scientists align each observation to the template $x$ using cross-correlation, and then average the aligned observations. Specifically, for each $i = 0, \ldots, M-1$, they compute the shift that maximizes the inner product with the template:
\begin{align} \label{eqn:OptShiftRealSpace}
    \hat{\s{R}}_i \triangleq \underset{0 \leq \ell < d}{\arg\max} \, \langle n_i, \mathcal{T}_\ell x \rangle,
\end{align}
where $n_i$ is the $i$-th noise observation, and $\s{\hat{R}}_i$ defines the optimal cyclic shift that aligns the template signal $x$ with the noise observation $n_i$ in terms of cross-correlation.

Then, the \EfN estimator is given by the average of the noise observations, but each is first aligned according to the above maximal shifts, i.e.,
\begin{align}
    \hat{x}\triangleq \frac{1}{M} \sum_{i=0}^{M-1} {\mathcal{T}_{\s{-\hat{R}}_i}} n_i, \label{eqn:efnEstimatorRealSpace}
\end{align}
where $\mathcal{T}_{\s{-\hat{R}}_i} n_i$ represents the noise observation $n_i$ aligned by applying the inverse cyclic shift $-\s{\hat{R}}_i$ to best match the template signal. Throughout the text, we refer to $\hat{x}$ as the \EfN estimator.

The \EfN phenomenon states that, at least empirically, $\hat{x}$ and $x$ appear ``close" in some sense; our goal is to understand this phenomenon mathematically. To that end, we will consider the two asymptotic regimes: the first corresponds to the classical setting where the number of observations $M \to \infty$ while the dimension $d$ is fixed (i.e., the number of observations diverges to infinity and the template vector dimension is fixed); the second is the high-dimensional regime, where $d \to \infty$ after $M \to \infty$. (i.e., both the number of observations and the dimension of the template signal diverge).


\paragraph{Fourier space notation.}
As will become clear in the next sections, it is convenient to work in the Fourier domain. Let $\phi_{\s{Z}}\triangleq\sphericalangle\s{Z}$ denote the phase of a complex number $\s{Z}\in\mathbb{C}$, and recall that the discrete Fourier transform (DFT) of a $d$-length signal $y\in\mathbb{R}^d$ is given by,
\begin{align}
    \s{Y}[k] \triangleq \calF\ppp{y} = \frac{1}{\sqrt{d}}\sum_{\ell=0}^{d-1}y_\ell e^{-j\frac{2\pi}{d}k\ell},
    \label{eqn:def-DFT}
\end{align}
where $j\triangleq\sqrt{-1}$, and $0\leq k\leq d-1$. Accordingly, we let $\s{X}$, $\hat{\s{X}}$, and $\s{N}_i$, denote the DFTs of $x$, $\hat{x}$, and $n_i$, respectively, for $0\leq i\leq M-1$.  These DFT sequences can be equivalently represented in the magnitude-phase domain as follows, 
\begin{equation}\label{FourierSpace}
		\s{X} = \{\abs{\s{X}[k]}e^{j\phi_{\s{X}}[k]}\}_{k=0}^{d-1}, \quad 
        \hat{\s{X}}  = \{|\hat{\s{X}}[k]|e^{j\phi_{\s{\hat{X}}}[k]}\}_{k=0}^{d-1}, \quad 
        \s{N}_i = \{\abs{\s{N}_i[k]} e^{j\phi_{\s{N}_i}[k]}\}_{k=0}^{d-1},
\end{equation}
for $0\leq i\leq M-1$, where $\abs{\s{X}[k]}$, $\abs{\s{\hat{X}}[k]}$, and $\abs{\s{N}_i[k]}$ are the $k$-th Fourier component magnitudes of the template signal, the \EfN\ estimator, and the $i$-th noise observation, respectively. Similarly, $\phi_{\s{X}}[k]$, $\phi_{\s{\hat{X}}}[k]$, and $\phi_{\s{N}_i}[k]$ represent the corresponding $k$-th Fourier phases. Note that the random variables $\ppp{|\s{N}_i[k]|}_{k=0}^{d/2}$ and $\ppp{\phi_{\s{N}_i}[k]}_{k=0}^{d/2}$ are two independent sequences of i.i.d. random variables, such that, $|\s{N}_i[k]| \sim \s{Rayleigh} \left({\sigma^2}\right)$ has Rayleigh distribution, and the phase $\phi_{\s{N}_i}[k] \sim \s{Unif}[-\pi,\pi)$ is uniformly distributed over $[-\pi,\pi)$. 

With the definitions above, we can express the estimation process in the Fourier domain. Since a shift in real-space corresponds to a linear phase shift in the Fourier space, it follows that,
    \begin{align}
        \hat{\s{X}}[k] & =  \frac{1}{M} \sum_{i=0}^{M-1} \abs{\s{N}_i[k]} e^{j\phi_{\s{N}_i}[k]} e^{j\frac{2\pi k}{d}\s{\hat{R}}_i},
        \label{eqn:estimatorFourierRepresentation_pre}
    \end{align}
for $k=0,1,\ldots,d-1$, where $ \abs{\s{N}_i[k]} $ and $ \phi_{\s{N}_i}[k] $ are defined in \eqref{FourierSpace}. It is important to note that the expression above will converge to zero without the last term that captures the dependency in $\s{\hat{R}}_i$---the location of the maximum correlation. This term reflects the fundamental properties of the \EfN process and its dependency on the template signal, as well as the connections between the different spectral components.
We denote by $\mathbb{E} {\abs{\phi_{\s{\hat{X}}}[k] - \phi_{\s{X}}[k]}^2}$ the MSE of the Fourier phases of the $k$-th spectral component.

\paragraph{Assumptions.}
Throughout this paper, we assume that the template signal $x$ is normalized, i.e., $\norm{x}_2^2 = 1$, where $\norm{\cdot}_2$ is the Euclidean norm, and further assume that its Fourier transform in non-vanishing, except possibly at the DC (zero-frequency) component. The first assumption is used for convenience and does not alter (up to a normalization factor) our main results in Theorems \ref{thm:1} and \ref{thm:2}. The second assumption is essential for the theoretical analysis of the \EfN process and is expected to hold in many applications, including cryo-EM. A similar assumption is frequently taken in related work, e.g., \cite{bandeira2023estimation,perry2019sample,bendory2017bispectrum}.
It is worth noting that since the Fourier transform of $x$ is assumed to be non-vanishing, the maximizing shift $\hat{\s{R}}_i$ in~\eqref{eqn:OptShiftRealSpace} is almost surely unique. In addition, without loss of generality, we assume that the signal length $d$ is even.

\section{Cryo-EM and Empirical Demonstration} \label{sec:cryoEM}
Cryo-EM is a powerful tool of modern structural biology, offering advanced methods to visualize complex biological macromolecules with ever-increasing precision. One of its central advantages lies in its capability to resolve the structures of proteins that are hard to crystallize in traditional methods, especially in a near-physiological environment, see e.g.,~\cite{nogales2016development,singer2020computational}. This advantage enables researchers to delve into the dynamic behaviors of proteins and their complexes, shedding light on fundamental biological processes. 

Single-particle cryo-EM uses electron microscopy to reconstruct 3D structures from 2D tomographic projection images~\cite{bendory2020single}. Typically, the 3D reconstruction involves two main steps: detecting and extracting single particle images using a particle picking algorithm,~\cite{scheres2015semi,heimowitz2018apple,bepler2019positive,eldar2024object}, and then reconstructing the 3D density map~\cite{scheres2012relion,punjani2017cryosparc}. Most detection algorithms use template-matching techniques, which can introduce bias if improper templates are chosen, especially in low signal-to-noise ratio (SNR) conditions, which is the standard scenario in cryo-EM. 

\paragraph{The \EfN controversy.} 
A publication of the 3D structure of an HIV molecule in PNAS in 2013~\cite{mao2013molecular} initiated a fundamental controversy about validation techniques within the cryo-EM community, published as four follow-up PNAS publications~\cite{henderson2013avoiding, van2013finding, subramaniam2013structure, mao2013reply}. 
The \EfN pitfall played a central role in this discussion.
 The primary question of the discussion was whether the collected datasets contained informative biological data or merely pure noise images. The core of the debate emphasized the importance of exercising caution and implementing cross-validation techniques when fitting data to a predefined model.
 This precautionary approach aims to mitigate the risk of erroneous fittings, which could ultimately lead to inaccuracies in 3D density map reconstruction. 
 Model bias is still a fundamental problem in cryo-EM, as highlighted by an ongoing debate concerning validation tools, see for example,~\cite{stewart2004noise, shatsky2009method, henderson2012outcome,cohen2013high, cossio2020need, heymann2015validation,kleywegt2024community,sorzano2022image}.

\paragraph{Empirical demonstration.}
As introduced in the previous section, the \EfN phenomenon depends on several key parameters: (1) the number of observations, denoted by $M$; (2) the dimension of the signal, denoted as $d$ (for example, the number of pixels in Einstein's image); and (3) the statistical properties of the template signal, and in particular its power spectral density (PSD). To demonstrate the dependency on these parameters and provide insight into our main results, Figures \ref{fig:2} and \ref{fig:3} show the convergence of the \EfN estimator. 
Both figures were generated according to the procedure outlined in Section \ref{sec:problem_formulation}. When referring to Monte Carlo trials, it means that the \EfN estimator procedure, as specified in \eqref{eqn:efnEstimatorRealSpace}, was executed multiple times (the number of Monte-Carlo trials), each trial with fresh data.

Figure \ref{fig:2} shows the \EfN estimator as a function of $M$. Figure~\ref{fig:2}(a) shows that as $M$ increases, the \EfN estimator becomes more structurally similar to the template Einstein image. Indeed, Figure~\ref{fig:2}(b) shows that as the number of observations $M$ grows, the MSE between the Fourier phases of the template image and the corresponding Fourier phases of \EfN estimator decreases. Figure~\ref{fig:2}(c) highlights that the convergence rate is proportional to $1/M$, with a faster convergence rate for stronger spectral components. 

Figure \ref{fig:3} illustrates the impact of the template signal's PSD on the cross-correlation between the template signal and the \EfN estimator. Notably, a flatter PSD (i.e., a faster decay of the auto-correlation) leads to a higher correlation between the template and the estimator signals. These empirical results are proved theoretically in Theorems~\ref{thm:1} and \ref{thm:2}.

\paragraph{More applications.}
The \EfN phenomenon extends to various applications employing template matching, whether through a feature-based or direct template-based approach. For instance, template matching holds significance in computational anatomy, where it aids in discovering unknown diffeomorphism to align a template image with a target image \cite{christensen1996deformable}. Other areas include medical imaging processing \cite{aberneithy2007automatic}, manufacturing quality control \cite{aksoy2004industrial}, and navigation systems for mobile robots \cite{kyriacou2005vision}. 
This pitfall may also arise in the feature-based approach, which relies on extracting image features like shapes, textures, and colors to match a target image by neural networks and deep-learning classifiers \cite{zhang2018unreasonable, moscovich2022cross, talmi2017template, li2005fast}.

\begin{figure}[!t]
    \centering
    \includegraphics[width=0.9\linewidth]{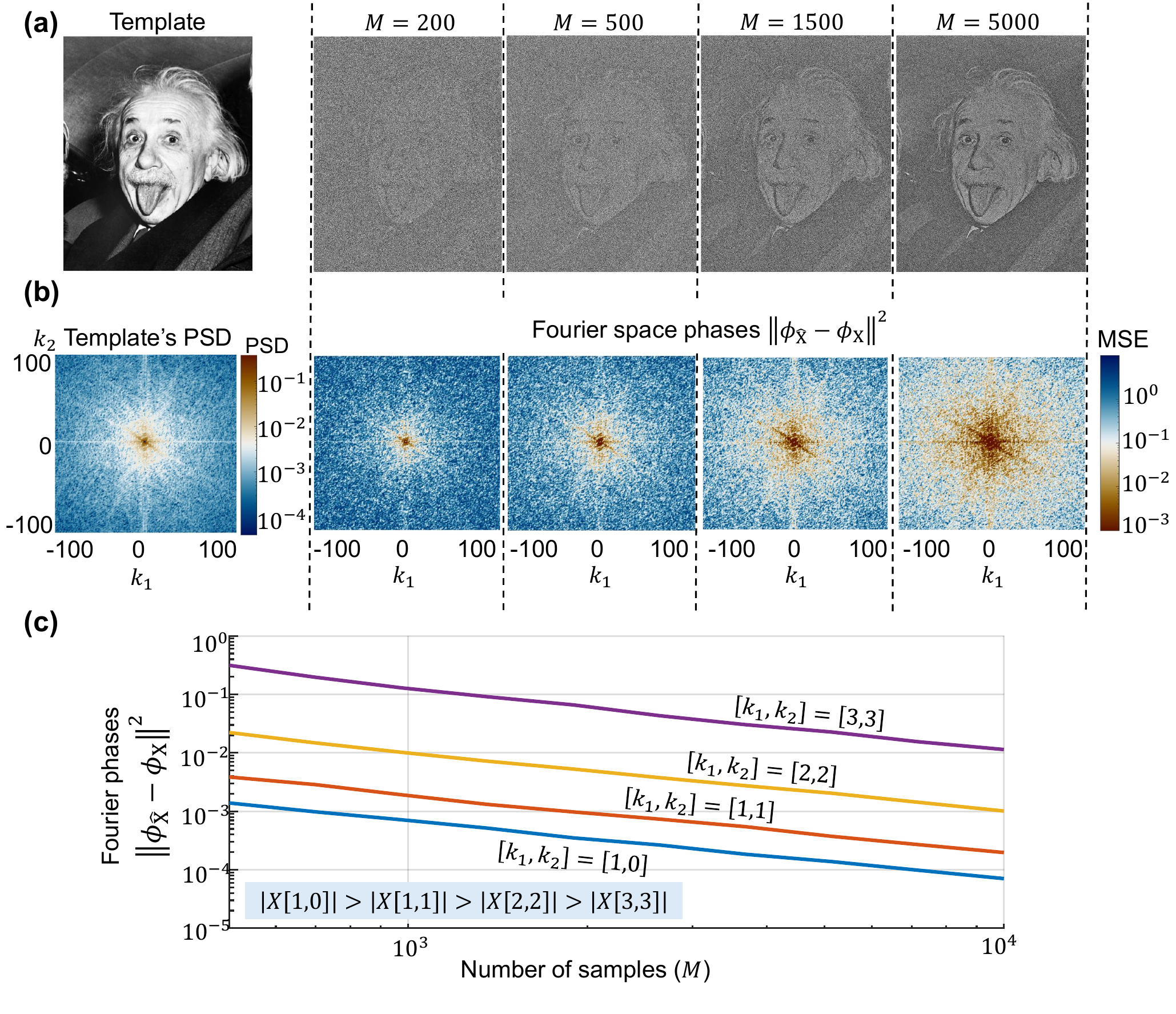}
    \caption{\textbf{The impact of the number of noise observations on the \EfN estimator.} The \EfN estimator is defined in real space by \eqref{eqn:efnEstimatorRealSpace} and in Fourier space by \eqref{eqn:estimatorFourierRepresentation_pre}. \textbf{(a)} The structural similarity between the \EfN estimator and the template image increases as a function of the number of noise observations ($M$). 
    \textbf{(b)} The mean-square-error (MSE) between the Fourier phases of the template image $\s{X}[k_1, k_2]$ and the \EfN estimator $\hat{\s{X}}[k_1, k_2]$ for $-100 \leq k_1, k_2 \leq 100$, where $k_1, k_2$ are the indices of the 2D DFT. The colors in the left panel in (b) represent the power spectral density (PSD) of the Einstein image, while the colors in the four right panels represent the MSE between the Fourier phases of the Einstein image and the \EfN estimator, for each spectral component,  with a varying number of observations ($M = 200, 500, 1500, 5000$). An increase in the number of observations leads to a lower MSE of the Fourier phases between the \EfN estimator and the template signal. A similar trend can be seen with respect to the strength of the spectral components, i.e., stronger spectral components lead to lower Fourier phases MSE.
    \textbf{(c)} The convergence rate of the MSE between the Fourier phases of the \EfN estimator and the Fourier phases of the template signal as a function of the number of observations across different frequencies. 
    The MSE decays as $1/M$.
    In addition, stronger spectral components lead to lower MSE. Figures \textbf{(b)} and \textbf{(c)} were generated through 200 Monte-Carlo trials of the \EfN process defined in \eqref{eqn:efnEstimatorRealSpace}.}
    \label{fig:2}
\end{figure}

\begin{figure}[t!]
    \centering
    \includegraphics[width=0.9\linewidth]{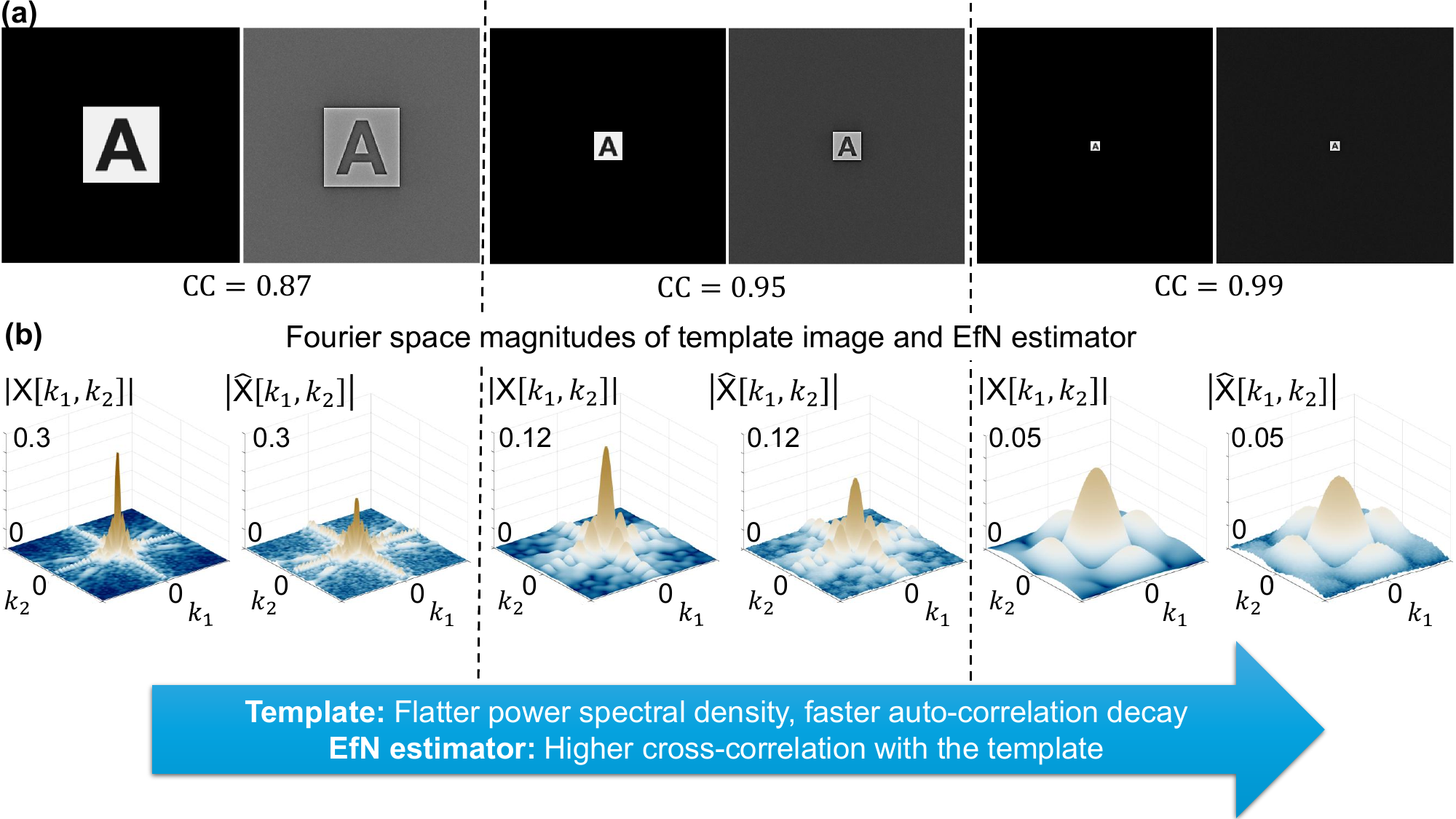}
    \caption{\textbf{The influence of the power-spectral-density (PSD) of the template signal on the correlation between the template and the \EfN estimator.}  \textbf{(a)} Three images of the letter A are shown, with an increasing zero-padding ratio. As the zero-padding ratio increases, the PSD flattens, and the cross-correlation (CC) between the template and the \EfN estimator increases. 
    This higher cross-correlation is evident in both the image background and the colors of the letter A. \textbf{(b)} Flatter PSDs lead to \EfN estimators whose Fourier magnitudes are closer to those of the template image. The \EfN estimators in these experiments were generated using $M=10^5$ observations. 
    }
    \label{fig:3}
\end{figure}

\paragraph{Previous work.}
The \EfN phenomenon has been investigated in earlier studies. In particular, it was shown that the ratios between the expected values of the Fourier coefficients of the \EfN estimator and those of the template are real-valued~\cite[Chapter 5]{Zhu_thesis}. In this work, we build upon and significantly extend these results. Specifically, we establish the convergence of the \EfN estimator to a non-vanishing signal, derive its convergence rate, analyze its behavior in the high-dimensional regime, and generalize the analysis to encompass a broader class of noise models beyond white Gaussian noise.

A closely related work is that of Wang et al.~\cite{wang2021quantification}, who conducted a rigorous statistical analysis of model bias in a different but complementary setting. 
They analyze the effects of selectively averaging only samples that exhibit the highest cross-correlation with a fixed reference signal (e.g., Einstein's image). This selection mechanism introduces a bias toward the reference, and their analysis reveals a phase transition in the resulting reconstruction, governed by the number of samples, the signal dimension, and the size of the selected subset. Notably, their results show that a structured image can emerge even when averaging purely noisy data. A related but statistically distinct selection-based mechanism is studied in~\cite{balanov2025structure} in the context of template matching for particle picking and extraction: candidate (pure-noise) observations are filtered by thresholding their cross-correlation with a bank of templates, and downstream averages inherit a strong template imprint.

In contrast, our work investigates the \EfN estimator in the absence of any selection step: all pure-noise observations are first aligned to a fixed template and then averaged. This difference in mechanism leads to different statistical behavior. Whereas template-based selection can yield averages that closely resemble the templates up to a global scale factor, we show that the non-selective \EfN procedure exhibits a more structured form of alignment-induced bias, manifested most prominently through phase locking: the Fourier phases of the estimator converge to those of the template, even though the resulting average is generally not identical to the template. Closely related alignment-induced artifacts and their connection to Fourier-phase behavior have also been discussed in the multireference alignment literature; see, e.g.,~\cite{shahverdi2024moment}.

\section{Main Results} \label{sec:main_results}

We begin by analyzing the regime where $M\to\infty$ and the dimension of the signal $d$ is fixed. In this setting, we show that the Fourier phases of the \EfN estimator converge almost surely to those of the underlying template signal, and we characterize the convergence rate. We also analyze the behavior of the Fourier magnitudes. Then, we turn to the high-dimensional regime, where $d\to\infty$. Under additional assumptions, we derive refined asymptotic guarantees for both the phases and magnitudes. Throughout, we assume that the template signal $x \in \mathbb{R}^d$ has a unit norm and that its spectrum is non-vanishing, i.e., $\s{X}[k] \neq 0$ for every $0 < k \leq d-1$, as discussed in the previous section.

\subsection{Finite-dimensional signal} 
We begin with the case where the template signal has a fixed dimension $d$, as captured in the following result, whose proof is provided in Appendix~\ref{thm:proofs1}.

\begin{thm} [Fourier phases convergence for finite-dimensional signal]\label{thm:1}
Fix $d\geq 2$ and assume that $\s{X}[k] \neq 0$, for all $0< k \leq d-1$.
\begin{enumerate}
  \item For any $0\leq k\leq d-1$, we have,
    \begin{align}
        \phi_{\s{\hat{X}}}[k] \xrightarrow[]{\s{a.s.}}  \phi_{{\s{X}}}[k],\label{eqn:FirstRes}
    \end{align}
    as $M\to\infty$. Furthermore, 
    \begin{align}
       \lim_{M\to\infty} \frac{\mathbb{E} |\phi_{\s{\hat{X}}}[k] - \phi_{\s{X}}[k]|^2}{1/M} = C_k,
        \label{eqn:asymptoticComnvergenceOfPhases}
    \end{align}
    for a finite constant $C_k<\infty$.
    \item For any $0\leq k\leq d-1$, we have,
        \begin{align}
            |\s{\hat{X}}[k]| \xrightarrow[]{\s{a.s.}} 
            \mathbb{E} \left[ \abs{\s{N}[k]} \cos\left(\frac{2\pi k}{d}\hat{\s{R}}_1 + \phi_{\s{N}}[k] - \phi_{\s{X}}[k]\right) \right] > 0,            \label{eqn:magnitudeConvergenceAsymptoticM}
        \end{align}
        as $M\to\infty$, where $\hat{\s{R}}_1$ is defined in \eqref{eqn:OptShiftRealSpace}.
\end{enumerate}
\end{thm}

Theorem~\ref{thm:1} captures two central properties. The first addresses the convergence of the \EfN estimator's phases to those of the template signal. In addition, the corresponding convergence rate in MSE is proportional to $1/M$. The second result captures the convergence of the \EfN estimator's magnitudes to the term given in the right-hand-side (r.h.s.) of~\eqref{eqn:magnitudeConvergenceAsymptoticM}, which is strictly greater than zero. Thus, the \EfN estimator converges to a non-vanishing signal. Interestingly, this term is not necessarily proportional to the magnitudes $\abs{\s{X}[k]}$ of the template signal and, thus, the \EfN estimator reproduces (asymptotically) only the phases of the template signal but not the magnitudes.

A central component of the proof of Theorem~\ref{thm:1} is the circulant structure inherent in the alignment of the noise, which arises from the cyclic shift operations. This symmetry implies that the covariance matrix of the noise-aligned sum is circulant, corresponding to a cyclo-stationary Gaussian process.
In particular, we apply the central limit theorem (CLT) and the strong law of large numbers (SLLN) for this setting, which yields 
\[
    \phi_{\s{\hat{X}}}[k] - \phi_{\s{X}}[k] \xrightarrow[]{\mathcal{D}} \arctan(\s{Q}_k),
\]
as $M \to \infty$, where $\s{Q}_k$ is a zero-mean Gaussian random variable with variance $\sigma_{\s{Q}}^2[k] = C_k / M$, and the constant $C_k$ admits a closed-form expression. By leveraging properties of cyclo-stationary Gaussian processes, which is justified by the circulant structure of the problem, we establish that $C_k < \infty$ for all $0 \leq k \leq d - 1$. This directly leads to the results stated in \eqref{eqn:FirstRes}--\eqref{eqn:asymptoticComnvergenceOfPhases}.

\subsection{High-dimensional regime}
We now turn to the high-dimensional setting where $d \to \infty$, taken after the limit $M \to \infty$. In this regime, we impose additional technical conditions on the template signal, formalized in Assumption~\ref{assump:1}. Intuitively, these conditions reflect the empirical phenomenon illustrated in Figure~\ref{fig:3}, where a flatter PSD, which corresponds to a more rapidly decaying autocorrelation function, results in an improved alignment between the template and the estimator.

More precisely, Assumption~\ref{assump:1} below requires control over the decay of both the autocorrelation function and the spectral magnitudes as functions of $d$. Specifically, for a length-$d$ signal $x\in\mathbb{R}^d$, we define the (circular) autocorrelation in the time domain by
\begin{align}
    \label{eqn:auto-correlation-def}
    \s{R}_{\s{X}\s{X}}[\tau] \triangleq \frac{1}{\sqrt{d}}\sum_{n=0}^{d-1}{x}[n]\;{x}[n+\tau \!\!\!\pmod d],
\end{align}
for $\tau\in\{0,1,\ldots,d-1\}$.
By the discrete Wiener-Khinchin theorem, this is equivalent to taking the inverse discrete Fourier transform of the PSD, namely, 
\begin{align}
    \s{R}_{\s{X}\s{X}}[\tau] = \mathcal{F}^{-1}\!\left\{\,|\s{X}|^2\,\right\}[\tau] = \frac{1}{\sqrt{d}}\sum_{k=0}^{d-1}|\s{X}[k]|^2\,e^{j\frac{2\pi}{d}k\tau},
\end{align}
where $\s{X}=\mathcal{F}\{\s{x}\}$ is the DFT of $x$. Assumption~\ref{assump:1} requires that the autocorrelation $\s{R}_{\s{X}\s{X}}[\tau]$ decay faster than $1/\log d$, and that the maximum magnitude among nonzero Fourier components $\max_{k\neq 0}|\s{X}[k]|$ decay faster than $1/\sqrt{\log d}$. In addition, we assume the DC component is vanishing, i.e., $|\s{X}[0]|=0$, to avoid degeneracies in alignment.

\begin{assum}[High-dimensional regularity of the template]
\label{assump:1}
Consider a sequence of template signals $\{x^{(d)}\}_{d \in \mathbb{N}}$ with $x^{(d)}\in\mathbb{R}^d$, and let $\s{X}^{(d)}=\mathcal{F}\{x^{(d)}\}$ denote their DFT.
Let $\s{R}^{(d)}_{\s{X}\s{X}}$ denote the autocorrelation of $x^{(d)}$, as defined in~\eqref{eqn:auto-correlation-def}.
When taking the limit $d\to\infty$, we assume the signals are normalized, namely $\|x^{(d)}\|_2 = 1$ for all $d \in \mathbb{N}$.
We say that the template sequence $\{x^{(d)}\}_{d \in \mathbb{N}}$ satisfies Assumption~\ref{assump:1} if the following hold:
\begin{enumerate}
    \item \emph{Autocorrelation decay.} The autocorrelation away from the zero lag satisfies
    \begin{align}
        \label{eq:assump_autocorr_decay} \lim_{d\to\infty}\left(\max_{1\le \tau \le d-1}\big|\s{R}^{(d)}_{\s{X}\s{X}}[\tau]\big|\right)\cdot \log d = 0.
    \end{align}
    \item \emph{Spectral magnitude decay.} The non-DC Fourier magnitudes satisfy
    \begin{align}
        \label{eq:assump_specmag_decay}  \lim_{d\to\infty}\left(\max_{1\le k \le d-1}\big|\s{X}^{(d)}[k]\big|\right)\cdot \sqrt{\log d} = 0.
    \end{align}
    \item \emph{Vanishing DC component.} The signal's DC component is zero, i.e., ${\abs{\s{X}^{(d)}[0]}} = 0$.
\end{enumerate}
\end{assum}

Although the conditions in Assumption~\ref{assump:1} may seem technical, they are essential for establishing Theorem~\ref{thm:2}, which relies on classical limit theorems for the maxima of stationary Gaussian processes, most notably, convergence to the Gumbel distribution~\cite{leadbetter2012extremes, berman1964limit, adler2009random, azais2009level}. Each part of the assumption plays a specific role: Part (1) ensures that the noise process lacks long-range dependencies, which corresponds to a sufficiently flat PSD; Part (2) guarantees that no individual Fourier component dominates the behavior of the \EfN estimator. The final condition, requiring the DC component to vanish (i.e., $|\s{X}^{(d)}[0]| = 0$), is not strictly necessary from an empirical standpoint but is introduced to streamline the theoretical analysis.

\begin{thm}[Fourier phases convergence for high-dimensional signal]\label{thm:2}
Assume that $\s{X}^{(d)}[k] \neq 0$, for all $0< k \leq d-1$, and that $x$ satisfies Assumption~\ref{assump:1}. Then,
\begin{enumerate}
\item For any $0\leq k\leq d-1$, we have,
     \begin{align}        \lim_{d\to\infty}\lim_{M\to\infty}\frac{\mathbb{E} |\phi_{\s{\hat{X}}^{(d)}}[k] - \phi_{\s{X}^{(d)}}[k]|^2}{1/(M\log d)} \frac{1}{1/(4\abs{\s{X}^{(d)}[k]}^2)}  = 1.        \label{eqn:phaseConvergeneRateForAsymptoticD}
    \end{align}
\item For any $0\leq k\leq d-1$, we have,
    \begin{align}
        \frac{1}{\sigma\sqrt{2\log d}}\frac{|\s{\hat{X}}^{(d)}[k]|}{\abs{\s{X}^{(d)}[k]}} \xrightarrow[]{\s{a.s.}} 1,  \label{eqn:magnitudeConvergenceAsymptoticMandD}
    \end{align}
    as $M,d\to\infty$. 
\end{enumerate}
\end{thm}

The proof of Theorem~\ref{thm:2} is presented in Appendix~\ref{thm:proofs2}. Based on Theorem~\ref{thm:2}, as $M,d\to\infty$, the convergence rate of the Fourier phases of the \EfN estimator is inversely proportional to the squared Fourier magnitude. In addition, unlike the fixed-$d$ result in Theorem~\ref{thm:1}, the high-dimensional regime reveals an explicit dependence on $d$ in the phase-error constant. Specifically, while Theorem~\ref{thm:1} states that for each fixed dimension $d$, \eqref{eqn:asymptoticComnvergenceOfPhases} holds with a finite constant $C_k$ (which depends on $d$ and on the template), Theorem~\ref{thm:2} makes this dependence explicit when $d\to\infty$: comparing \eqref{eqn:phaseConvergeneRateForAsymptoticD} with \eqref{eqn:asymptoticComnvergenceOfPhases} shows that in the high-dimensional limit the constant scales as 
\begin{align}
    \label{eq:Ckd_scaling_after_thm2} C_{k} = \frac{1}{4\,|\s{X}^{(d)}[k]|^2\,\log d}.
\end{align}
Intuitively, the $\log d$ factor arises from the alignment step: $\hat{\s{R}}_i$ is chosen by maximizing a stationary Gaussian correlation process over $d$ shifts, and the maximizer is governed by extreme-value statistics. In particular, the maximum of such a process grows on the scale $\sqrt{2\log d}$; see, e.g., \cite{leadbetter2012extremes, berman1964limit, adler2009random, azais2009level}, and the proof in Appendix~\ref{sec:proof-of-thm-2}.

Moreover, unlike Theorem~\ref{thm:1}  for a fixed $d$, Theorem~\ref{thm:2} also shows that the Fourier magnitudes of the \EfN estimator satisfy \eqref{eqn:magnitudeConvergenceAsymptoticMandD}, namely they recover the template magnitudes up to the known normalization factor $\sigma\sqrt{2\log d}$. Therefore, when $d\to\infty$ under Assumption~\ref{assump:1}, the normalized estimate $\hat{x}^{(d)}$ recovers the template signal, which in turn implies that the normalized cross-correlation between the template and the \EfN estimator approaches unity.

Empirically, we observe that Theorem~\ref{thm:2} provides accurate predictions of the convergence behavior when Assumption~\ref{assump:1} holds. As illustrated in Figure~\ref{fig:4}, the convergence rate is strongly influenced by the PSD of the template signal. In particular, Figure~\ref{fig:4}(b) shows that increasing the signal length and a flatter PSD lead to a stronger correlation between the \EfN estimator and the true template. Furthermore, Figure~\ref{fig:4}(c) demonstrates that the convergence of the Fourier phases of the \EfN estimator aligns closely with the theoretical predictions as the PSD becomes flatter. When the template violates Assumption~\ref{assump:1} (e.g., if its autocorrelation decays too slowly), the predicted convergence rates become less accurate, highlighting the importance of the assumption for the theorem’s formal guarantees.
However, even when the spectral decay is moderate, and the assumption is not strictly met, we find that the analytical convergence rates still rather closely match empirical observations (Figure \ref{fig:4}(c)). Notably, the key phenomenon that the convergence rate of the Fourier phases is inversely related to the magnitude of the corresponding spectral components remains robust beyond the regime where the theorem formally applies.

\begin{figure}[t!]
    \centering
    \includegraphics[width=1\linewidth]{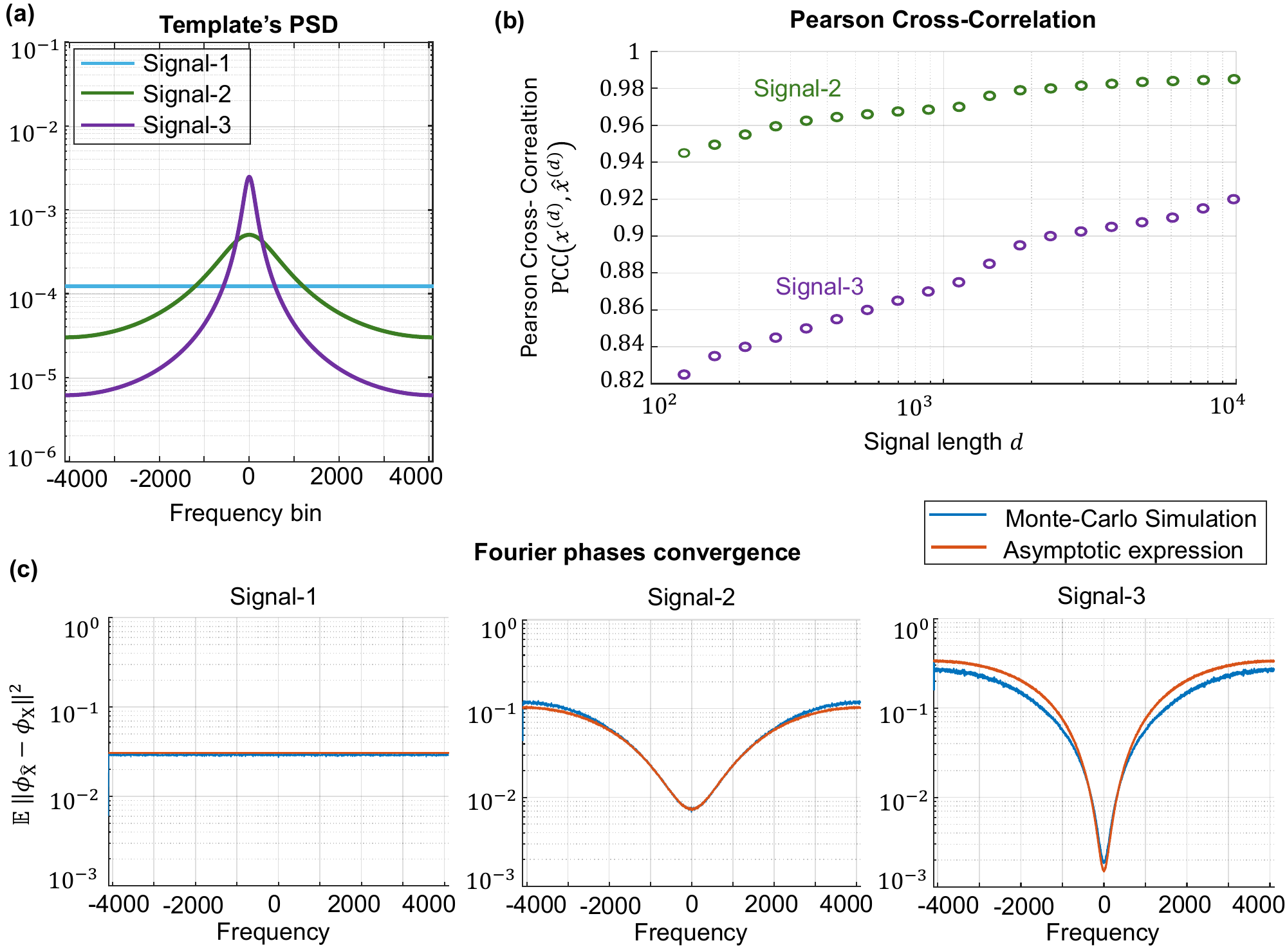}
    \caption{\textbf{Comparison between analytic expression and Monte-Carlo simulations for high-dimensional signals, $d$, and for signals with varying power spectral densities.}
    The analytic predictions for Fourier-phase convergence and Fourier-magnitude scaling are given by    \eqref{eqn:phaseConvergeneRateForAsymptoticD} and \eqref{eqn:magnitudeConvergenceAsymptoticMandD}, respectively.
    \textbf{(a)} Template PSDs for three template families at a representative dimension $d=8192$.
    For each dimension $d$, the template $x^{(d)}\in\mathbb{R}^d$ is generated directly at length $d$ as an exponentially decaying
    signal, $x^{(d)}_\ell[m]\propto \exp(-m/\alpha_\ell)$, $m=0,1,\ldots,d-1$, with decay parameters
    $\alpha_\ell\in\{0.02,\,2,\,10\}$ (Signals~1-3, respectively), followed by mean removal and normalization.
    \textbf{(b)} Monte-Carlo estimates of the Pearson cross-correlation $\mathrm{PCC}(x^{(d)}_\ell,\hat{x}^{(d)}_\ell)$ between the
    template $x^{(d)}_\ell$ and the EfN estimate $\hat{x}^{(d)}_\ell$ as a function of the signal length $d$ (with fixed sample
    size $M=10^4$). As $d$ increases, the correlation increases, particularly for templates with slower-decaying PSDs.
    \textbf{(c)} Per-frequency phase mean-squared error $\mathbb{E}|\phi_{\hat{\s{X}}^{(d)}_\ell}[k]-\phi_{\s{X}^{(d)}_\ell}[k]|^2$ at $d=8192$: Monte-Carlo estimates (blue) are compared with the asymptotic expression (red), i.e., the large-$(M,d)$ closed-form approximation predicted by \eqref{eqn:phaseConvergeneRateForAsymptoticD}. All Monte-Carlo curves are averaged over $2000$ independent trials.}
    \label{fig:4}
\end{figure}

\section{Extension to other noise statistics} \label{sec:extenstionToOtherNoise}
So far, we have analyzed the setting in which the noise is white Gaussian. In this section, we extend the analysis to a broader class of noise distributions. Specifically, we now assume that the observations $y_0, y_1, \ldots, y_{M-1} \in \mathbb{R}^d$ are i.i.d. samples drawn from an arbitrary distribution with zero mean and a fixed covariance matrix, namely,
\begin{align} \label{eqn:generalModel}
    \mathbb{E}[y_1] = \mathbf{0}, \quad \mathbb{E}[y_1 y_1^\top] = \Sigma,
\end{align}
where $\Sigma \succ 0$ is a positive-definite matrix with bounded operator norm, i.e., $\| \Sigma \| < \infty$. Notably, the entries of each sample $y_i$ are not required to be independent or identically distributed.

\subsection{Positive correlation}

In general, the Fourier phase convergence property established under the white Gaussian assumption does not hold for arbitrary noise distributions, as demonstrated empirically in Figures~\ref{fig:5} and~\ref{fig:6}. Nonetheless, we establish a positive correlation result between the EfN estimator and the underlying template signal.

\begin{proposition}[Positive correlation]
\label{prop:positiveCorrelation}
Let $d \geq 2$, and suppose the observations $\{y_i\}_{i=0}^{M-1}$ are drawn i.i.d. according to the model in \eqref{eqn:generalModel}. Let $x \in \mathbb{R}^d$ denote the template signal, and assume its discrete Fourier transform $\s{X}$ satisfies $\s{X}[k] \neq 0$ for all $1 \leq k \leq d - 1$. Let $\hat{x}$ be the EfN estimator computed from the observations $\{y_i\}$. Then, as $M \to \infty$, the following inequality holds almost surely,
\begin{align}
    \langle \hat{x}, x \rangle 
    \geq \max_{0 \leq r_1, r_2 < d-1} \frac{1}{2} \, \mathbb{E}\left[ \left| \left\langle y_1, \mathcal{T}_{r_1} x - \mathcal{T}_{r_2} x \right\rangle \right| \right] > 0. \label{eqn:positiveCorrelation}
\end{align}
\end{proposition}

The proof of Proposition~\ref{prop:positiveCorrelation} is provided in Appendix~\ref{sec:proofOfPositiveCorrelation}. This result implies that the EfN estimator is positively correlated with the true template signal. Although this is a weaker guarantee than the Fourier phase convergence obtained under Gaussian white noise, it still ensures that the estimator retains meaningful structural information from the template.

\subsection{High-dimensional i.i.d. noise}
Our next result demonstrates that the Fourier phase convergence established in Theorem~\ref{thm:1} for Gaussian white noise extends to a broader class of noise models in the high-dimensional regime. To this end, we impose an additional assumption that the entries of each observation vector $y_i \in \mathbb{R}^d$ are i.i.d.
Namely, the covariance matrix $\Sigma$ is diagonal.

\begin{thm}[High-dimensional i.i.d. noise] 
\label{thm:highDimentionalNoiseExtention}
Let $\{y_i\}_{i=0}^{M-1}$ be i.i.d. observations drawn according to the model in \eqref{eqn:generalModel}, and assume further that the entries of each $y_i \in \mathbb{R}^d$ are i.i.d., with finite variance, and satisfy $\mathbb{E}[(y_i[\ell])^4] < \infty$, for all $\ell \in \{0, 1, \ldots, d-1\}$. Let $\hat{\s{X}}$ denote the discrete Fourier transform of the EfN estimator under this noise model. Assume that the Fourier coefficients of the template $x$ are non-vanishing, i.e., $\s{X}[k] \neq 0$ for all $k \in \mathbb{N}^+$. Then, for any fixed $k \in \mathbb{N}^+$, we have,
\begin{align}
    \phi_{\s{\hat{X}}}[k] - \phi_{\s{X}}[k] \xrightarrow[]{\s{a.s.}} 0, \label{eqn:FirstRes2}
\end{align}
as $M, d \to \infty$. Moreover,
\begin{align}
   \lim_{d\to\infty} \lim_{M\to\infty} \frac{\mathbb{E} \left[|\phi_{\s{\hat{X}}}[k] - \phi_{\s{X}}[k]|^2\right]}{1/M} = C_k, \label{eqn:asymptoticComnvergenceOfPhasesCirculant}
\end{align}
for some finite constant $C_k < \infty$. Finally, if $x$ satisfies Assumption~\ref{assump:1}, then,
\begin{align}
    \lim_{d\to\infty} \lim_{M\to\infty} \frac{\mathbb{E} \left[|\phi_{\s{\hat{X}}}[k] - \phi_{\s{X}}[k]|^2\right]}{1/(M \log d)} \cdot \frac{1}{1/(4|\s{X}[k]|^2)} = 1. \label{eqn:asymptoticComnvergenceOfPhases3}
\end{align}
\end{thm}

The proof of Theorem~\ref{thm:highDimentionalNoiseExtention} is given in Appendix~\ref{sec:proofOfHighDimetnionalNoiseExtention}. In essence, this result extends the Fourier phase convergence of Theorem~\ref{thm:1} to a broader class of noise distributions in the high-dimensional setting. The main idea of the proof is to apply the functional central limit theorem to the DFT coefficients \cite{peligrad2010central, cerovecki2017clt, brillinger2001time}. As $d \to \infty$, the Fourier components of the noise converge in distribution to those of a circulant Gaussian random process, owing to the i.i.d. structure of the entries in $y_i$. This asymptotic Gaussianity enables us to apply the same analytical framework developed for the white noise case to establish convergence of the Fourier phases. 

\paragraph{Empirical demonstration.} Figure~\ref{fig:5} provides empirical validation of Theorem~\ref{thm:highDimentionalNoiseExtention} in settings where the noise distribution is non-Gaussian. In particular, we consider $y_i \in \mathbb{R}^d$ with i.i.d. entries drawn from either the uniform or Poisson distribution. As the figure shows, when $d$ is relatively small, the Fourier phases fail to converge and instead plateau. However, as the dimension increases, phase convergence emerges at the predicted $1/M$ rate, aligning with our theoretical results.

\begin{figure}[t!]
    \centering
    \includegraphics[width=0.95\linewidth]{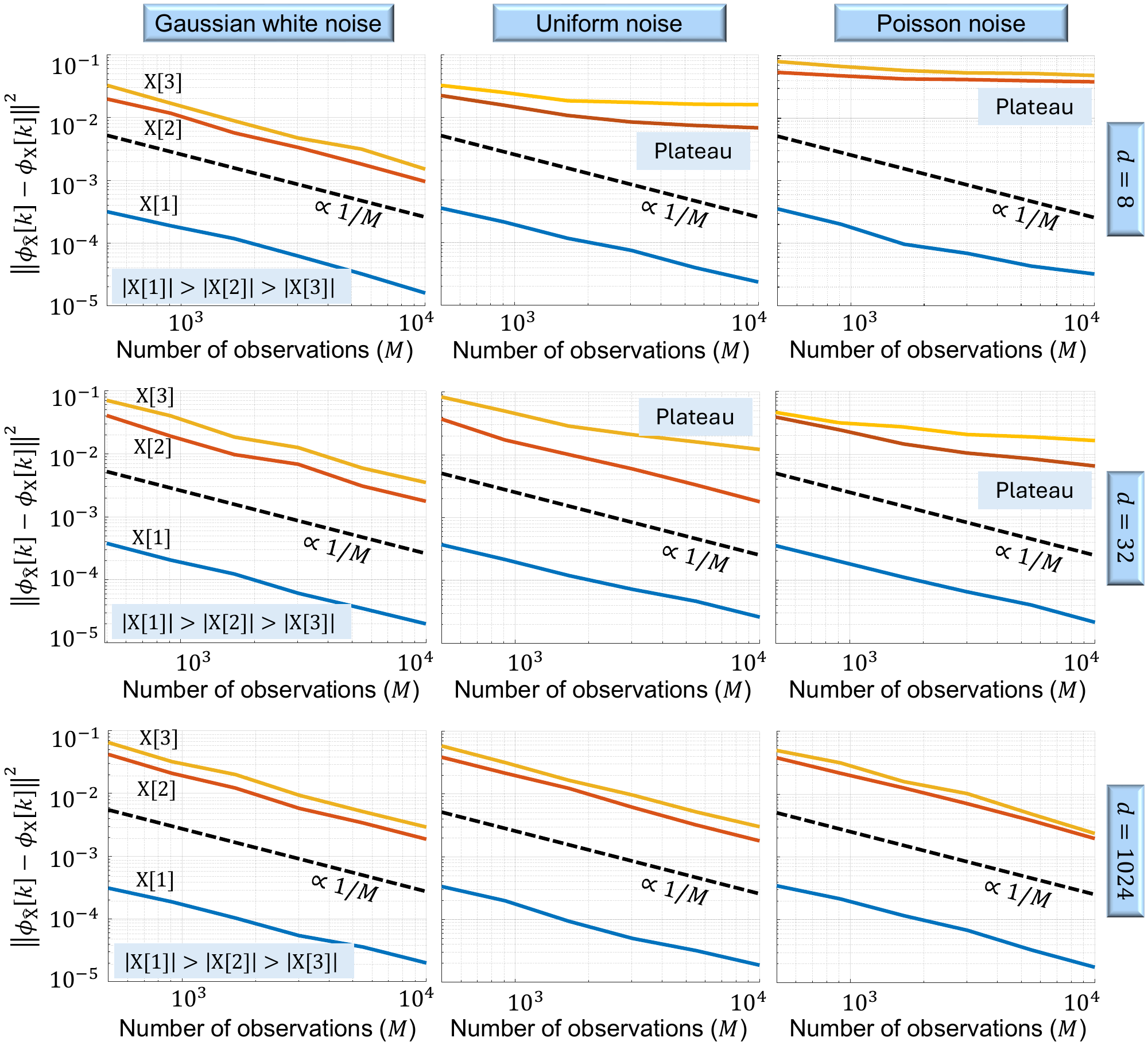}
    \caption{\textbf{The impact of noise statistics and signal dimension ($d$) on Fourier phase convergence.} Each panel displays the mean squared error (MSE) between the Fourier phases of the true template and those estimated by EfN, shown for three representative Fourier components. The dashed line represents the theoretical $1/M$ convergence rate. Columns correspond to different noise distributions: white Gaussian noise, i.i.d. noise drawn from a uniform distribution over the interval $[0,1]$, and i.i.d.  
    Poisson  noise with parameter $\lambda = 10$. Rows correspond to increasing signal dimensions: $d = 8$, $32$, and $1024$.
    For white Gaussian noise, the Fourier phases converge at the expected $1/M$ rate across all signal dimensions, in agreement with Theorem~\ref{thm:1}. In contrast, under uniform and Poisson noise, the MSE plateaus at low dimensions. However, increasing the signal dimension restores convergence, even under non-Gaussian noise, consistent with the high-dimensional regime described in Theorem~\ref{thm:highDimentionalNoiseExtention}. Notably, for $d = 1024$, all three noise models produce similar MSE values across the selected Fourier components, suggesting that their phase noise statistics become nearly indistinguishable. Each data point represents an average of 300 Monte Carlo trials.}
    \label{fig:5}
\end{figure}

\subsection{Circulant Gaussian process}

In this section, we consider the setting in which the noise exhibits correlations between entries. As previously noted, Fourier phase convergence does not generally hold under arbitrary noise models. However, we show that convergence is maintained when the noise follows a \textit{circulant Gaussian} distribution, a structured class of Gaussian noise characterized by rotational symmetry.

\begin{definition}[Symmetric circulant matrix] \label{def:circulantMatrix}
A matrix $\Sigma \in \mathbb{R}^{d \times d}$ is called \emph{circulant} if each row is a right cyclic shift of the previous one. That is, there exists a vector $c = (c_0, c_1, \dots, c_{d-1}) \in \mathbb{R}^d$ such that
\begin{align}
    \Sigma = \mathrm{circ}(c) = 
    \begin{bmatrix}
    c_0 & c_1 & c_2 & \dots & c_{d-1} \\
    c_{d-1} & c_0 & c_1 & \dots & c_{d-2} \\
    \vdots & \vdots & \vdots & \ddots & \vdots \\
    c_1 & c_2 & c_3 & \dots & c_0 \\
    \end{bmatrix}.
\end{align}
The matrix is said to be \emph{symmetric circulant} if $c_j = c_{d-j}$ for all $j = 1, \dots, d-1$.
\end{definition}

\begin{proposition}[Fourier phase convergence under circulant Gaussian noise]
\label{prop:circulantGauusianNoise}
Let $d \geq 2$ be fixed, and suppose the observations $\{y_i\}_{i=0}^{M-1}$ are i.i.d samples drawn from the multivariate normal distribution $\mathcal{N}(0, \Sigma)$, where $\Sigma$ is a symmetric circulant matrix as defined in Definition~\ref{def:circulantMatrix}. Assume further that the eigenvalues of $\Sigma$ are strictly positive, and that the template signal $x \in \mathbb{R}^d$ satisfies $\s{X}[k] \neq 0$ for all $1 \leq k \leq d-1$. Let $\hat{x}$ denote the EfN estimator under this noise model. Then, for each $0 \leq k \leq d-1$:
\begin{align}
    \phi_{\s{\hat{X}}}[k] \xrightarrow[]{\s{a.s.}} \phi_{\s{X}}[k], \label{eqn:FirstRes3}
\end{align}
as $M \to \infty$. Moreover,
\begin{align}
    \lim_{M \to \infty} \frac{\mathbb{E} \left[ |\phi_{\s{\hat{X}}}[k] - \phi_{\s{X}}[k]|^2 \right]}{1/M} = C_k, \label{eqn:asymptoticComnvergenceOfPhases2}
\end{align}
for some finite constant $C_k < \infty$.
\end{proposition}

The proof of Proposition~\ref{prop:circulantGauusianNoise} is given in Appendix~\ref{sec:circulantGauusianNoise}. In essence, this result serves as a generalization of Theorem~\ref{thm:1}, which considered the case of white Gaussian noise, to the broader setting of symmetric circulant Gaussian noise. Notably, white noise with covariance $\sigma^2 I_{d \times d}$ is a special case of circulant noise, making this extension a natural generalization. The critical insight here is that circulant covariance matrices remain diagonalizable in the Fourier basis, which preserves the independence of the DFT coefficients and enables phase convergence to proceed as in the white Gaussian case.

\paragraph{Empirical demonstration.} Figure~\ref{fig:6} presents an empirical comparison of the MSE of the Fourier phase estimates, as a function of the number of observations $M$, under three distinct noise models: (1) white Gaussian noise with covariance $\Sigma = \sigma^2 I_{d \times d}$; (2) Gaussian noise with a symmetric circulant covariance matrix, as defined in Definition~\ref{def:circulantMatrix}; and (3) Gaussian noise with a Toeplitz (but non-circulant) covariance matrix. As shown in the figure, both the i.i.d. and circulant models exhibit the expected $1/M$ decay in the phase MSE curve, though the constants $C_k$ differ, reflecting their distinct covariance structures. In contrast, under Toeplitz noise, the phase estimates do not converge: the MSE plateaus, and no $1/M$ scaling is observed. These results empirically confirm that the convergence of Fourier phases is tightly linked to the circulant structure of the noise covariance.

\begin{figure}[t!]
    \centering
    \includegraphics[width=0.95\linewidth]{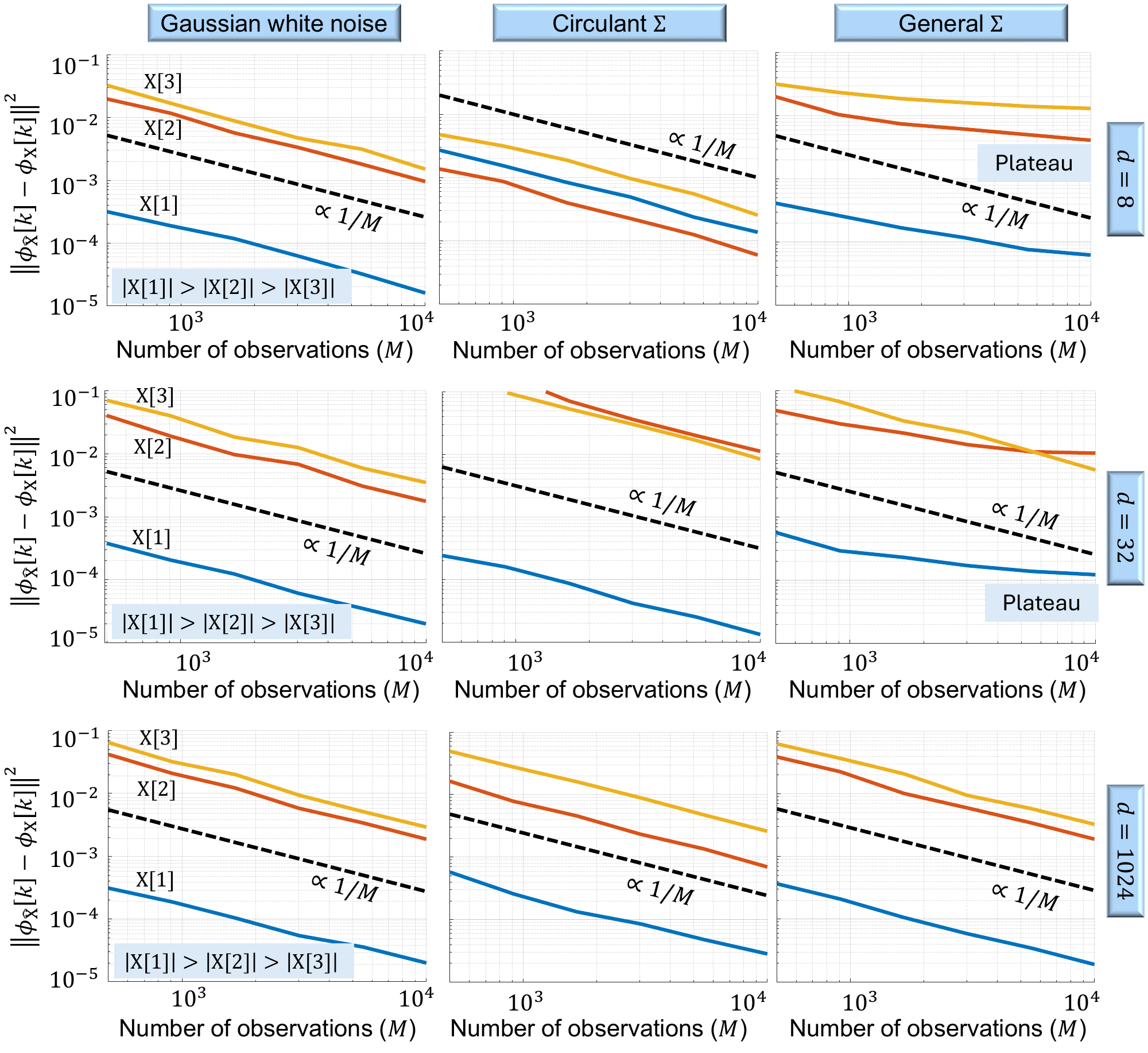}
    \caption{\textbf{The impact of the noise covariance structure and signal dimension ($d$) on Fourier phase convergence.} Each plot shows the mean squared error (MSE) between the Fourier phases of the ground-truth template and those estimated by EfN, evaluated across three spectral components. The dashed line indicates the theoretical $1/M$ convergence rate.
    Columns correspond to three types of noise: (1) white Gaussian noise with covariance $\Sigma = \sigma^2 I_{d \times d}$, (2) symmetric circulant covariance, and (3) a Toeplitz covariance matrix that is not circulant. Rows represent increasing signal dimensions: $d = 8$, $32$, and $1024$.
    Under white Gaussian noise, the Fourier phases converge at the expected $1/M$ rate, independent of the signal dimension (Theorem~\ref{thm:1}). A similar trend is observed when the noise has a circulant covariance structure: the same $1/M$ scaling holds, although the MSE is different compared to the white noise case.
    In contrast, for a Toeplitz covariance matrix that is not circulant, the MSE plateaus at small signal dimensions, indicating a failure of convergence. However, when the signal dimension increases to $d = 1024$, convergence at the $1/M$ rate is restored even under this more structured noise model.
    Each data point represents an average of 300 Monte Carlo trials.}
    \label{fig:6}
\end{figure}

\section{Discussion and outlook}\label{sec:outlook}

In this work, we have shown that the Fourier phases of the \EfN estimator converge to those of the template signal for an asymptotic number of observations. Since Fourier phases are crucial for perceiving image structure, the reconstructed image appears structurally similar to the template signal, even in cases where the estimator's spectral magnitudes differ from those of the template \cite{oppenheim1981importance, shechtman2015phase}. 
We have also shown that the Fourier phases of spectral components with stronger magnitudes converge faster, leading to faster structural similarity in the overall image perception. In addition, we have extended our analysis beyond white Gaussian noise, examining other noise models. We have shown that the \EfN estimator remains positively correlated with the template for arbitrary noise settings, and we have analyzed the Fourier phases convergence properties for high-dimensional i.i.d. noise (which is not necessarily Gaussian) and circulant Gaussian noise.

\subsection{Extensions and implications}
We anticipate that the findings of this paper will be beneficial in various fields. For example, the paper sheds light on a fundamental pitfall in template matching techniques, which may lead  
engineers and statisticians to misleading results. In addition, physicists and biologists working with data sets of low SNRs will benefit from understanding limitations and potential biases introduced by template matching techniques. More generally, this work provides a cautionary framework for the broader scientific community, highlighting the importance of exercising care when interpreting noisy observations.

\paragraph{Extension to higher dimensions.}
While this paper focuses on one-dimensional signals, the analysis can be readily extended to higher dimensions. This extension involves replacing the one-dimensional DFT with its $N$-dimensional counterpart. The symmetry properties established in Theorem \ref{thm:1}, including the results in Propositions \ref{prop:1} and \ref{prop:2}, remain valid. For the high-dimensional case of Theorem \ref{thm:2}, the conditions on the PSD adjust for the $N$-dimensional case. Specifically, the auto-correlation decay rate of the multidimensional array should be faster than $1/{\log d}$ in each dimension. 

\paragraph{Implications to cryo-EM.}
The findings of Theorem \ref{thm:2} have practical implications for cryo-EM. Typically, protein spectra exhibit rapid decay at low frequencies (known as the Guinier plot) and remain relatively constant at high frequencies, a behavior characterized by Wilson in~\cite{yu1942determination} and known as Wilson statistics. Wilson statistics is used to sharpen 3-D structures~\cite{singer2021wilson}.
To mitigate the risk of model bias, we suggest using templates with reduced high frequencies, recommending filtered, smooth templates. This insight may also relate to or support the common practice of initializing the expectation-maximization (EM) algorithm for 3-D refinement with a smooth 3-D volume. Each iteration of the EM algorithm effectively applies a version of template matching multiple times, although projection images typically contain actual signal rather than pure noise, as in the \EfN case. 

The key message for the cryo-EM community is that, regardless of the specific setting, one should not rely solely on the raw alignment average when working with low-SNR data. Instead, robust validation practices, such as cross-validation, independent reconstructions, and other consistency checks, are essential to guard against artifact-driven effects like the \EfN phenomenon. In this context, we mention a recent work suggesting that processing data in smaller mini-batches can help reduce the risk of EfN, offering a practical approach to mitigating model bias in such settings~\cite{balanov2026expectation}.  

\paragraph{Noise statistics in cryo-EM.}
The results in Section~\ref{sec:extenstionToOtherNoise} are particularly relevant to the noise characteristics commonly encountered in cryo-EM. While cryo-EM noise is often modeled as Poisson in nature, the standard practical assumption is that it follows a Gaussian distribution with a decaying power spectrum. These properties align well with the broader class of noise models considered in our analysis. Consequently, the conclusions of Theorem~\ref{thm:highDimentionalNoiseExtention} can be extended to the cryo-EM setting, and we expect similar asymptotic phase convergence behavior to hold.

\subsection{Future work}
Here, we list open questions and directions for future work.

\paragraph{Extension to non-cyclic group actions.} 
A natural direction for future work is to extend the \EfN analysis beyond the simplified setting of cyclic translations, as defined in \eqref{eqn:calQModel}, to more general group actions, particularly those arising in practical applications such as cryo-EM. In this context, the relevant transformations are elements of the rotation group $\s{SO}(3)$, and the postulated observations are two-dimensional projections of a three-dimensional structure rather than simple translations of a one-dimensional signal. Recent empirical evidence suggests that template-induced bias can persist in such non-abelian settings as well, including under $\s{SO}(3)$ group actions~\cite{xu2025bayesian}.

However, extending the theoretical analysis to non-abelian groups presents more substantial challenges. In particular, the property of circular Gaussian statistics, which underpins the \EfN analysis for cyclic groups, does not naturally extend to the non-abelian setting. Preliminary simulations for the non-abelian dihedral group (not shown here) indicate that the convergence of the \EfN estimator's Fourier phases observed in the abelian setting does not directly carry over. We nonetheless expect that alignment over a richer group may still induce systematic structural biases, potentially governed by the group’s representation-theoretic decomposition. This broader perspective is consistent with recent observations of confirmation bias effects even in unstructured latent-variable models: for Gaussian mixture models with pure-noise data, a single iteration of population $k$-means or EM initialized at hypothesized centroids produces estimates that remain positively correlated with the initialization~\cite{balanov2025confirmation}, suggesting that algorithmic bias can persist beyond group-alignment settings, albeit often in weaker forms than the phase-locking phenomenon characterized here.

\paragraph{Hard assignment algorithms and the EM algorithm.} 
One promising avenue for future research involves examining hard-assignment algorithms. These algorithms iteratively refine estimates of an underlying signal from noisy observations, where the signal is obscured by high noise (unlike the pure noise scenario in EfN). The process begins by aligning observations with a template signal in the initial iteration and averaging them to improve the template for subsequent iterations. A central objective is to understand and characterize the model bias introduced throughout this iterative process, specifically, how the final output depends on the initial template.
Notably, the results presented in this work can be interpreted as describing a single iteration of a hard-assignment algorithm in the limit as the SNR approaches zero.

Another important direction is investigating the EM algorithm, a cornerstone of cryo-EM algorithms~\cite{scheres2012relion, punjani2017cryosparc}. EM maximizes the likelihood function of models incorporating nuisance parameters~\cite{dempster1977maximum}, a topic of significant recent interest~\cite{daskalakis2017ten, xu2016global}. Unlike hard-assignment algorithms, EM operates iteratively as a soft assignment algorithm, assigning probabilities to various possibilities and computing a weighted average rather than selecting a single optimal alignment per observation. Further exploration of EM could provide deeper insights into iterative methodologies in cryo-EM and their associated model biases.

\paragraph{Extension to the non-i.i.d. case.} 
While Theorem~\ref{thm:highDimentionalNoiseExtention} assumes that the noise entries within each observation vector $y_i$ are independent and identically distributed, an important direction for future research is to extend these results to more general noise settings. Specifically, the analysis could be broadened to cover cases where the noise entries are independent but not identically distributed, provided that their variances remain uniformly bounded and a Lindeberg-type condition is fulfilled \cite{durrett2019probability}. Moreover, the framework may apply to noise that exhibits certain weak dependence structures, such as mixing conditions, allowing the use of functional central limit theorems and ensuring asymptotic Gaussianity of the Fourier components~\cite{peligrad2010central, cerovecki2017clt, brillinger2001time}.

\paragraph{Asymptotic regimes.}  
In this work, we analyzed two asymptotic regimes: (1) $M \to \infty$ with fixed $d$ (Theorem \ref{thm:1}), and (2) $M \to \infty$ followed by $d \to \infty$ (Theorem \ref{thm:2}). These regimes capture distinct theoretical and practical scenarios. Our approach relies on classical probabilistic tools in the first limit ($M \to \infty$), such as SLLN and CLT, and results from the theory of Gaussian extremes (e.g., convergence to the Gumbel distribution) in the second ($d \to \infty$).

Other challenging asymptotic regimes merit further investigation. In particular, it is of interest to understand the behavior in the joint high-dimensional regime where both $M, d \to \infty$ with a fixed ratio, i.e., $\frac{d}{M} \to c \in (0,\infty)$. This regime, common in modern high-dimensional statistics, differs from the sequential limits we analyze. More broadly, other asymptotic behaviors of $(M,d)$ are possible. When both $M = M_n$ and $d = d_n$ vary according to general sequences, a variety of additional regimes may arise, each potentially requiring different analytical techniques. Typically, in such settings, classical limit theorems may no longer apply directly, and new challenges arise, such as subtle phase transitions in statistical behavior and the breakdown of averaging effects when $d$ and $M$ grow at comparable rates. Addressing these phenomena typically requires more advanced tools from high-dimensional probability. We view the analysis of further asymptotic settings as a valuable direction for future research.

\paragraph{Statistical inference.}
While the present work establishes the asymptotic consistency of the \EfN estimator's Fourier phases, an important direction for future research is to investigate their behavior in the finite-sample regime. In particular, developing tools for statistical inference, such as confidence intervals or non-asymptotic error bounds, would enhance the practical utility of the analysis. Addressing these questions may require the use of sharper probabilistic techniques beyond classical limit theorems, such as Berry–Esseen-type results, concentration inequalities, or non-asymptotic deviation bounds tailored to the specific structure of the problem.

\section*{Acknowledgment}
T.B. is supported in part by BSF under Grant 2020159, in part by NSF-BSF under Grant 2024791, and in part by ISF under Grant 1924/21.
W.H. is supported by ISF under Grant 1734/21. 

\bibliographystyle{plain}

\begin{appendices}

{\centering{\section*{Appendix}}}

\paragraph{Appendix organization.} 
Appendix~\ref{thm:preProof} provides general preliminaries used throughout the paper, including notation and common technical tools. Appendix~\ref{sec:auxForTheorem1} presents the auxiliary lemmas required for Theorem~\ref{thm:1}, whose full proof appears in Appendix~\ref{thm:proofs1}. For Theorem~\ref{thm:2}, the necessary supporting results are given in Appendices~\ref{sec:preliminariesToTheorem2}, with the proof provided in Appendix~\ref{sec:proof-of-thm-2}. Appendix~\ref{sec:proofOfPositiveCorrelation} contains the proof of Proposition~\ref{prop:positiveCorrelation}, establishing the positive correlation property. Appendix~\ref{sec:proofOfHighDimetnionalNoiseExtention} proves Theorem~\ref{thm:highDimentionalNoiseExtention}, which extends the results to high-dimensional settings with i.i.d. noise that is not necessarily Gaussian. Finally, Appendix~\ref{sec:circulantGauusianNoise} provides the proof of Proposition~\ref{prop:circulantGauusianNoise}, addressing the case of structured noise with a circulant Gaussian covariance.

\section{Preliminaries}\label{thm:preProof}
Before we delve into the proofs, we fix notations and definitions and prove auxiliary results that will be used in the proofs. 

\subsection{Notations}
Recall the definitions of the Fourier transforms of $x$ and $n_i$ from~\eqref{FourierSpace}, and recall that the signal length $d$ is assumed to be even. Note that since $x$ and $n_i$ are real-valued, their Fourier coefficients satisfy the conjugate-symmetry relation: 
\begin{align}
    \s{X}[k]=\overline{\s{X}[d-k]},\quad 
    \s{N}_i[k]=\overline{\s{N}_i[d-k]}. \label{eqn:conjRelation}
\end{align}
In particular, $|\s{N}_i[k]| = |\s{N}_i[d-k]|$ and $\phi_{\s{N}_i}[k] = -\phi_{\s{N}_i}[d-k]$, which implies that only the first $d/2 + 1$ components of $\s{N}[k]$ are statistically independent. 

The definition of the maximal correlation in~\eqref{eqn:OptShiftRealSpace} can be represented in the Fourier domain as follows,
  \begin{align}
    \s{\hat{R}}_i &\triangleq \underset{0 \leq r \leq d-1}{\argmax} {\, \langle{n_i}, \mathcal{T}_r x\rangle}
    \\ & = \underset{0\leq r\leq d-1}{\argmax} \ { \langle \mathcal{F}\ppp{n_i},  {\mathcal{F}\ppp{\mathcal{T}_r x} }\rangle}
    \\ & = \underset{0\leq r\leq d-1}{\argmax} \sum_{k=0}^{d-1} {\abs{\s{X}[k]} \abs{\s{N}_i[k]} \, \cos \left( \frac{2\pi kr}{d} +  \phi_{\s{N}_i}[k] - \phi_{\s{X}}[k] \right)}.\label{eqn:OptShiftFourier}
\end{align}
To simplify notation, we define 
 \begin{align}
    {\s{S}_{i}[r]} &\triangleq \sum_{k=0}^{d-1} {\abs{\s{X}[k]} \abs{\s{N}_i[k]} \, \cos \left( \frac{2\pi kr}{d} +  \phi_{\s{N}_i}[k] - \phi_{\s{X}}[k] \right)} , \label{equ:maxGaussDef}
\end{align}
for $0\leq r\leq d-1$, and therefore, $\s{\hat{R}}_i = \argmax_{0 \leq r \leq {d-1}}\s{S}_{i}[r]$. We note that for any $0\leq i\leq M-1$, the random vector $\s{S}_i \triangleq (\s{S}_i[0],\s{S}_i[1],\ldots,\s{S}_i[d-1])^T$ is Gaussian distributed, with zero mean vector, and a circulant covariance matrix; therefore, it is a cyclo-stationary random process.

Our goal is to investigate the phase and magnitude of the estimator $\hat{\s{X}}$ in \eqref{eqn:estimatorFourierRepresentation_pre}. Simple manipulations reveal that, for any $0\leq k\leq d-1$, the estimator's phases are given by,
\begin{align}
    \phi_{\s{\hat{X}}}[k] = \phi_{\s{X}}[k] + \arctan \left( \frac{\sum_{i=0}^{M-1}\abs{\s{N}_i[k]} \sin \left( \phi_{e,i}[k] \right)}{\sum_{i=0}^{M-1}\abs{\s{N}_i[k]} \cos \left( \phi_{e,i}[k] \right)} \right ),\label{eqn:estimatorPhase}
\end{align}
where we define,
\begin{align}
    \phi_{e,i}[k] \triangleq \frac{2\pi k\s{\hat{R}}_i}{d} +  \phi_{\s{N}_i}[k] - \phi_{\s{X}}[k],
    \label{eqn:phaseDifferenceTerm}
\end{align}
and
\begin{align}
   |\s{\hat{X}}[k]|  = 
     \frac{1}{M} \abs{ \sum_{i=0}^{M-1}  \abs{\s{N}_i[k]} e^{j\phi_{e,i}[k]}}.
    \label{eqn:estimatorMagnitude}
\end{align}

\subsection{The convergence of the Einstein from Noise estimator} \label{sec:convergenceOfEfNestimator}
Recall the definition of $\phi_{e,i}[k]$ in \eqref{eqn:phaseDifferenceTerm}. Then, following \eqref{eqn:estimatorFourierRepresentation_pre}, and simple algebraic manipulation,
    \begin{align}
        \hat{\s{X}}[k] & =  \frac{1}{M} \sum_{i=0}^{M-1} \abs{\s{N}_i[k]} e^{j\phi_{\s{N}_i}[k]} e^{j\frac{2\pi k}{d}\s{\hat{R}}_i} 
        \\ & = \frac{e^{j\phi_{\s{X}}[k]}}{M} \sum_{i=0}^{M-1} \abs{\s{N}_i[k]} e^{j\phi_{\s{N}_i}[k]} e^{j\frac{2\pi k}{d}\s{\hat{R}}_i} e^{-j\phi_{\s{X}}[k]}
        \\ & = \frac{e^{j\phi_{\s{X}}[k]}}{M} \sum_{i=0}^{M-1} \abs{\s{N}_i[k]} e^{j\phi_{e,i}[k]}. \label{eqn:A11}
    \end{align}
By applying the strong law of large numbers (SLLN) on the right-hand-side of \eqref{eqn:A11}, for $M \to \infty$, we have,
\begin{align}
   \s{\hat{X}}[k] e^{-j\phi_{\s{X}}[k]}  = & \frac{1}{M} { \sum_{i=0}^{M-1}  \abs{\s{N}_i[k]} e^{j\phi_{e,i}[k]}} 
   \\ & \xrightarrow[]{\s{a.s.}} 
    \mathbb{E} \left[ \abs{\s{N}_1[k]} \cos\p{\phi_{e,1}[k]} \right] + j \mathbb{E} \left[ \abs{\s{N}_1[k]} \sin\p{\phi_{e,1}[k]} \right],\label{eqn:strongLLN}
\end{align}
where we have used the fact that the sequences of random variables $\{\abs{\s{N}_i[k]} \sin \left( \phi_{e,i}[k] \right)\}_{i=0}^{M-1}$ and
$\{\abs{\s{N}_i[k]} \cos \left( \phi_{e,i}[k] \right)\}_{i=0}^{M-1}$ are i.i.d. with finite mean and variances. 

We denote for every $0 \leq k \leq d-1$:
\begin{align}
    \mu_{\s{A},k} &\triangleq \mathbb{E}\left[ \abs{\s{N}_1[k]} \sin (\phi_{e,1}[k]) \right],  \label{eqn:muA}
    \\
    \mu_{\s{B},k} &\triangleq \mathbb{E}\left[ \abs{\s{N}_1[k]} \cos (\phi_{e,1}[k]) \right], \label{eqn:muB}
\end{align}
the imaginary and real parts of the right-hand-side of \eqref{eqn:strongLLN}, respectively. In addition, we denote:
\begin{align}
    \sigma_{\s{A},k}^2 \triangleq \s{Var}\p{\abs{\s{N}_1[k]} \sin (\phi_{e,1}[k])},
    \label{eqn:sigmaA}
    \\ \sigma_{\s{B},k}^2 \triangleq \s{Var}\p{\abs{\s{N}_1[k]} \cos (\phi_{e,1}[k])}.
    \label{eqn:sigmaB}    
\end{align}
In Theorem \ref{thm:1}, we prove that $\mu_{\s{A},k} = 0$ while $\mu_{\s{B},k} > 0$. Consequently, by \eqref{eqn:strongLLN}, as $M \to \infty$, the EfN estimator converges to a non-vanishing signal, and its Fourier phases converge those of the template (Einstein).

\subsection{Conditioning on the Fourier frequency noise component} 
Throughout the proofs, we condition the noise realization $\s{S}_i$ \eqref{equ:maxGaussDef} on the $k$-th Fourier coefficient $\s{S}_i\vert\s{N}_i[k]$, to capture the dependence of $\hat{\s{R}}_i$ on the noise component. Specifically, we prove the following:
\begin{lem} \label{lemma:conditioning}
    Recall the definition of $\s{S}_i$ \eqref{equ:maxGaussDef}. Then, for every $k \in \ppp{1,2, \dots, \frac{d}{2}-1, \frac{d}{2}+1, \dots d-1}$,
    \begin{align}
        \s{S}_i\vert\s{N}_i[k] \sim \calN (\s{\mu}_{k,i}, \s{\Sigma}_{k,i}), \label{eqn:conditionalGaussian}
    \end{align}
    where,
    \begin{align}
        \mu_{k,i}[r] &\triangleq \mathbb{E}\left[ \s{S}_i[r]\vert\s{N}_i[k] \right]= 2\abs{\s{X}[k]} \abs{\s{N}_i[k]} \cos \left( \frac{2\pi kr}{d} + \phi_{\s{N}_i}[k] - \phi_{\s{X}}[k] \right),
        \label{eqn: 14}
    \end{align}
    for $0\leq r \leq d-1$, and
    \begin{align}
        \s{\Sigma}_{k,i}[r,s] & \triangleq \mathbb{E}\left[ \left(\s{S}_i[r] - \mathbb{E}\s{S}_i[r] \right) \left( \s{S}_i[s] - \mathbb{E}\s{S}_i[s] \right)\vert\s{N}_i[k] \right] \nonumber\\ 
        & = \frac{\sigma^2}{2} \sum_{\ell = 0}^{d-1} |\widetilde{\s{X}}_k[\ell]|^2 \cos \left( \frac{2\pi \ell}{d}(r-s) \right), \label{eqn:covarainceMatrix}
    \end{align}
    for $0\leq r,s\leq d-1$, where $\widetilde{\s{X}}_k$ is defined by:
    \begin{align}
        \widetilde{\s{X}}_k[\ell] \triangleq  \begin{cases}
                0  & \s{if}  \ell = k, d-k, \\
             \s{X}[\ell]  & \s{if}  \ell = 0, d/2, \\
                \sqrt{2}\cdot\s{X}[\ell]  & \s{otherwise}.
          \end{cases}
           \label{eqn:XtildeDefenition}
    \end{align}
\end{lem}

\begin{remark}
    In Lemma \ref{lemma:conditioning}, and throughout this work, we condition on $\s{S}_i \vert \s{N}_i[k]$ for all $k \neq 0, d/2$. Since the signals $x$ and $n_i$ lie in $\mathbb{R}^d$, their Fourier phases satisfy $\phi_{\s{X}}[0] = 0$ and $\phi_{\s{X}}[d/2] = 0$. Therefore, we restrict our analysis to the convergence of the Fourier phases for $k \neq 0, d/2$, as the convergence at $k = 0$ and $k = d/2$ is trivial.
\end{remark}

Note that the conditional process $\s{S}_i\vert\s{N}_i[k]$ is Gaussian because it is given by a linear transform of i.i.d. Gaussian variables. Also, since its covariance matrix is circulant and depends only on the difference between the two indices, i.e., $\s{\Sigma}_{k,i}[r,s] = \sigma_{k,i}[\abs{r-s}]$, it is cycle-stationary with a cosine trend. The eigenvalues of this circulant matrix are given by the DFT of its first row, and thus its $\ell$-th eigenvalue equals $|\widetilde{\s{X}}_k[\ell]|^2$, for $0 \leq \ell\leq d-1$.

For simplicity of notation, whenever it is clear from the context, we will omit the dependence of the above quantities on the $i$-th observation and $k$-th frequency indices, and we will use $\mu[r]$ and $\s{\Sigma}[r,s]$, instead. 
Furthermore, for convenience, we assume that the template vector is normalized to unity, i.e. $\sum_{\ell=0}^{d-1} |\s{X}[\ell]|^2 = 1$.

\begin{proof}[Proof of Lemma~\ref{lemma:conditioning}]
    By definition of $\s{S}_i$ \eqref{equ:maxGaussDef}, we have for every $k \neq 0, d/2$,
    \begin{align}
        \nonumber \s{S}_i\pp{r}\vert \s{N}_i[k] = & 2\abs{\s{X}[k]} \abs{\s{N}_i[k]} \cos \left( \frac{2\pi kr}{d} + \phi_{\s{N}_i}[k] - \phi_{\s{X}}[k] \right) 
        \\ & +\sum_{\ell\neq k, d-k} {\abs{\s{X}[\ell]} \abs{\s{N}_i[\ell]} \, \cos \left( \frac{2\pi \ell r}{d} +  \phi_{\s{N}_i}[\ell] - \phi_{\s{X}}[\ell] \right)}, \label{eqn:A21}
    \end{align}
    where we have used the property of $\s{X}[k]=\overline{\s{X}[d-k]}, \  
    \s{N}_i[k]=\overline{\s{N}_i[d-k]}$, \eqref{eqn:conjRelation}. 
    Clearly, as $\mathbb{E} \pp{\s{N}_i \pp{\ell}} = 0$, for every $0 \leq \ell \leq d-1$, we have,
    \begin{align}
        \mathbb{E} \pp{{\abs{\s{X}[\ell]} \abs{\s{N}_i[\ell]} \, \cos \left( \frac{2\pi \ell r}{d} +  \phi_{\s{N}_i}[\ell] - \phi_{\s{X}}[\ell] \right)}} = 0, \label{eqn:A22}
    \end{align}
    for every $0 \leq \ell \leq d-1$. Combining \eqref{eqn:A21} and \eqref{eqn:A22} results,
    \begin{align}
        \mu_{k,i}[r] & = \mathbb{E}\left[ \s{S}_i[r]\vert\s{N}_i[k] \right]= 2\abs{\s{X}[k]} \abs{\s{N}_i[k]} \cos \left( \frac{2\pi kr}{d} + \phi_{\s{N}_i}[k] - \phi_{\s{X}}[k] \right),
    \end{align}
    proving the first result about the means.

    \paragraph{The covariance term.}
    In the following, we derive the covariance term,
    \begin{align}
        \s{\Sigma}_{k,i}[r,s] & \triangleq \mathbb{E}\left[ \left(\s{S}_i[r] - \mathbb{E}\s{S}_i[r] \right) \left( \s{S}_i[s] - \mathbb{E}\s{S}_i[s] \right)\vert\s{N}_i[k] \right].
    \end{align}
    Denote,
    \begin{align}
        \nonumber \rho_{k,i} \pp{r} &  \triangleq \s{S}_i[r] - \mathbb{E}\s{S}_i[r]  
        \\ & = \sum_{\ell\neq k, d-k} {\abs{\s{X}[\ell]} \abs{\s{N}_i[\ell]} \, \cos \left( \frac{2\pi \ell r}{d} +  \phi_{\s{N}_i}[\ell] - \phi_{\s{X}}[\ell] \right)}.
        \label{eqn:A24}
    \end{align}
    Denote the set 
    \begin{align}
        \mathcal{I} = \ppp{1,2, \dots k-1, k+1, \dots, d/2-1},
    \end{align}
    which defines the indices of the Fourier coefficients, excluding $\ppp{0, k, d/2}$.

    As the sequences $\ppp{|\s{N}_i[\ell]|}_{\ell=0}^{d/2}$ and $\ppp{\phi_{\s{N}_i}[\ell]}_{\ell=0}^{d/2}$ satisfy $\s{N}_i[\ell]=\overline{\s{N}_i[d-\ell]}$, as well as $\s{X}[\ell]=\overline{\s{X}[d-\ell]}$, we have,
    \begin{align}
        \nonumber  \rho_{k,i} \pp{r}= \sum_{\ell\neq k, d-k} & {\abs{\s{X}[\ell]} \abs{\s{N}_i[\ell]} \, \cos \left( \frac{2\pi \ell r}{d} +  \phi_{\s{N}_i}[\ell] - \phi_{\s{X}}[\ell] \right)} 
         = \\ = & \nonumber \sum_{\ell \in \ppp{0,d/2}} \abs{\s{X}[\ell]}\abs{\s{N}_i[\ell]} \, \cos \left( \frac{2\pi \ell r}{d} +  \phi_{\s{N}_i}[\ell] - \phi_{\s{X}}[\ell] \right) \\ & + 2 \cdot  \sum_{\ell \in \mathcal{I}} \abs{\s{X}[\ell]}\abs{\s{N}_i[\ell]} \, \cos \left( \frac{2\pi \ell r}{d} +  \phi_{\s{N}_i}[\ell] - \phi_{\s{X}}[\ell] \right),
        \label{eqn:A25}
    \end{align}
    where each one of the terms in the sum is independent. 
    
    Since the terms in the sum on the right-hand side of \eqref{eqn:A25} are independent, that is, $\mathbb{E} \pp{\s{N}_i\pp{\ell_1} \overline{\s{N}_i\pp{\ell_2}}} = \mathbb{E}\pp{\abs{\s{N}_i\pp{\ell_1}}^2} \delta_{\ell_1,\ell_2}$, it follows that,    
    \begin{align}
        \nonumber\s{\Sigma}_{k,i}&[r,s] =  \mathbb{E} \pp{\rho_{k,i} \pp{r}\rho_{k,i} \pp{s} \vert \s{N}_i[k] } 
        \\ & = \nonumber \mathbb{E} \pp{{\sum_{\ell \in \ppp{0, d/2}} {\abs{\s{X}[\ell]}^2 \abs{\s{N}_i[\ell]}^2 \, \cos \left( \frac{2\pi \ell r}{d} +  \phi_{\s{N}_i}[\ell] - \phi_{\s{X}}[\ell] \right) \cos \left( \frac{2\pi \ell s}{d} +  \phi_{\s{N}_i}[\ell] - \phi_{\s{X}}[\ell] \right)}}}
        \\ & + 4 \cdot \mathbb{E} \pp{{\sum_{\ell \in \mathcal{I}} {\abs{\s{X}[\ell]}^2 \abs{\s{N}_i[\ell]}^2 \, \cos \left( \frac{2\pi \ell r}{d} +  \phi_{\s{N}_i}[\ell] - \phi_{\s{X}}[\ell] \right) \cos \left( \frac{2\pi \ell s}{d} +  \phi_{\s{N}_i}[\ell] - \phi_{\s{X}}[\ell] \right)}}}. \label{eqn:A28}
    \end{align}
    The expectation value in \eqref{eqn:A28} is composed of the multiplications of cosines. Applying trigonometric identities, we obtain:
    \begin{align}
        \nonumber \cos & \left( \frac{2\pi \ell r}{d} + \phi_{\s{N}_i}[\ell] - \phi_{\s{X}}[\ell] \right)\cos \left( \frac{2\pi \ell s}{d} + \phi_{\s{N}_i}[\ell] - \phi_{\s{X}}[\ell] \right) 
        \\ & = \frac{1}{2}\cos \left( \frac{2\pi \ell (r-s)}{d} \right) + \frac{1}{2} \cos \left( \frac{2\pi \ell (r+s)}{d} + 2 \p{\phi_{\s{N}_i}[\ell] - \phi_{\s{X}}[\ell]} \right). \label{eqn:trigonometricIdentity}
    \end{align}
    For $\ell \in \left\{1,\ldots,\frac{d}{2}-1\right\}$, the DFT coefficients $N_i[\ell]$ are i.i.d.\ circular complex Gaussian (under white Gaussian noise in the time domain), hence independent across $\ell$, with $\phi_{N_i}[\ell] \sim \mathrm{Unif}\!\left[-\pi,\pi\right)$ independent of $\lvert N_i[\ell]\rvert$. Thus, 
    \begin{align}
        \nonumber \mathbb{E} & \pp{\abs{\s{N}_i[\ell]}^2 \cos \left( \frac{2\pi \ell r}{d} + \phi_{\s{N}_i}[\ell] - \phi_{\s{X}}[\ell] \right) \cos \left( \frac{2\pi \ell s}{d} + \phi_{\s{N}_i}[\ell] - \phi_{\s{X}}[\ell] \right)}
        \\ & = \frac{1}{2}\mathbb{E} \pp{\abs{\s{N}_i[\ell]}^2} \cos \left( \frac{2\pi \ell (r-s)}{d} \right) = \frac{\sigma^2}{2} \cos \left( \frac{2\pi \ell (r-s)}{d} \right). \label{eqn:A27}
    \end{align}
    Substituting \eqref{eqn:A27} into \eqref{eqn:A28} leads to,
    \begin{align}
         \nonumber \frac{2}{\sigma^2}\mathbb{E} & \pp{\rho_{k,i} \pp{r}\rho_{k,i} \pp{s} \vert \s{N}_i[k] } =  {{\sum_{\ell \in \ppp{0, d/2}} {\abs{\s{X}[\ell]}^2  \, \cos \left( \frac{2\pi \ell}{d}(r-s) \right)}}} 
        \\ & + 4 \cdot {{\sum_{\ell \in \mathcal{I}} {\abs{\s{X}[\ell]}^2  \, \cos \left( \frac{2\pi \ell}{d}(r-s) \right)}}}. \label{eqn:A30}
    \end{align}
    As for every $\ell \in \mathcal{I}$, $\abs{\s{X}[\ell]} = \abs{\s{X}[d-\ell]}$, we have,
    \begin{align}
        {{\sum_{\ell \in \mathcal{I}} {4\abs{\s{X}[\ell]}^2  \, \cos \left( \frac{2\pi \ell}{d}(r-s) \right)}}} = {{\sum_{\ell \neq \ppp{0,k,d/2,d-k}} {2\abs{\s{X}[\ell]}^2  \, \cos \left( \frac{2\pi \ell}{d}(r-s) \right)}}}. \label{eqn:A31}
    \end{align}
    Substituting \eqref{eqn:A31} into \eqref{eqn:A30}
    \begin{align}
        \nonumber \mathbb{E} & \pp{\rho_{k,i} \pp{r}\rho_{k,i} \pp{s} \vert \s{N}_i[k] } =  \frac{\sigma^2}{2} \sum_{\ell = 0}^{d-1} |\widetilde{\s{X}}_k[\ell]|^2 \cos \left( \frac{2\pi \ell}{d}(r-s) \right),
    \end{align}
    for $\widetilde{\s{X}}_k[\ell]$ defined in \eqref{eqn:XtildeDefenition}, which completes the proof.
\end{proof}

\subsection{Uniqueness of the maximizer}
\begin{lem}[Uniqueness of the maximizer]
\label{lem:unique_conditional_argmax_model}
Recall the definition of $\s{S}_i$ from~\eqref{equ:maxGaussDef}. Assume $d\ge 6$ is even. Fix
$k\in\{1,\ldots,\frac{d}{2}-1,\frac{d}{2}+1,\ldots,d-1\}$ and recall from Lemma~\ref{lemma:conditioning} that
\begin{align}
    \s{S}_i\vert \s{N}_i[k]\sim\mathcal{N}(\mu_{k,i},\Sigma_{k,i}),\qquad
    \Sigma_{k,i}[r,s]=\frac{\sigma^2}{2}\sum_{\ell=0}^{d-1}\big|\widetilde{\s{X}}_k[\ell]\big|^2 \cos\!\left(\frac{2\pi\ell}{d}(r-s)\right).
\end{align}
Assume moreover that the template spectrum is non-vanishing, i.e.\ $\s{X}[\ell]\neq 0$ for all $\ell\in\{0,\ldots,d-1\}$.
Then, for every $r\neq s$,
\begin{align}
    \label{eq:Var_diff_conditional_proved}
    \operatorname{Var}\!\Big( (\s{S}_i\vert \s{N}_i[k])_r - (\s{S}_i\vert \s{N}_i[k])_s \Big)
    &= \sigma^2\sum_{\ell=0}^{d-1}\big|\widetilde{\s{X}}_k[\ell]\big|^2
    \Big(1-\cos\!\big(\tfrac{2\pi\ell}{d}(r-s)\big)\Big)
     > 0 .
\end{align}
Consequently, the maximizer $\s{\hat{R}}_i=\argmax_{0\le r\le d-1} (\s{S}_i\vert \s{N}_i[k])_r$ is unique almost surely.
\end{lem}

\begin{proof}[Proof of Lemma~\ref{lem:unique_conditional_argmax_model}]
Fix $r\neq s$ and set $m\triangleq r-s\not\equiv 0\ (\mathrm{mod}\ d)$. Using the covariance formula and circulantness,
\begin{align}
    \operatorname{Var}\!\Big( (\s{S}_i\vert \s{N}_i[k])_r - (\s{S}_i\vert \s{N}_i[k])_s \Big)
    &= \Sigma_{k,i}[r,r]+\Sigma_{k,i}[s,s]-2\Sigma_{k,i}[r,s] \nonumber\\
    &= \sigma^2\sum_{\ell=0}^{d-1}\big|\widetilde{\s{X}}_k[\ell]\big|^2
    \Big(1-\cos\!\big(\tfrac{2\pi\ell}{d}m\big)\Big),
\end{align}
which proves the identity in \eqref{eq:Var_diff_conditional_proved}. Each summand is nonnegative.

It remains to show strict positivity. We show that there is at least one term in the sum that is strictly positive. Define
\begin{align}
    H_m  \triangleq  \{\ell\in\{0,\ldots,d-1\}:\ \ell m \equiv 0\ (\mathrm{mod}\ d)\}.
\end{align}
Then $|H_m|=\gcd(d,m)\le d/2$ since $m\not\equiv 0\ (\mathrm{mod}\ d)$. Because $d\ge 6$, we have
\begin{align}
    \big|\{0,\ldots,d-1\}\setminus(H_m\cup\{k,d-k\})\big| 
    &  \ge  d-|H_m|-2 \\ &  \ge  d-\frac d2-2  =  \frac d2-2  > 0.
\end{align}
Hence we may choose $\ell_0\in\{0,\ldots,d-1\}\setminus(H_m\cup\{k,d-k\})$. By construction,
$\ell_0\notin H_m$ implies $\ell_0 m\not\equiv 0\ (\mathrm{mod}\ d)$, i.e.\ $\cos(2\pi\ell_0 m/d)\neq 1$ and thus
$1-\cos(2\pi\ell_0 m/d)>0$. Moreover, since $\widetilde{\s{X}}_k[\ell]=0$ only for $\ell\in\{k,d-k\}$ and
$\s{X}[\ell]\neq 0$ for all $\ell$, we have $|\widetilde{\s{X}}_k[\ell_0]|^2>0$. Therefore the $\ell_0$-term in the sum is
strictly positive, and the entire variance is strictly positive, proving \eqref{eq:Var_diff_conditional_proved}.

Finally, for any $r\neq s$, the difference $(\s{S}_i\vert \s{N}_i[k])_r-(\s{S}_i\vert \s{N}_i[k])_s$ is a non-degenerate Gaussian,
hence $\mathbb{P}\!\big((\s{S}_i\vert \s{N}_i[k])_r=(\s{S}_i\vert \s{N}_i[k])_s\big)=0$. A union bound over finitely many pairs
implies ties occur with probability $0$, so the maximizer is unique almost surely.
\end{proof}

\subsection{Positive probability of each maximizer event}
\begin{lem}[Positive probability of each strict maximizer event]
\label{lem:positive_prob_each_argmax_conditional}
Recall the definition of $\s{S}_i$ from~\eqref{equ:maxGaussDef}. Assume $d\ge 6$ is even. Fix
$k\in\{1,\ldots,\frac{d}{2}-1,\frac{d}{2}+1,\ldots,d-1\}$ and recall from Lemma~\ref{lemma:conditioning} that
\begin{align}
    \label{eqn:model-and-covaraince-matrix}
    \s{S}_i\vert \s{N}_i[k]\sim\mathcal{N}(\mu_{k,i},\Sigma_{k,i}),
    \qquad
    \Sigma_{k,i}[r,s]=\frac{\sigma^2}{2}\sum_{\ell=0}^{d-1}\big|\widetilde{\s{X}}_k[\ell]\big|^2 \cos\!\Big(\frac{2\pi\ell}{d}(r-s)\Big).
\end{align}
Assume moreover that the template spectrum is non-vanishing, i.e.\ $\s{X}[\ell]\neq 0$ for all $\ell\in\{0,\ldots,d-1\}$.
Then, for every $r\in\{0,\ldots,d-1\}$, the event
\begin{align}
    \mathcal{C}_r  \triangleq  \Big\{\, (\s{S}_i\vert \s{N}_i[k])_r > \max_{t\neq r} (\s{S}_i\vert \s{N}_i[k])_t \,\Big\}
\end{align}
has strictly positive probability:
\begin{align}
    \mathbb{P}\big(\mathcal{C}_r \,\big|\, \s{N}_i[k]\big)>0,
\end{align}
for almost every realization of $\s{N}_i[k]$.
\end{lem}

\begin{proof}[Proof of Lemma~\ref{lem:positive_prob_each_argmax_conditional}]
Fix $k$ as in the statement and condition on a realization of $\s{N}_i[k]$. Then $Y\triangleq \s{S}_i\vert \s{N}_i[k]\sim\mathcal{N}(m,\Sigma)$ with $m=\mu_{k,i}$ and $\Sigma=\Sigma_{k,i}$.
By Lemma~\ref{lem:unique_conditional_argmax_model}, for every $r\neq s$, $\operatorname{Var}(Y_r-Y_s)>0$, hence $Y_r-Y_s$ is a non-degenerate Gaussian and $\mathbb{P}(Y_r=Y_s)=0$. In particular, ties occur with probability $0$.

Since $\Sigma$ is real, symmetric, and circulant, it is diagonalized by the DFT basis: there exist eigenvectors $\{f_\ell\}_{\ell=0}^{d-1}$ (the Fourier modes) such that the corresponding eigenvalues $\{\lambda_\ell\}_{\ell=0}^{d-1}$ are given by the DFT of the first row of $\Sigma$.
In Model~\ref{eqn:model-and-covaraince-matrix}, Lemma~\ref{lemma:conditioning} shows that these eigenvalues satisfy
\begin{align}
    \lambda_\ell \ \propto\ |\widetilde{\s{X}}_k[\ell]|^2,\qquad \ell=0,\ldots,d-1.
\end{align}
Under the non-vanishing spectrum assumption $\s{X}[\ell]\neq 0$ for all $\ell$, and by the definition
\eqref{eqn:XtildeDefenition}, we have
\begin{align}
    |\widetilde{\s{X}}_k[\ell]|^2>0 \ \quad \text{for all }\ell\notin\{k,d-k\},
\end{align}
while $|\widetilde{\s{X}}_k[k]|^2=|\widetilde{\s{X}}_k[d-k]|^2=0$.
Hence $\Sigma$ has exactly two zero eigenvalues, corresponding to the $\ell=k$ and $\ell=d-k$ Fourier modes.
In the real domain, these two modes are equivalently represented by the cosine and sine vectors $c,s \in \mathbb{R}^d$
\begin{align}
    \label{eqn:def-c-s}
    c_t=\cos\!\Big(\tfrac{2\pi k}{d}t\Big),\qquad s_t=\sin\!\Big(\tfrac{2\pi k}{d}t\Big),
\end{align}
for $t=0,\ldots,d-1$, so
\begin{align}
    \label{eqn:c-s-span-ker}
    \ker(\Sigma)=\operatorname{span}\{c,s\}.
\end{align}
Let $L\triangleq \operatorname{range}(\Sigma)$. Since $\Sigma$ is symmetric, we have $L=(\ker\Sigma)^\perp$, and therefore the Gaussian vector $Y\sim\mathcal{N}(m,\Sigma)$ is supported on the affine subspace $m+L$.

We fix $r\in\{0,\ldots,d-1\}$ and consider the open cone
\begin{align}
    \mathcal{C}_r \triangleq  \{y\in\mathbb{R}^d:\ y_r>\max_{t\neq r}y_t\}.
\end{align}
We claim that $\mathcal{C}_r\cap(m+L)\neq\emptyset$, i.e., the affine support contains at least one point whose unique largest coordinate is the $r$-th.

To build such a point, we start from the $r$-th standard basis vector $u=e_r$, whose maximizer is trivially at $r$ (i.e., $e_r \in \mathcal{C}_r$), but which may fail to belong to $m + L$.
We therefore remove from $u$ its components along the two forbidden directions $c$ and $s$ (which span $\ker(\Sigma)$~\eqref{eqn:c-s-span-ker}). We define
\begin{align}
    \alpha \triangleq \frac{\langle u,c\rangle}{\langle c,c\rangle},\qquad
    \beta \triangleq \frac{\langle u,s\rangle}{\langle s,s\rangle},\qquad
    v\triangleq u-\alpha\,c-\beta\,s .
\end{align}
Because $k\notin\{0,d/2\}$, the sine and cosine vectors are orthogonal and have equal energy:
\begin{align}
    \label{eqn:orth-basis}
    \langle c,s\rangle=0,\qquad \langle c,c\rangle=\langle s,s\rangle=d/2.
\end{align}
Hence $\langle v,c\rangle=\langle v,s\rangle=0$, which means $v\in(\ker\Sigma)^\perp=L$.

Next we show that $v$ remains strongly peaked at the $r$-th coordinate. Since $u=e_r$, we have $\langle u,c\rangle=c_r$ and $\langle u,s\rangle=s_r$, and therefore
\begin{align}
    v_r = 1-\frac{c_r^2}{\langle c,c\rangle}-\frac{s_r^2}{\langle s,s\rangle} = 1-\frac{c_r^2+s_r^2}{d/2} = 1-\frac{2}{d},
\end{align}
where we used $c_r^2+s_r^2=1$~\eqref{eqn:def-c-s}. For $t\neq r$, since $u=e_r$ is the $r$-th standard basis vector, we have $u_t=0$; thus, using~\eqref{eqn:def-c-s} and~\eqref{eqn:orth-basis}, we obtain
\begin{align}
    v_t = -\alpha \,c_t- \beta \,s_t = -\frac{c_r c_t+s_r s_t}{d/2} = -\frac{2}{d}\cos\!\Big(\tfrac{2\pi k}{d}(t-r)\Big),
\end{align}
so $|v_t|\le 2/d$ for all $t\neq r$.
Consequently, letting
\begin{align}
    \delta \triangleq v_r-\max_{t\neq r}v_t,
\end{align}
we obtain the uniform lower bound for $d > 4$
\begin{align}
    \delta \ge \Big(1-\frac{2}{d}\Big)-\frac{2}{d} = 1-\frac{4}{d} > 0.
\end{align}
Thus, within the direction $v\in L$, the $r$-th coordinate exceeds every other coordinate by at least $\delta$.

Finally, because the affine support is $m+L$, any point of the form $y(\gamma)=m+\gamma v$ lies in the support. The offset $m$ may change the ordering for small $\gamma$, but scaling $\gamma$ makes the $v$-term dominate. In particular, let
\begin{align}
    M\triangleq \max_{t\neq r}|m_r-m_t|.
\end{align}
Then for any $t\neq r$,
\begin{align}
    y_r(\gamma)-y_t(\gamma)=(m_r-m_t)+\gamma(v_r-v_t)\ \ge\ -M+\gamma\delta.
\end{align}
Choosing $\gamma>M/\delta$ guarantees $y_r(\gamma)>y_t(\gamma)$ for all $t\neq r$, hence
\begin{align}
    y(\gamma)\in \mathcal{C}_r\cap(m+L),
\end{align}
proving $\mathcal{C}_r\cap(m+L)\neq\emptyset$.

Since $\mathcal{C}_r$ is open in $\mathbb{R}^d$, the intersection $\mathcal{C}_r\cap(m+L)$ contains a nonempty open subset of the affine support $m+L$. A Gaussian measure assigns positive probability to any nonempty open subset of its affine support~\cite{kallenberg1997foundations}; therefore, $\mathbb{P}(Y\in\mathcal{C}_r)>0$, or equivalently $\mathbb{P}\!\big(\mathcal{C}_r \,\big|\, \s{N}_i[k]\big)>0$, completing the proof.
\end{proof}

\subsection{Auxiliary result for Proposition~\ref{prop:2}}

Let $\s{S}^{(+)} \sim \calN\!\left( \mu, \s{\Sigma} \right)$ and $\s{S}^{(-)} \sim \calN\!\left( -\mu, \s{\Sigma} \right)$ be two $d$-dimensional Gaussian vectors, where $\s{\Sigma}$ is a real, symmetric, circulant covariance matrix. Define the maximizers
    \begin{align}
        \s{\hat{R}}^{(+)} &= \argmax_{0 \leq \ell \leq {d-1}}\s{S}_\ell^{(+)},
        \label{eqn:27}\\
        \s{\hat{R}}^{(-)} &= \argmax_{0 \leq \ell \leq {d-1}}\s{S}_\ell^{(-)}.
        \label{eqn:28}
    \end{align}
and assume they are unique almost surely. We define the entries of $\mu$ as,
    \begin{align}   
        \mu_\ell\triangleq[\mu]_\ell = \cos { \left( \frac{2\pi k}{d}\ell + \varphi \right)},
        \label{eqn:35}
    \end{align}
for $\varphi \in [0, 2\pi)$, and $0\leq \ell\leq d-1$. Note that $-\mu_\ell = \cos { \left( \frac{2\pi k}{d}\ell + \varphi +\pi\right)}$, for $0\leq \ell\leq d-1$. 

Then, we have the following result.
\begin{proposition} \label{lemma:1}
Consider the definitions in \eqref{eqn:35}--\eqref{eqn:28}, and assume the maximizers in~\eqref{eqn:28} are unique a.s. Moreover, assume that for every $r\in\{1,\ldots,d-1\}$,
\begin{align}
    \label{eq:pos_prob_Cr_assumption} \mathbb{P}\!\left(\s{S}^{(+)}_r>\max_{t\neq r}\s{S}^{(+)}_t\right)>0
    \qquad\text{and}\qquad    \mathbb{P}\!\left(\s{S}^{(-)}_r>\max_{t\neq r}\s{S}^{(-)}_t\right)>0 .
\end{align}
Fix $0 \leq \ell \leq {d-1}$. If $\mu_\ell> 0$, then,
    \begin{align}   
        \mathbb{P} \left[ \hat{\s{R}}^{(+)} = \ell \right] > \mathbb{P} \left[  \hat{\s{R}}^{(-)} = \ell \right],\label{eqn:lemma1conc0}
    \end{align} 
    otherwise, if $\mu_\ell<0$, then,
    \begin{align}   
        \mathbb{P} \left[ \hat{\s{R}}^{(-)} = \ell \right] > \mathbb{P}  \left[  \hat{\s{R}}^{(+)} = \ell \right].\label{eqn:lemma1conc1}
    \end{align}
    In particular, for any $\varphi\in[0,2\pi)$ and $0\leq k\leq d-1$,
    \begin{align}   
        \mathbb{E}\pp{\cos\left(\frac{2\pi k}{d} \hat{\s{R}}^{(+)} + \varphi \right)} + \mathbb{E}\pp{\cos\left(\frac{2\pi k}{d} \hat{\s{R}}^{(-)} + \varphi + \pi \right)} > 0.\label{eqn:lemma1conc}
    \end{align}
\end{proposition}

\begin{proof}[Proof of Proposition~\ref{lemma:1}]
By definition, it is clear that,
    \begin{align}   
        \mathbb{P} \left[ \hat{\s{R}}^{(+)} = \ell \right] = \mathbb{P} \left[\s{S}_\ell^{(+)} \geq \max_{m\neq\ell}\s{S}_{m}^{(+)} \right], 
    \end{align}
and,
    \begin{align}   
     \mathbb{P} \left[ \hat{\s{R}}^{(-)} = \ell \right] = \mathbb{P} \left[\s{S}_\ell^{(-)} \geq \max_{m\neq\ell}\s{S}_{m}^{(-)} \right],
    \end{align}
for $0\leq\ell\leq d-1$. Since $\s{S}^{(+)}$ and $\s{S}^{(-)}$ can be decomposed as $\s{S}^{(+)} = \s{Z} + \mu$ and $\s{S}^{(-)} = \s{Z} - \mu$, where $\s{Z}$ is a cyclo-stationary process, and $\mu$ is defined in \eqref{eqn:35}. Then,
    \begin{align}   
        \mathbb{P} \left[\s{S}_\ell^{(+)} \geq \max_{m\neq\ell}\s{S}_{m}^{(+)} \right] = \mathbb{P} \left[ \s{Z}_\ell + \mu_\ell \geq \max_{m \neq \ell}\s{Z}_m + \mu_m \right],
    \end{align}
and,
\begin{align}   
        \mathbb{P} \left[\s{S}_\ell^{(-)} \geq \max_{m\neq\ell}\s{S}_{m}^{(-)} \right] = \mathbb{P} \left[ \s{Z}_\ell - \mu_\ell \geq \max_{m \neq \ell}\s{Z}_m - \mu_m \right].
    \end{align}
We will show that for any $\ell$ such that $\mu_\ell > 0$, we have,
    \begin{align}   
        \mathbb{P} \left[\s{Z}_\ell \geq \max_{m \neq \ell} \left\{ \s{Z}_m + \mu_m - \mu_\ell \right\} \right] >  \mathbb{P} \left[ \s{Z}_\ell \geq \max_{m \neq \ell} \left\{ \s{Z}_m - \mu_m + \mu_\ell \right\} \right], \label{eqn:A44}
    \end{align}
which in turn implies that $ \mathbb{P}\{ \hat{\s{R}}^{(+)} = \ell\} > \mathbb{P}\{ \hat{\s{R}}^{(-)} = \ell\}$. 

By definition, since $\s{Z}$ is a zero-mean Gaussian, cyclo-stationary random process (i.e., with a real, symmetric, circulant covariance matrix), its cumulative distribution function $F_{\s{Z}}$ is invariant under cyclic shifts, i.e.,
    \begin{align}   
         F_{\s{Z}}\left(z_0, z_1,\ldots,z_{d-1} \right) = F_{\s{Z}}\left( z_\tau, z_{\tau+1},\ldots,z_{\tau+d-1} \right),
         \label{eqn:stationarityTimeShift}
    \end{align}
for any $\tau \in \mathbb{Z}$, with indices taken modulo $d$. Moreover, reversing the time indices does not affect the distribution; that is,    \begin{align}   
         F_{\s{Z}}\left( z_0, z_1,\ldots, z_{\ell-1}, z_\ell, z_{\ell+1},\ldots,z_{d-1} \right) = F_{\s{Z}}\left( z_{d-1}, z_{d-2},\ldots,z_{\ell+1}, z_\ell, z_{\ell-1}, ..., z_{0} \right).
         \label{eqn:stationarityTimeReverse}
    \end{align}
Combining \eqref{eqn:stationarityTimeShift} and \eqref{eqn:stationarityTimeReverse} yields,
    \begin{align}   
         F_{\s{Z}}\left( z_\ell, z_{\ell+1}, z_{\ell+2},\ldots,z_{\ell-2}, z_{\ell-1} \right) = F_{\s{Z}}\left( z_\ell, z_{\ell-1}, z_{\ell-2},\ldots,z_{\ell+2}, z_{\ell+1} \right).
         \label{eqn:ZrRevert}
    \end{align}
Accordingly, let us define the Gaussian vectors $\s{Z}^{(1)}$ and $\s{Z}^{(2)}$, such that their $m$-th entry is,
    \begin{align}   
         [\s{Z}^{(1)}]_m = \s{Z}_{\ell+m} - \s{Z}_\ell,
    \end{align}
    \begin{align}   
         [\s{Z}^{(2)}]_m = \s{Z}_{\ell-m} - \s{Z}_\ell,
    \end{align}
for $1 \leq m \leq d-1$. It is clear from \eqref{eqn:ZrRevert} that $\s{Z}^{(1)}$ and $\s{Z}^{(2)}$ have the same cumulative distribution function, i.e.,
    \begin{align}   
         F_{\s{Z}^{(1)}} = F_{\s{Z}^{(2)}}.
         \label{eqn:cyclicPropoertyOfS}
    \end{align}
Therefore, the following holds,
    \begin{align}   
         &\mathbb{P} \left[ \s{Z}_\ell \geq \max_{m \neq 0} \left\{ \s{Z}_{\ell+m} + \mu_{\ell+m} - \mu_\ell \right\} \right] \nonumber\\
         &\qquad\qquad= \mathbb{P} \left[0 \geq \max_{m \neq 0} \left\{ \s{Z}_{\ell+m} - \s{Z}_\ell + \mu_{\ell+m} - \mu_\ell \right\} \right] \nonumber \\ 
         &\qquad\qquad  = \mathbb{P} \left[ 0 \geq \max_{m \neq 0} \left\{ \s{Z}_{\ell-m} - \s{Z}_\ell + \mu_{\ell+m} - \mu_\ell \right\} \right]\nonumber\\
         &\qquad\qquad= \mathbb{P} \left[ \s{Z}_\ell \geq \max_{m \neq 0} \left\{ \s{Z}_{\ell-m} + \mu_{\ell+m} - \mu_\ell \right\} \right],
        \label{eqn:77}
    \end{align}
where the second equality follows from \eqref{eqn:cyclicPropoertyOfS}. Next, we note that for every $0 < m \leq d-1$ and $\mu_\ell > 0$,
\begin{align}   
        \mu_{\ell-m} + \mu_{\ell+m} = 2\mu_\ell \cos \left( \frac{2\pi k}{d}m \right).
    \end{align}
Therefore,
    \begin{align}   
        \mu_\ell - \mu_{\ell-m} + \mu_\ell - \mu_{\ell+m} = 2\mu_\ell \left( 1 - \cos \left( \frac{2\pi k}{d}m \right)  \right ) \geq 0,
        \label{eqn:67}
    \end{align}
which implies
    \begin{align}   
        \mu_\ell - \mu_{\ell-m} \geq \mu_{\ell+m} - \mu_{\ell},
        \label{eqn:70}
    \end{align}
or, equivalently,
    \begin{align}   
        \mu_\ell - \mu_{\ell+m} \geq \mu_{\ell-m} - \mu_{\ell}.
        \label{eqn:71}
    \end{align}
\begin{remark} \label{remark:A4}
    According to \eqref{eqn:67}, the inequalities in \eqref{eqn:70} and \eqref{eqn:71} are strict whenever $\cos \left( \frac{2\pi k}{d}m \right) < 1$, which holds for the majority of values of $m$. In particular, at least $d/2$ of the inequalities are strict for $0 \leq m \leq d-1$.
\end{remark}

Following from \eqref{eqn:70}, \eqref{eqn:71}, and the last remark, we have the following auxiliary Lemma, which we prove below.
\begin{lem} \label{lem:A4}
    Assume the maximizers in~\eqref{eqn:28} are unique a.s. and satisfy~\eqref{eq:pos_prob_Cr_assumption}. Then, for $\mu_\ell > 0$, we have,
    \begin{align}   
        \mathbb{P} \left[ \s{Z}_\ell \geq \max_{m \neq 0} \left\{ \s{Z}_{\ell+m} + \mu_{\ell+m} - \mu_\ell \right\} \right]  > \mathbb{P} \left[ \s{Z}_\ell \geq \max_{m \neq 0} \left\{ \s{Z}_{\ell-m} - \mu_{\ell-m} + \mu_\ell \right\} \right]. \label{eqn:A57}
    \end{align}
\end{lem}
Note that \eqref{eqn:A57} is equivalent to the following expression, by a change of index notation:
\begin{align}   
        \mathbb{P} \left[ \s{Z}_\ell \geq \max_{m \neq \ell} \left\{ \s{Z}_m + \mu_m - \mu_\ell \right\} \right] >  \mathbb{P} \left[ \s{Z}_\ell \geq \max_{m \neq \ell} \left\{ \s{Z}_m - \mu_m + \mu_\ell \right\} \right],
    \end{align}
which proves \eqref{eqn:A44}.    
A similar result can be obtained for the case where $\mu_\ell < 0$, i.e.,
    \begin{align}   
        \mathbb{P} \left[ \s{Z}_\ell \geq \max_{m \neq \ell} \left\{ \s{Z}_m + \mu_m - \mu_\ell \right\} \right] <  \mathbb{P} \left[ \s{Z}_\ell \geq \max_{m \neq \ell} \left\{ \s{Z}_m - \mu_m + \mu_\ell \right\} \right],
    \end{align}
which completes the proofs of \eqref{eqn:lemma1conc0}--\eqref{eqn:lemma1conc1}. 

Finally, we prove \eqref{eqn:lemma1conc}. By definition, it is clear that
    \begin{align}   
        &\mathbb{E} \left[\cos \left(\frac{2\pi k}{d} \hat{\s{R}}^{(+)} + \varphi \right) \right]+\mathbb{E} \left[ \cos \left(\frac{2\pi k}{d} \hat{\s{R}}^{(-)} + \varphi + \pi \right) \right] \nonumber \\
        &\qquad\qquad=  \sum_{\ell=0}^{d-1} \cos \left(\frac{2\pi k}{d}\ell+ \varphi \right) \left[ \mathbb{P} \left(\hat{\s{R}}^{(+)} = \ell  \right) - \mathbb{P} \left( \hat{\s{R}}^{(-)} = \ell \right)\right],\label{eqn:sumExpec}
    \end{align}
where we have used the fact that $\cos(\alpha+\pi) = -\cos\alpha$, for any $\alpha\in\mathbb{R}$.

By \eqref{eqn:lemma1conc0}--\eqref{eqn:lemma1conc1}, as for any $0\leq\ell\leq d-1$ such that $\mu_\ell=\cos \left(\frac{2\pi k}{d}\ell + \varphi \right) > 0$ it holds that $\mathbb{P} [\hat{\s{R}}^{(+)} = \ell] > \mathbb{P} [\hat{\s{R}}^{(-)} = \ell]$, otherwise, for $0\leq\ell\leq d-1$ such that $\mu_\ell=\cos \left(\frac{2\pi k}{d}\ell + \varphi \right) < 0$, it holds that $\mathbb{P} [\hat{\s{R}}^{(+)} = \ell] < \mathbb{P} [\hat{\s{R}}^{(-)} = \ell]$. Therefore,
    \begin{align}   
        \sum_{\ell=0}^{d-1} \cos \left(\frac{2\pi k}{d}\ell+ \varphi \right) \left[ \mathbb{P} \left(\hat{\s{R}}^{(+)} = \ell  \right) - \mathbb{P} \left( \hat{\s{R}}^{(-)} = \ell \right)\right]> 0,
    \end{align}
which in light of \eqref{eqn:sumExpec} concludes the proof. 
\end{proof}   

It is left to prove Lemma \ref{lem:A4}.

\begin{proof}[Proof of Lemma \ref{lem:A4}]
Using \eqref{eqn:70} and \eqref{eqn:71}, we obtain the following inequalities for $\mu_\ell > 0$,
    \begin{align}   
        \max_{m \neq 0} \left\{ \s{Z}_{\ell-m} - \mu_{\ell+m} + \mu_\ell \right\} \geq \max_{m \neq 0} \left\{ \s{Z}_{\ell-m} + \mu_{\ell-m} - \mu_\ell \right\},
        \label{eqn:64}
    \end{align}
and
    \begin{align}   
        \max_{m \neq 0} \left\{ \s{Z}_{\ell-m} - \mu_{\ell-m} + \mu_\ell \right\} \geq \max_{m \neq 0} \left\{ \s{Z}_{\ell-m} + \mu_{\ell+m} - \mu_\ell \right\},
        \label{eqn:65}
    \end{align}
As a result, we also have the following probabilistic inequalities,
        \begin{align}   
        \mathbb{P} \pp{\s{Z}_\ell < \max_{m \neq 0} \left\{ \s{Z}_{\ell-m} - \mu_{\ell+m} + \mu_\ell \right\}} 
        \geq \mathbb{P} \pp{\s{Z}_\ell < \max_{m \neq 0} \left\{ \s{Z}_{\ell-m} + \mu_{\ell-m} - \mu_\ell \right\}},
        \label{eqn:A63}
    \end{align}
and,
    \begin{align}   
        \mathbb{P} \pp{\s{Z}_\ell < \max_{m \neq 0} \left\{ \s{Z}_{\ell-m} - \mu_{\ell-m} + \mu_\ell \right\}}
        \geq \mathbb{P} \pp{\s{Z}_\ell < \max_{m \neq 0} \left\{ \s{Z}_{\ell-m} + \mu_{\ell+m} - \mu_\ell \right\}}.
        \label{eqn:A64}
    \end{align} 
Next, we show that these probabilistic inequalities are in fact strict. Define the set of indices where the inequality in \eqref{eqn:71} is strict,
\begin{align}
    \mathcal{M} = \ppp{m : \mu_\ell - \mu_{\ell+m} > \mu_{\ell-m} - \mu_{\ell}}. \label{eqn:A65}
\end{align}
From Remark~\ref{remark:A4}, we know that $\abs{\mathcal{M}} \geq d/2$.
Now define the event,
\begin{align}
    \mathcal{C}_r =  \ppp{\s{Z}_{\ell-r} - \mu_{\ell-r} + \mu_\ell  = \max_{m \neq 0} \ppp{\s{Z}_{\ell-m} - \mu_{\ell-m} + \mu_\ell } },
\end{align}
i.e., the event that $\s{Z}_{\ell-r} - \mu_{\ell-r} + \mu_\ell$ attains the maximum in the expression above.
By assumption~\eqref{eq:pos_prob_Cr_assumption}, the event $\mathcal{C}_r$ for every $r$, i.e.\ $\mathbb{P}(\mathcal{C}_r)>0$. Thus, we may choose $r\in\mathcal{M}$ with $\mathbb{P}(\mathcal{C}_r)>0$, i.e.,
\begin{align}
    r \in \mathcal{M}, \quad\text{and}\quad
    \mathbb{P}(\mathcal{C}_r) =\mathbb{P}\!\left( \max_{m \neq 0}\Big\{\s{Z}_{\ell-m} - \mu_{\ell-m} + \mu_\ell\Big\} = \s{Z}_{\ell-r} - \mu_{\ell-r} + \mu_\ell \right) > 0 .
\end{align}Then, by the law of total probability,
\begin{align}
    \nonumber \mathbb{P} & \pp{\s{Z}_\ell < \max_{m \neq 0} \ppp{ \s{Z}_{\ell-m} - \mu_{\ell-m} + \mu_\ell }} = \mathbb{P} \pp{\s{Z}_\ell < \max_{m \neq 0} \ppp{ \s{Z}_{\ell-m} - \mu_{\ell-m} + \mu_\ell } \ \vert \ \mathcal{C}_r} \mathbb{P}\pp{\mathcal{C}_r} 
    \\ &  \hspace{2cm}+ \mathbb{P} \pp{\s{Z}_\ell < \max_{m \neq 0} \ppp{ \s{Z}_{\ell-m} - \mu_{\ell-m} + \mu_\ell } \ \vert \ \mathcal{C}_r^c} \mathbb{P}\pp{\mathcal{C}_r^c}. \label{eqn:A67}
\end{align}
From \eqref{eqn:A64}, we have,
\begin{align}
    \mathbb{P} \pp{\s{Z}_\ell < \max_{m \neq 0} \ppp{ \s{Z}_{\ell-m} - \mu_{\ell-m} + \mu_\ell } \ \vert \ \mathcal{C}_r^c} \geq \mathbb{P} \pp{\s{Z}_\ell < \max_{m \neq 0} \ppp{ \s{Z}_{\ell-m} + \mu_{\ell+m} - \mu_\ell } \ \vert \ \mathcal{C}_r^c}. \label{eqn:A68}
\end{align}
Additionally, since $r \in \mathcal{M}$, it follows that,
\begin{align}
     \mathbb{P} \pp{\s{Z}_\ell < \max_{m \neq 0} \ppp{ \s{Z}_{\ell-m} - \mu_{\ell-m} + \mu_\ell } \ \vert \ \mathcal{C}_r}  >
     \mathbb{P} \pp{\s{Z}_\ell < \max_{m \neq 0} \ppp{ \s{Z}_{\ell-m} + \mu_{\ell+m} - \mu_\ell } \ \vert \ \mathcal{C}_r}. \label{eqn:A69}
\end{align}
Substituting \eqref{eqn:A68} and \eqref{eqn:A69} into \eqref{eqn:A67} yields,
\begin{align}
    \nonumber \mathbb{P} & \pp{\s{Z}_\ell < \max_{m \neq 0} \ppp{ \s{Z}_{\ell-m} - \mu_{\ell-m} + \mu_\ell }} 
    > \mathbb{P} \pp{\s{Z}_\ell < \max_{m \neq 0} \ppp{ \s{Z}_{\ell-m} + \mu_{\ell+m} - \mu_\ell } \ \vert \ \mathcal{C}_r} \mathbb{P}\pp{\mathcal{C}_r} 
    \\ &  + \mathbb{P} \pp{\s{Z}_\ell < \max_{m \neq 0} \ppp{ \s{Z}_{\ell-m} + \mu_{\ell+m} - \mu_\ell } \ \vert \ \mathcal{C}_r^c} \mathbb{P}\pp{\mathcal{C}_r^c}. \label{eqn:A71}
\end{align}
By the law of total probability, the right-hand-side of \eqref{eqn:A71} is,
\begin{align}
    \nonumber \mathbb{P} & \pp{\s{Z}_\ell < \max_{m \neq 0} \ppp{ \s{Z}_{\ell-m} + \mu_{\ell+m} - \mu_\ell }} 
    = \mathbb{P} \pp{\s{Z}_\ell < \max_{m \neq 0} \ppp{ \s{Z}_{\ell-m} + \mu_{\ell+m} - \mu_\ell } \ \vert \ \mathcal{C}_r} \mathbb{P}\pp{\mathcal{C}_r} 
    \\ &  + \mathbb{P} \pp{\s{Z}_\ell < \max_{m \neq 0} \ppp{ \s{Z}_{\ell-m} + \mu_{\ell+m} - \mu_\ell } \ \vert \ \mathcal{C}_r^c} \mathbb{P}\pp{\mathcal{C}_r^c}. \label{eqn:A72}
\end{align}
Combining \eqref{eqn:A71} and \eqref{eqn:A72}, we conclude,
\begin{align}
      \mathbb{P} \pp{\s{Z}_\ell < \max_{m \neq 0} \ppp{ \s{Z}_{\ell-m} - \mu_{\ell-m} + \mu_\ell }} > \mathbb{P} \pp{\s{Z}_\ell < \max_{m \neq 0} \ppp{ \s{Z}_{\ell-m} + \mu_{\ell+m} - \mu_\ell }}. \label{eqn:A73}
\end{align}
Equivalently, we can express \eqref{eqn:A73} as a complementary event, and obtain,
\begin{align}   
    \mathbb{P} \left[ \s{Z}_\ell \geq \max_{m \neq 0} \left\{ \s{Z}_{\ell+m} + \mu_{\ell+m} - \mu_\ell \right\} \right]  > \mathbb{P} \left[ \s{Z}_\ell \geq \max_{m \neq 0} \left\{ \s{Z}_{\ell-m} - \mu_{\ell-m} + \mu_\ell \right\} \right],
\end{align}
which proves \eqref{eqn:A57}, and completes the proof.
\end{proof}

\section{Proof of Theorem \ref{thm:1}} \label{sec:auxForTheorem1}

First, we prove several auxiliary statements needed in the proof of Theorem \ref{thm:1}. Recall the definition of $\phi_{e,i}[k]$ in \eqref{eqn:phaseDifferenceTerm} and of $\mu_{\s{A},k}, \mu_{\s{B},k}, \sigma_{\s{A},k}^2, \sigma_{\s{B},k}^2$ in \eqref{eqn:muA}--\eqref{eqn:sigmaB}.

\paragraph{Notation for convergence rate of the Fourier phases.}
Denote for every $0 \leq k \leq d-1$,
\begin{align}
    \s{A}_{M,k} \triangleq \frac{1}{\sqrt{M}}\sum_{i=0}^{M-1}\abs{\s{N}_i[k]} \sin \left( \phi_{e,i}[k] \right), \label{eqn:AMdef}
\end{align}
and,
\begin{align}
    \s{B}_{M,k} \triangleq \frac{1}{M}\sum_{i=0}^{M-1}\abs{\s{N}_i[k]} \cos \left( \phi_{e,i}[k] \right). \label{eqn:BMdef}
\end{align}
Note that $\s{A}_{M,k}$ is the imaginary part of the \EfN estimator, multiplied by the phase of the template signal as defined in \eqref{eqn:A11}, but is normalized by $1/\sqrt{M}$ instead of $1/M$, to facilitate the analysis of the convergence rate. Similarly, $\s{B}_{M,k}$ corresponds to the real part in \eqref{eqn:A11}. Additionally, we define the following Gaussian random variable $\s{Q}_k$
\begin{align}
   \s{Q}_k \sim \calN \left( 0, \frac{\sigma_{\s{A},k}^2}{\mu_{\s{B},k}^2} \right),
   \label{eqn: A51}
\end{align}
for every $0 \leq k \leq d-1$.

\paragraph{The main results of this section.}
Recall that by the SLLN \eqref{eqn:strongLLN}, the EfN estimator converges to,
\begin{align}
   \s{\hat{X}}[k] & \xrightarrow[]{\s{a.s.}} e^{j\phi_{\s{X}}[k]} \cdot
    \mathbb{E} \left[ \abs{\s{N}_1[k]} \cos\p{\phi_{e,1}[k]} \right] + j \mathbb{E} \left[ \abs{\s{N}_1[k]} \sin\p{\phi_{e,1}[k]} \right]
    \\ & = e^{j\phi_{\s{X}}[k]} \p{\mu_{\s{B},k} + j \mu_{\s{A},k}}.\label{eqn:B5}
\end{align}
In Sections \ref{sec:imaginartPartVanishing} and \ref{sec:realPartGreaterThan0}, we prove that $\mu_{\s{A},k} = 0$ and $\mu_{\s{B},k} > 0$. These results, combined with \eqref{eqn:B5}, imply that the EfN estimator converges to a non-vanishing signal, and its Fourier phases converge those of Einstein as $M \to \infty$. In Sections \ref{sec:proof_prop_1} and \ref{sec:convergenceInExpectationOfFourierPhases}, we analyze the convergence rate of the Fourier phases, first establishing convergence rate in distribution to $\s{Q}_k$, and then proving convergence rate in expectation.

\subsection{Convergence of the Fourier phases} \label{sec:imaginartPartVanishing}

\begin{lem}[Convergence of the Fourier phases] \label{lemma:A1}
    Recall the definition of $\phi_{e,i}[k]$ in \eqref{eqn:phaseDifferenceTerm}. Then we have,
    \begin{align}
        \mu_{\s{A},k} &\triangleq \mathbb{E}\left[ \abs{\s{N}_1[k]} \sin (\phi_{e,1}[k]) \right] = 0,
        \label{eqn:muAdefenition1}
    \end{align}
    for every $0 \leq k \leq d-1$.
\end{lem}
\begin{proof}[Proof of Lemma~\ref{lemma:A1}]
Let $\s{D}[k]\triangleq\phi_{\s{X}}[k] - \phi_{\s{N}_1}[k]$, and recall the definition of $\s{\hat{R}}_i$ in \eqref{eqn:OptShiftFourier}:
  \begin{align}
    \s{\hat{R}}_i = \underset{0\leq r\leq d-1}{\argmax} \ \sum_{k=0}^{d-1} {\abs{\s{X}[k]} \abs{\s{N}_i[k]} \, \cos \left( \frac{2\pi kr}{d} +  \phi_{\s{N}_i}[k] - \phi_{\s{X}}[k] \right)}.
\end{align}
Note that $\s{\hat{R}}_i$ is a function of
\begin{align}
    \s{\hat{R}}_i = \s{\hat{R}}_i \p{\ppp{|\s{N}_i[k]|}_{k=0}^{d-1}, \ppp{|\s{X}[k]|}_{k=0}^{d-1}, \ppp{\phi_{\s{N}_i}[k]}_{k=0}^{d-1}, \ppp{\phi_{\s{X}}[k]}_{k=0}^{d-1}},
\end{align}
and it depends on $\phi_{\s{N}_i}[k]$ and $\phi_{\s{X}}[k]$ only through $\s{D}[k]$. Accordingly, viewing $\s{\hat{R}}_1$ as a function of $\s{D}[k]$, for fixed $\ppp{|\s{N}_i[k]|}_{k=0}^{d-1}, \ppp{|\s{X}[k]|}_{k=0}^{d-1}$,  we have,
    \begin{align}
        \s{\hat{R}}_1 \left(-\s{D}[0],-\s{D}[1],\ldots,-\s{D}[d-1]\right)  = -\s{\hat{R}}_1 \left(\s{D}[0],\s{D}[1],\ldots,\s{D}[d-1]\right).
        \label{eqn:antiSymmetryPropoeryOfR}
    \end{align}
Namely, from symmetry arguments, by flipping the signs of all the phases, the location of the maximum flips its sign as well. Then, by the law of total expectation,
\begin{align}
         \mu_{\s{A},k} &= \mathbb{E} \left[ \abs{\s{N}_1[k]} \sin \left(\frac{2\pi k}{d}\s{\hat{R}}_1 + \phi_{\s{N}_1}[k] - \phi_{\s{X}}[k]  \right)  \right] \nonumber \\ 
         & = \mathbb{E} \left\{\abs{\s{N}_1[k]}\cdot\left. \mathbb{E}\left[\sin \left(\frac{2\pi k}{d}\s{\hat{R}}_1 + \phi_{\s{N}_1}[k] - \phi_{\s{X}}[k]  \right)\right|\{\abs{\s{N}_1[k]}\}_{k=0}^{d-1}  \right] \right\}.
         \label{eqn:lawOfTotalExpectationProp1}
\end{align}
The inner expectation in \eqref{eqn:lawOfTotalExpectationProp1} is taken w.r.t. the uniform distribution randomness of the phases $\{\phi_{\s{N}_1}[k]\}_{k=0}^{d-1} \in [-\pi, \pi)$. However, due to \eqref{eqn:antiSymmetryPropoeryOfR}, and since the sine function is odd around zero, the integration in~\eqref{eqn:lawOfTotalExpectationProp1} nullifies. Therefore,
    \begin{align}
        \left. \mathbb{E}\left[\sin \left(\frac{2\pi k}{d}\s{\hat{R}}_1 + \phi_{\s{N}_1}[k] - \phi_{\s{X}}[k]  \right)\right|\{\abs{\s{N}_1[k]}\}_{k=0}^{d-1}  \right] = 0,
        \label{eqn:muAisZero}
    \end{align}
and thus $\mu_{\s{A},k}=0$. 
\end{proof}

\subsection{Convergence to non-vanishing signal} \label{sec:realPartGreaterThan0}

\begin{proposition}[Convergence to non-vanishing signal]\label{prop:2}
Recall the definition of $\phi_{e,i}[k]$ in \eqref{eqn:phaseDifferenceTerm}. Fix $d\in\mathbb{N}$, and assume that $\s{X}[k] \neq 0$ for all $0< k \leq d-1$. Then, for any $0\leq k\leq d-1$,
    \begin{align}   
        \mu_{\s{B},k} \triangleq \mathbb{E}{\left[ \abs{\s{N}_1[k]} \cos(\phi_{e,1}[k]) \right]}>0.
        \label{eqn:targetFunction}
    \end{align}
\end{proposition}

\begin{proof}[Proof of Proposition~\ref{prop:2}]

By the law of total expectation, we have,
\begin{align}
    \mathbb{E}{\left[ \abs{\s{N}_1[k]} \cos(\phi_{e,1}[k]) \right]} &= \mathbb{E}\pp{\left.\abs{\s{N}_1[k]} \cdot\mathbb{E}\left( \cos(\phi_{e,1}[k])\right|\s{N}_1[k] \right)}\nonumber\\
    & = \mathbb{E}\pp{\left.\abs{\s{N}_1[k]} \cdot\mathbb{E}\left( \cos\p{\frac{2\pi k\s{\hat{R}}_1}{d} +  \phi_{\s{N}_1}[k] - \phi_{\s{X}}[k]}\right|\s{N}_1[k] \right)}.
\end{align}
More explicitly, we can write,
    \begin{align}
         &\mathbb{E}{\left[ \abs{\s{N}_1[k]} \cos(\phi_{e,1}[k]) \right]} = \nonumber \\ 
         & \frac{1}{2\pi} \int_{0}^{\infty} \mathrm{d}nn f_{\abs{{\s{N}_1[k]}}}(n)  \int_{-\pi}^{\pi} \mathrm{d}\varphi \mathbb{E} \left[\left.\cos \left( \frac{2\pi k}{d}\hat{\s{R}}_1 + \varphi \right)\right|\abs{\s{N}_1[k]} = n, \phi_{\s{N}_1}[k] = \phi_{\s{X}}[k] + \varphi \right].\label{eqn:totalExp}
\end{align}
Now, note that the inner integral can be written as,
\begin{align}
         &\int_{-\pi}^{\pi} \mathrm{d}\varphi \mathbb{E} \left[\left.\cos \left( \frac{2\pi k}{d}\hat{\s{R}}_1 + \varphi \right)\right|\abs{\s{N}_1[k]} = n, \phi_{\s{N}_1}[k] = \phi_{\s{X}}[k] + \varphi \right]\nonumber\\
         &\qquad=\int_{0}^{\pi} \mathrm{d}\varphi \ \mathbb{E} \left[\left.\cos \left( \frac{2\pi k}{d}\hat{\s{R}}_1 + \varphi \right)\right|\abs{\s{N}_1[k]}=n, \phi_{\s{N}_1}[k] = \phi_{\s{X}}[k] + \varphi \right] + \nonumber \\ 
         &\qquad\quad +\int_{0}^{\pi} \mathrm{d}\varphi \  \mathbb{E} \left[\left.\cos \left( \frac{2\pi k}{d}\hat{\s{R}}_1 + \varphi + \pi \right)\right|\abs{\s{N}_1[k]}=n, \phi_{\s{N}_1}[k] = \phi_{\s{X}}[k] + \varphi + \pi \right].
         \label{eqn:lawOfTotalExpectationInnerIntegral}
    \end{align}

Now, we apply Proposition \ref{lemma:1} on the integrands in \eqref{eqn:lawOfTotalExpectationInnerIntegral}. Using its notation, we define the Gaussian process:  
\begin{align}  
    \s{S}^{(+)} = \s{S}_1 \vert \s{N}_1 [k],  
\end{align}  
where the right-hand side is defined as in~\eqref{eqn:conditionalGaussian}. By \eqref{eqn: 14}, the mean vector of $\s{S}_1 \vert \s{N}_1 [k]$ has a cosine trend, as assumed in Proposition \ref{lemma:1} in \eqref{eqn:35}. Additionally, $\s{S}_1 \vert \s{N}_1 [k]$ is a Gaussian cyclo-stationary process, as described in \eqref{eqn:covarainceMatrix}. The final condition to verify is~\eqref{eq:pos_prob_Cr_assumption}, which is satisfied by Lemma~\ref{lem:positive_prob_each_argmax_conditional}.

Since the conditional distribution of $\hat{\s{R}}_1$ given $\{\abs{\s{N}_1[k]}=n, \phi_{\s{N}_1}[k] = \phi_{\s{X}}[k] + \varphi\}$ matches that of $\hat{\s{R}}^{(+)}$ in \eqref{eqn:27}, and similarly, given $\{\abs{\s{N}_1[k]}=n, \phi_{\s{N}_1}[k] = \phi_{\s{X}}[k] + \varphi+\pi\}$, it matches $\hat{\s{R}}^{(-)}$ in \eqref{eqn:28}, the sum of the integrands on the right-hand side of \eqref{eqn:lawOfTotalExpectationInnerIntegral} equals the left-hand side of \eqref{eqn:lemma1conc}. By Proposition \ref{lemma:1}, this sum is positive for all $\varphi \in [0, \pi]$. Together with \eqref{eqn:totalExp}, this completes the proof of Proposition~\ref{prop:2}.  

\end{proof}

\subsection{Convergence rate in distribution of the Fourier phases} \label{sec:proof_prop_1}

\begin{proposition}[Convergence in distribution of the Fourier phases]\label{prop:1}
Fix $d\in\mathbb{N}$. Then, for any $0\leq k\leq d-1$, 
\begin{align}
    \sqrt{M}\cdot\tan\p{\phi_{\s{\hat{X}}}[k] - \phi_{\s{X}}[k]} \xrightarrow[]{\calD} \s{Q}_k,    \label{eqn:phaseConvergenceInDistribution}
\end{align}
as $M\to\infty$, where $\s{Q}_k$ is defined in \eqref{eqn: A51}.
\end{proposition}

\begin{proof}[Proof of Proposition~\ref{prop:1}]
Recall the definition of $\s{A}_{M,k}, \s{B}_{M,k}$ in \eqref{eqn:AMdef}, \eqref{eqn:BMdef}, of $\s{Q}_k$ in \eqref{eqn: A51}, and of $\phi_{e,i}[k]$ in \eqref{eqn:phaseDifferenceTerm}. Then, following \eqref{eqn:estimatorPhase}, the left-hand-side of \eqref{eqn:phaseConvergenceInDistribution} is given by,
\begin{align}
    \sqrt{M}\cdot\tan\pp{\phi_{\s{\hat{X}}}[k] - \phi_{\s{X}}[k]} = \frac{\frac{1}{\sqrt{M}}\sum_{i=0}^{M-1}\abs{\s{N}_i[k]} \sin \left( \phi_{e,i}[k] \right)}{\frac{1}{M}\sum_{i=0}^{M-1}\abs{\s{N}_i[k]} \cos \left( \phi_{e,i}[k] \right)}\triangleq\frac{\s{A}_{M,k}}{\s{B}_{M,k}}.\label{eqn:arctanExpression0}
\end{align}

Since $\{\s{N}_i\}_{i=0}^{M-1}$ is an i.i.d. sequence of random variables, and because each $\phi_{e,i}$ depends on $\s{N}_i$ solely (in particular, independent of $\s{N}_j$, for $j \neq i$), we have that $\{\abs{\s{N}_i[k]} \sin \left( \phi_{e,i}[k] \right)\}_{i=0}^{M-1}$ and $\{\abs{\s{N}_i[k]} \cos \left( \phi_{e,i}[k] \right)\}_{i=0}^{M-1}$ are two sequences of i.i.d. random variables. Recall the definition of $\mu_{\s{A},k}$, and $\sigma_{\s{A},k}^2$, in \eqref{eqn:muA}, \eqref{eqn:sigmaA}, respectively:
\begin{align}
        \mu_{\s{A},k} &\triangleq \mathbb{E}\left[ \abs{\s{N}_1[k]} \sin (\phi_{e,1}[k]) \right],
        \label{eqn:muAdefenition}\\
        \sigma_{\s{A},k}^2 & \triangleq \s{Var}\p{\abs{\s{N}_1[k]} \sin (\phi_{e,1}[k])},
        \label{eqn:muBdefenition}
\end{align}
which are the mean value and variance of $\s{A}_{M,k}$, as defined in \eqref{eqn:AMdef}. Then, by the CLT:
\begin{align}
        \p{\s{A}_{M,k} - \sqrt{M}\mu_{\s{A},k}} \xrightarrow[]{\calD} \s{A}_k, \label{eqn:B21}
\end{align}
where $\s{A}_k\sim \calN (0,\sigma_{\s{A},k}^2)$. In particular, by Lemma \ref{lemma:A1}, $\mu_{\s{A},k} = 0$.

Next, we analyze the denominator in \eqref{eqn:arctanExpression0}. Specifically, we already saw that $\{\abs{\s{N}_i[k]} \cos \left( \phi_{e,i}[k] \right)\}_{i=0}^{M-1}$ form a sequence of i.i.d. random variables, and thus by the SLLN we have $\s{B}_{M,k}\xrightarrow[]{\s{a.s.}} \mu_{\s{B},k}$, where,
    \begin{align}
        \mu_{\s{B},k} \triangleq \mathbb{E}\left[ \abs{\s{N}_1[k]} \cos (\phi_{e,1}[k]) \right].
    \end{align}
By Proposition~\ref{prop:2}, $\mu_{\s{B},k}>0$. Thus, applying Slutsky's Theorem on the ratio $\frac{\s{A}_{M,k}}{\s{B}_{M,k}}$, we obtain,
    \begin{align}
        \frac{\s{A}_{M,k}}{\s{B}_{M,k}} \xrightarrow[]{\calD} \calN \left(0,\frac{\sigma_{\s{A},k}^2}{\mu_{\s{B},k}^2}\right) = \s{Q}_k \label{eqn:A_M_B_M_Convergence_in_distribution},
    \end{align}
which concludes the proof.
\end{proof}

\subsection{Convergence rate in expectation of the Fourier phases}
\label{sec:convergenceInExpectationOfFourierPhases} 

\begin{proposition}[Convergence rate of the Fourier phases]
\label{prop:convergeneInExpectation} 
Recall the definitions of $\mu_{\s{B},k}$, and $\sigma_{\s{A},k}^2$ in \eqref{eqn:muB}, and \eqref{eqn:sigmaB}, respectively. Assume that $\s{X}[k] \neq 0$, for all $0< k \leq d-1$. Then, as $M \to \infty$,
    \begin{align}
       \lim_{M\to\infty} \frac{\mathbb{E} |\phi_{\s{\hat{X}}}[k] - \phi_{\s{X}}[k]|^2}{1/M} = \frac{\sigma_{\s{A},k}^2}{\mu_{\s{B},k}^2}.
        \label{eqn:B19}
    \end{align}
\end{proposition}

\begin{proof}[Proof of Proposition~\ref{prop:convergeneInExpectation}]
Recall the definitions of $\s{A}_{M,k}$ and $\s{B}_{M,k}$ in \eqref{eqn:AMdef} and \eqref{eqn:BMdef}, respectively, and of $\s{Q}_k$ in \eqref{eqn: A51}. Then, using the phase difference expression in \eqref{eqn:estimatorPhase}, it follows that establishing \eqref{eqn:B19} is equivalent to proving the following:
\begin{align}
        \lim_{M\to\infty}\frac{ \mathbb{E} \pp{\arctan^2 {\left(\frac{1}{\sqrt{M}}\frac{\s{A}_{M,k}}{\s{B}_{M,k}}\right)}}}{ \frac{1}{{M}}\mathbb{E} \pp{ {\s{Q}^2_k}}} = 1, \label{eqn:B20}
    \end{align}
    for every $0 \leq k \leq d-1$. Recall by the definition of $\s{Q}_k$ in \eqref{eqn: A51} that $\mathbb{E} \pp{\s{Q}_k^2} = \sigma_{\s{A},k}^2 / \mu_{\s{B},k}^2$, which is equivalent to the right-hand-side of \eqref{eqn:B19}.

For brevity, we fix $k$, and denote $\s{A}_{M} = \s{A}_{M,k}$, $\s{B}_{M} = \s{B}_{M,k}$, $\mu_{\s{B}} = \mu_{\s{B},k}$, $\sigma_{\s{A}}^2 = \sigma_{\s{A},k}^2$.
Using \eqref{eqn:estimatorPhase} it is clear that,
\begin{align}
    \sqrt{M}\cdot\tan\pp{\phi_{\s{\hat{X}}}[k] - \phi_{\s{X}}[k]} = \frac{\frac{1}{\sqrt{M}}\sum_{i=0}^{M-1}\abs{\s{N}_i[k]} \sin \left( \phi_{e,i}[k] \right)}{\frac{1}{M}\sum_{i=0}^{M-1}\abs{\s{N}_i[k]} \cos \left( \phi_{e,i}[k] \right)}\triangleq\frac{\s{A}_M}{\s{B}_M},\label{eqn:arctanExpression}
\end{align}
It is important to note that the denominator $\s{B}_M$ can be zero with positive probability, implying that the expression in \eqref{eqn:arctanExpression} may diverge with non-zero probability. Therefore, it is necessary to control the occurrence of such events. To this end, $\s{B}_M \xrightarrow[]{\s{a.s.}}  \mu_\s{B}$, by SLLN (see Section \ref{sec:convergenceOfEfNestimator}), where $\mu_\s{B}$ is defined in \eqref{eqn:muB}. Fix $0 < \epsilon < \mu_\s{B}$, and proceed by decomposing as follows:
\begin{align}
    \mathbb{E} \pp{\arctan^2 {\left(\frac{1}{\sqrt{M}}\frac{\s{A}_M}{\s{B}_M}\right)}}& = \mathbb{E} \pp{\arctan^2 {\left(\frac{1}{\sqrt{M}}\frac{\s{A}_M}{\s{B}_M}\right) \mathbbm{1}_{\abs{\s{B}_M} > \epsilon}}}  \nonumber\\ &\quad\quad+\mathbb{E} \pp{\arctan^2 {\left(\frac{1}{\sqrt{M}}\frac{\s{A}_M}{\s{B}_M}\right) \mathbbm{1}_{\abs{\s{B}_M} < \epsilon}}} \label{eqn:splittingArgTanIntoTwoTerms}.
\end{align}

The next lemma shows that the second term at the r.h.s. of \eqref{eqn:splittingArgTanIntoTwoTerms} converges to zero with rate $O(1/M^2)$.
\begin{lem} \label{lem:B5}
    The following inequality holds,
    \begin{align}
        \mathbb{E} \pp{\arctan^2 {\left(\frac{1}{\sqrt{M}}\frac{\s{A}_M}{\s{B}_M}\right) \mathbbm{1}_{\abs{\s{B}_M} < \epsilon}}} \leq \frac{D}{M^2}, \label{eqn:B23}
    \end{align}
    for a finite $D > 0$.
\end{lem}

In addition, we have the following asymptotic relation for the last term in \eqref{eqn:splittingArgTanIntoTwoTerms}.
\begin{lem} \label{lem:B6}
    The following asymptotic relation hold,
        \begin{align}
         \lim_{M\to\infty}\frac{ \mathbb{E} \pp{\arctan^2 {\left(\frac{1}{\sqrt{M}}\frac{\s{A}_M}{\s{B}_M}  \right) \mathbbm{1}_{\abs{\s{B}_M} > \epsilon}}}}{ \frac{1}{{M}}\mathbb{E} \pp{ {\s{Q}^2_k}}} = 1 \label{eqn:arctanWithPositiveEventConvergenceToOne1}.
    \end{align}
\end{lem}
We prove these lemmas below. Substituting \eqref{eqn:B23} and \eqref{eqn:arctanWithPositiveEventConvergenceToOne1} in \eqref{eqn:splittingArgTanIntoTwoTerms}, leads to \eqref{eqn:B20}, and completing the proof of the proposition.

\end{proof}

\begin{proof}[Proof of Lemma \ref{lem:B5}]
Since $\arctan(x)\leq \frac{\pi}{2}$, for any $x\in\mathbb{R}$, we have
\begin{align}
    \mathbb{E} \pp{\arctan^2 {\left(\frac{1}{\sqrt{M}}\frac{\s{A}_M}{\s{B}_M}\right) \mathbbm{1}_{\abs{\s{B}_M} < \epsilon}}}&\leq \frac{\pi^2}{4}\cdot\mathbb{E} \pp{\mathbbm{1}_{\abs{\s{B}_M} < \epsilon}}\\
    &\leq \frac{\pi^2}{4}\cdot\mathbb{P}\p{{\s{B}_M} < \epsilon}\\
    &= \frac{\pi^2}{4}\cdot\mathbb{P}\p{{\s{B}_M} - \mu_{\s{B}} < \epsilon - \mu_{\s{B}}}\\
    & \leq \frac{\pi^2}{4}\cdot\mathbb{P}\p{{\abs{\s{B}_M - \mu_{\s{B}}}} > \mu_{\s{B}} - \epsilon}. \label{eqn:indicatorNearZero}
\end{align}
Let us denote the summand in the denominator in \eqref{eqn:arctanExpression} by $V_i \triangleq \abs{\s{N}_i[k]} \cos \left( \phi_{e,i}[k] \right)$, for $0
\leq i\leq M-1$. Then, we note that,
\begin{align}
    \mathbb{E}(V_i^4) = \mathbb{E}\pp{\abs{\s{N}_i[k]} \cos \left( \phi_{e,i}[k] \right)}^4 \leq \mathbb{E}\pp{\abs{\s{N}_i[k]}}^4 < \infty.
\end{align}
Thus, by Chebyshev's inequality,
\begin{align}
    \mathbb{P}\p{{\abs{\s{B}_M - \mu_{\s{B}}}} > \mu_{\s{B}} - \epsilon} \leq \frac{\mathbb{E}\pp{\s{B_M} - \mu_{\s{B}}}^4}{\p{\mu_{\s{B}} - \epsilon}^4} \label{eqn:chebyshevInequalityFourthMoment}.
\end{align}
Now, by the definition of $\s{B}_M$, we have,
\begin{align}
    \mathbb{E}\pp{\s{B}_M - \mu_{\s{B}}}^4  &= \frac{1}{M^4}\sum_{i,j,k,l=0}^{M-1} \mathbb{E}\pp{\p{V_i - \mu_{\s{B}}}\p{V_j - \mu_{\s{B}}}\p{V_k - \mu_{\s{B}}}\p{V_l - \mu_{\s{B}}}} & \\ 
    &=\frac{1}{M^4} \pp{M\cdot\mathbb{E}\pp{V_1 - \mu_{\s{B}}}^4 + 3M(M-1)\p{\mathbb{E}\pp{V_1 - \mu_{\s{B}}}^2}^2}.
\end{align}
Therefore, it is evident that there exists a constant $D_1$, which depends on the second and fourth moments of $V_1$, such that,
\begin{align}
    \frac{\mathbb{E}\pp{\s{B_M} - \mu_{\s{B}}}^4}{\p{\mu_{\s{B}} - \epsilon}^4} \leq \frac{D_1}{\p{\mu_{\s{B}} - \epsilon}^4 M^2} \label{eqn:fourhMomentBound}.
\end{align}
Thus, plugging \eqref{eqn:chebyshevInequalityFourthMoment} and \eqref{eqn:fourhMomentBound} into \eqref{eqn:indicatorNearZero} leads to,
\begin{align}
     \mathbb{E} \pp{\arctan^2 {\left(\frac{1}{\sqrt{M}}\frac{\s{A}_M}{\s{B}_M}\right) \mathbbm{1}_{\abs{\s{B}_M} < \epsilon}}}&\leq \frac{\pi^2}{4}\cdot\frac{D_1}{\p{\mu_{\s{B}} - \epsilon}^4 M^2} \label{eqn:convergenceInMsquare}.
\end{align}
Thus, the second term at the r.h.s. of \eqref{eqn:splittingArgTanIntoTwoTerms} indeed converges to zero as $1/M^2$. 
\end{proof}

\begin{proof}[Proof of Lemma \ref{lem:B6}]

We analyze the first term at the r.h.s. of \eqref{eqn:splittingArgTanIntoTwoTerms}. We will show that,
    \begin{align}
         1 \leq \lim_{M\to\infty}\frac{ \mathbb{E} \pp{\arctan^2 {\left(\frac{1}{\sqrt{M}}\frac{\s{A}_M}{\s{B}_M}  \right) \mathbbm{1}_{\abs{\s{B}_M} > \epsilon}}}}{ \frac{1}{{M}}\mathbb{E} \pp{ {\s{Q}^2_k}}} \leq \frac{\mu_{\s{B}}^2}{\epsilon^2} \label{eqn:lowerAndUpperBoundOfConvergenceRatio}.
    \end{align}
As this is true for every $\epsilon < \mu_{\s{B}}$, it would imply that,
    \begin{align}
         \lim_{M\to\infty}\frac{ \mathbb{E} \pp{\arctan^2 {\left(\frac{1}{\sqrt{M}}\frac{\s{A}_M}{\s{B}_M}  \right) \mathbbm{1}_{\abs{\s{B}_M} > \epsilon}}}}{ \frac{1}{{M}}\mathbb{E} \pp{ {\s{Q}^2_k}}} = 1 \label{eqn:arctanWithPositiveEventConvergenceToOne}.
    \end{align}

First, due to the monotonicity of $\arctan^2 \p{\cdot}$,
\begin{align}
     \mathbb{E} \pp{\arctan^2 {\left(\frac{1}{\sqrt{M}}\frac{\s{A}_M}{\s{B}_M}  \right) \mathbbm{1}_{\abs{\s{B}_M} > \epsilon}}} \leq  \mathbb{E} \pp{\arctan^2 {\left(\frac{\s{A}_M}{\epsilon \sqrt{M}}  \right)}}. \label{eqn:B40}
\end{align}
We decompose the right-hand-side of \eqref{eqn:B40} into two events, as follows,
\begin{align}
    \nonumber \mathbb{E} \pp{\arctan^2 {\left(\frac{\s{A}_M}{\epsilon \sqrt{M}}  \right)}} 
    = \mathbb{E} & \pp{\arctan^2 {\left(\frac{\s{A}_M}{\epsilon \sqrt{M}}  \right)} \mathbbm{1}_{\abs{\s{A}_M} > \epsilon \sqrt{M}}} 
    \\ & + \mathbb{E} \pp{\arctan^2 {\left(\frac{\s{A}_M}{\epsilon \sqrt{M}}  \right)}\mathbbm{1}_{\abs{\s{A}_M} < \epsilon \sqrt{M}}}. 
    \label{eqn:splittingArgTanIntoTwoTerms2}
\end{align}
By the SLLN, $\s{A}_M/\sqrt{M} \xrightarrow[]{\s{a.s.}}  \mu_\s{A}$ (see Section \ref{sec:convergenceOfEfNestimator}), where $\mu_{\s{A}} = 0$ (Lemma \ref{lemma:A1}). In addition, by Proposition \ref{prop:1}, we have,
\begin{align}
        \s{A}_M \xrightarrow[]{\calD} \mathcal{N}\p{0, \sigma_{\s{A}}^2}, \label{eqn:B41}
\end{align}
by the CLT. 
Then, by arguments similar to those used in Lemma \ref{lem:B5}, the first term on the right-hand side of \eqref{eqn:splittingArgTanIntoTwoTerms2} satisfies:
\begin{align}
    \mathbb{E} & \pp{\arctan^2 {\left(\frac{\s{A}_M}{\epsilon \sqrt{M}}  \right)} \mathbbm{1}_{\abs{\s{A}_M} > \epsilon \sqrt{M}}}\leq \tilde{D}/M^2. \label{eqn:B43}
\end{align}
Namely, the first term at the r.h.s. of \eqref{eqn:splittingArgTanIntoTwoTerms2} converges to zero with rate $O(\frac{1}{M^2})$.

For the last term in the right-hand-side of \eqref{eqn:splittingArgTanIntoTwoTerms2}, we prove the following:
    \begin{align}
         \lim_{M\to\infty}\frac{\mathbb{E} \pp{\arctan^2 {\left(\frac{\s{A}_M}{\epsilon \sqrt{M}}  \right)} \mathbbm{1}_{\abs{\s{A}_M} < \epsilon \sqrt{M}}}}{\frac{1}{M} \mathbb{E} \pp{ {\left(\frac{\s{A}_M}{\epsilon}  \right)}^2 \mathbbm{1}_{\abs{\s{A}_M} < \epsilon \sqrt{M}}}} = 1. \label{eqn:B44}
    \end{align}
Since $[\arctan(x)]^2/x^2 \to 1$ as $x \to 0$, it follows that the Taylor expansion of $\arctan(x)$ around $x = 0$, which holds for $\abs{x} < 1$, and is applicable on the event $\ppp{\s{A}_M < \epsilon \sqrt{M}}$:
\begin{align}
    \mathbb{E} \pp{\arctan^2 {\left(\frac{\s{A}_M}{\epsilon \sqrt{M}}  \right)} \mathbbm{1}_{\abs{\s{A}_M} < \epsilon \sqrt{M}}}
    = \mathbb{E} \pp{\sum_{k=0}^{\infty} \p{-1}^k \frac{ \pp{\frac{1}{\sqrt{M}} \frac{\s{A}_M} {\epsilon}  }^{2k+1}}{2k+1} \mathbbm{1}_{\abs{\s{A}_M}<\epsilon \sqrt{M}}}^2  \label{eqn:TaylorSeriesOfArctan}.
\end{align}
The right-hand-side of \eqref{eqn:TaylorSeriesOfArctan} can be decomposed to,
\begin{align}
    \mathbb{E} & \pp{\sum_{k=0}^{\infty} \p{-1}^k \frac{ \pp{\frac{1}{\sqrt{M}} \frac{\s{A}_M} {\epsilon}  }^{2k+1}}{2k+1} \mathbbm{1}_{\abs{\s{A}_M}<\epsilon \sqrt{M}}}^2  = 
    \frac{1}{M} \mathbb{E} \pp{ {\left(\frac{\s{A}_M}{\epsilon}  \right)}^2 \mathbbm{1}_{\abs{\s{A}_M} < \epsilon \sqrt{M}}} \nonumber \\ &
    +  \sum_{(k_1,k_2)\neq(0,0)}\frac{\p{-1}^{k_1+k_2}}{\p{2k_1 + 1}\p{2k_2 + 1}} \mathbb{E} \pp{\p{\frac{\s{A}_M}{\sqrt{M} \epsilon}}^{\p{2k_1+2k_2 + 2}}\mathbbm{1}_{\abs{\s{A}_M} < \epsilon \sqrt{M}}} \label{eqn:TaylorSeriesOfArctan2}.
\end{align}
Now, since the term at the left-hand-side of \eqref{eqn:TaylorSeriesOfArctan2} as well as the first term at the right-hand-side of \eqref{eqn:TaylorSeriesOfArctan2}, are bounded for every $M$ and converges to zero, then also the last term at the right-hand-side of \eqref{eqn:TaylorSeriesOfArctan2} is bounded for every $M$ and converge to zero as $M \to \infty$. Specifically, we note that the last term converges to zero with rate $1/M^2$, while the first term in the right-hand-side converges to zero with rate $1/M$. Thus, \eqref{eqn:B44} is satisfied. Finally, we have,
\begin{align}
     \lim_{M\to\infty}\frac{\frac{1}{M} \mathbb{E} \pp{ {\left(\frac{\s{A}_M}{\epsilon}  \right)}^2 \mathbbm{1}_{\abs{\s{A}_M} < \epsilon \sqrt{M}}}}{\frac{1}{M} \mathbb{E} \pp{ {\left(\frac{\s{A}_M}{\epsilon}  \right)}^2}} = 1 - \lim_{M\to\infty}\frac{\frac{1}{M} \mathbb{E} \pp{ {\left(\frac{\s{A}_M}{\epsilon}  \right)}^2 \mathbbm{1}_{\abs{\s{A}_M} > \epsilon \sqrt{M}}}}{\frac{1}{M} \mathbb{E} \pp{ {\left(\frac{\s{A}_M}{\epsilon}  \right)}^2}}. \label{eqn:B47}
\end{align}
As the probability of the event $\ppp{\abs{\s{A}_M} > \epsilon \sqrt{M}}$  is $O \p{1/M^2}$, it follows that,
\begin{align}
    \lim_{M\to\infty}\frac{\frac{1}{M} \mathbb{E} \pp{ {\left(\frac{\s{A}_M}{\epsilon}  \right)}^2 \mathbbm{1}_{\abs{\s{A}_M} > \epsilon \sqrt{M}}}}{\frac{1}{M} \mathbb{E} \pp{ {\left(\frac{\s{A}_M}{\epsilon}  \right)}^2}} = 0. \label{eqn:B48}
\end{align}
Then, following \eqref{eqn:B47},\eqref{eqn:B48}, we have,
\begin{align}
    \lim_{M\to\infty}\frac{\frac{1}{M} \mathbb{E} \pp{ {\left(\frac{\s{A}_M}{\epsilon}  \right)}^2 \mathbbm{1}_{\abs{\s{A}_M} < \epsilon \sqrt{M}}}}{\frac{1}{M} \mathbb{E} \pp{ {\left(\frac{\s{A}_M}{\epsilon}  \right)}^2}} = 1. \label{eqn:B49}
\end{align}
By definition $\sigma_{\s{A}}^2 = \mathbb{E} \pp{\s{A}_M^2}$. Therefore substitution \eqref{eqn:B49} into \eqref{eqn:B44} leads to
    \begin{align}
         \lim_{M\to\infty}\frac{\mathbb{E} \pp{\arctan^2 {\left(\frac{\s{A}_M}{\epsilon \sqrt{M}}  \right)} \mathbbm{1}_{\abs{\s{A}_M} < \epsilon \sqrt{M}}}}{\frac{1}{M} \frac{\sigma_{\s{A}}^2}{\epsilon^2}} = 1. \label{eqn:B50}
    \end{align}
Substituting \eqref{eqn:B43}, and \eqref{eqn:B50} into \eqref{eqn:splittingArgTanIntoTwoTerms2} results,
\begin{align}
     \lim_{M\to\infty}\frac{ \mathbb{E} \pp{\arctan^2 {\left(\frac{1}{\sqrt{M}}\frac{\s{A}_M}{\epsilon}  \right) \mathbbm{1}_{\abs{\s{B}_M} > \epsilon}}}}{ \frac{1}{{M}}\mathbb{E} \pp{ {\s{Q}^2_k}}} = \frac{\mu_{\s{B}}^2}{\epsilon^2}, \label{eqn:B51}
\end{align}
where $\mathbb{E} \pp{\s{Q}_k^2} = \sigma_{\s{A}}^2 / \mu_{\s{B}}^2$. Then, substituting \eqref{eqn:B51} into \eqref{eqn:B40} results,
    \begin{align}
        \lim_{M\to\infty}\frac{ \mathbb{E} \pp{\arctan^2 {\left(\frac{1}{\sqrt{M}}\frac{\s{A}_M}{\s{B}_M}  \right) \mathbbm{1}_{\abs{\s{B}_M} > \epsilon}}}}{ \frac{1}{{M}}\mathbb{E} \pp{ {\s{Q}^2_k}}} \leq \frac{\mu_{\s{B}}^2}{\epsilon^2} \label{eqn:B52}.
    \end{align}
which proves the upper bound in \eqref{eqn:lowerAndUpperBoundOfConvergenceRatio}.

Similarly, since $\s{B}_M \xrightarrow[]{\s{a.s.}} \mu_{\s{B}}$, for any $\epsilon_2 > 0$, we have,
    \begin{align}
     \lim_{M\to\infty} \mathbb{E} \pp{\left(\frac{\s{A}_M}{\s{B}_M} \right)^2 \mathbbm{1}_{\ppp{{\s{B}_M} > \epsilon}}} &\geq
     \lim_{M\to\infty} \mathbb{E} \pp{\left(\frac{\s{A}_M}{\s{B}_M} \right)^2 \mathbbm{1}_{\ppp{{\s{B}_M} > \epsilon} \wedge \ppp{{\s{B}_M} < \mu_{\s{B}} + \epsilon_2}}}\\
     &\geq \frac{\sigma_A^2}{(\mu_{\s{B}} + \epsilon_2)^2} \label{eqn:lowerBoundOnArctan}.
    \end{align}
Since \eqref{eqn:lowerBoundOnArctan} is true for every $\epsilon_2 > 0$, we get the lower bound in \eqref{eqn:lowerAndUpperBoundOfConvergenceRatio}, which concludes the proof of \eqref{eqn:arctanWithPositiveEventConvergenceToOne}. 
\end{proof}

\subsection{Proof of Theorem~\ref{thm:1}}\label{thm:proofs1}
\paragraph{Convergence of the Fourier magnitudes.}
We start with the convergence of the estimator's magnitudes. Recall the definition of $\phi_{e,i}[k]$ in \eqref{eqn:phaseDifferenceTerm}. According to \eqref{eqn:strongLLN}, we have,
\begin{align}
   \abs{\s{\hat{X}}[k] e^{-j\phi_{\s{X}}[k]}}  
   \xrightarrow[]{\s{a.s.}} 
    \Bigl\lvert\mathbb{E} \pp{ \abs{\s{N}_1[k]} \cos\p{\phi_{e,1}[k]} } + j \mathbb{E} \pp{ \abs{\s{N}_1[k]} \sin\p{\phi_{e,1}[k]} } \Bigr\rvert, \label{eqn:C1}
\end{align}
Clearly, $\abs{e^{-j\phi_{\s{X}}[k]}} = 1$. By Lemma \ref{lemma:A1},
\begin{align}   
    \mu_{\s{A},k} =\mathbb{E} \left[ \abs{\s{N}_1[k]} \sin\p{\phi_{e,1}[k]} \right] = 0. \label{eqn:C2}
\end{align}
By  Proposition \ref{prop:2}, 
\begin{align}
    \mu_{\s{B},k} = \mathbb{E} \pp{ \abs{\s{N}_1[k]} \cos\p{\phi_{e,1}[k]} } > 0. \label{eqn:C3}
\end{align}
Combining \eqref{eqn:C1}, \eqref{eqn:C2}, and \eqref{eqn:C3} proves the convergence of the estimator's magnitudes of \eqref{eqn:magnitudeConvergenceAsymptoticM}.

\paragraph{Convergence of the Fourier phases.}
Next, we prove the Fourier phases convergence of Theorem~\ref{thm:1}, starting with \eqref{eqn:FirstRes}. To this end, recall \eqref{eqn:estimatorPhase} 
\begin{align}
    \phi_{\s{\hat{X}}}[k] - \phi_{\s{X}}[k] =  \arctan \left( \frac{\sum_{i=0}^{M-1}\abs{\s{N}_i[k]} \sin \left( \phi_{e,i}[k] \right)}{\sum_{i=0}^{M-1}\abs{\s{N}_i[k]} \cos \left( \phi_{e,i}[k] \right)} \right ),
\end{align}
Using the continuous mapping theorem, it is evident that it suffices to prove that,
\begin{align}
    \frac{\sum_{i=0}^{M-1}\abs{\s{N}_i[k]} \sin \left( \phi_{e,i}[k] \right)}{\sum_{i=0}^{M-1}\abs{\s{N}_i[k]} \cos \left( \phi_{e,i}[k] \right)}\xrightarrow[]{\s{a.s.}} 0.
\end{align}
This, however, follows by applying the SLLN,
\begin{align}
   \frac{\sum_{i=0}^{M-1}\abs{\s{N}_i[k]} \sin \left( \phi_{e,i}[k] \right)}{\sum_{i=0}^{M-1}\abs{\s{N}_i[k]} \cos \left( \phi_{e,i}[k] \right)}
    \xrightarrow[]{\s{a.s.}} \frac{\mu_{\s{A},k}}{\mu_{\s{B},k}},
    \label{eqn:lawOfLargeNumberPhaseRatio}
\end{align}
where $\mu_{\s{A},k} \triangleq \mathbb{E}\left[ \abs{\s{N}_1[k]} \sin (\phi_{e,1}[k]) \right]$ and $\mu_{\s{B},k} \triangleq \mathbb{E}\left[ \abs{\s{N}_1[k]} \cos (\phi_{e,1}[k]) \right]$, defined in \eqref{eqn:muA}, and \eqref{eqn:muB}, respectively. By Lemma \ref{lemma:A1}, $\mu_{\s{A},k} = 0$, while by Proposition \ref{prop:2}, we have that $\mu_{\s{B},k} > 0$, and thus their ratio converges a.s. to zero by the continuous mapping theorem. Thus, we proved that $\phi_{\s{\hat{X}}}[k] \xrightarrow[]{\s{a.s.}} \phi_{\s{X}}[k]$.

Finally, we prove the convergence rate, given in \eqref{eqn:asymptoticComnvergenceOfPhases}. According to Proposition \ref{prop:convergeneInExpectation}, we have,
    \begin{align}
       \lim_{M\to\infty} \frac{\mathbb{E} |\phi_{\s{\hat{X}}}[k] - \phi_{\s{X}}[k]|^2}{1/M} = \frac{\sigma_{\s{A},k}^2}{\mu_{\s{B},k}^2}, \label{eqn:app_CC7}
    \end{align}
which completes the proof of the Theorem.

\begin{remark}
    Note that the above result implies that $C_k$ in \eqref{eqn:asymptoticComnvergenceOfPhases} is given by,
\begin{align}
    C_k \triangleq \frac{\sigma_{\s{A},k}^2}{\mu_{\s{B},k}^2} = \frac{\mathbb{E}\p{\left[ \abs{\s{N}_1[k]} \sin(\phi_{e,1}[k]) \right]^2}}{\p{\mathbb{E}{\left[ \abs{\s{N}_1[k]} \cos(\phi_{e,1}[k]) \right]}}^2}.
    \label{eqn:constantCk}
\end{align}
\end{remark}

\section{High-dimensional argmax asymptotics} \label{sec:preliminariesToTheorem2}
In this section, we present a key proposition that plays a central role in the proof of Theorem~\ref{thm:2}.
\begin{proposition}[High-dimensional argmax asymptotics] \label{prop:3}
Let $\s{S}\sim\calN(\s{\mu},\s{\Sigma})$ be a $d$-dimensional Gaussian random vector, with mean $\s{\mu}$ and a covariance matrix $\s{\Sigma}$. Assume that $\abs{\s{\Sigma}_{ij}} = \rho_{\abs{{i-j}}}$, where $\{\rho_\ell\}_{\ell\in\mathbb{N}}$ is a sequence of real-valued numbers such that $\rho_0 = 1$, $\rho_\ell < 1$, and $\rho_\ell \log \ell\rightarrow 0$, as $\ell\to\infty$. Assume also that $\sqrt{\log d} \cdot \max_{1\leq i\leq d}{\abs{\mu_i}} \rightarrow 0$, as $d \rightarrow \infty$, and let $\hat{\s{R}} \triangleq \argmax{\left\{\s{S}_0, \s{S}_1,\ldots,\s{S}_{d-1}\right\}}$. Then, for a bounded deterministic function $f: \{0,1,\ldots,d-1\} \rightarrow \mathbb{R}$, we have,
    \begin{align}   
        \lim_{d\to\infty}\mathbb{E} [{f(\hat{\s{R}})}]  - \frac{\sum_{r=0}^{d-1}f(r)e^{\mu_{r} a_d }}{\sum_{r=0}^{d-1}e^{\mu_{r} a_d }}= 0,\label{eqn:Lemma2}
    \end{align}
where $a_d \triangleq \sqrt{2\log d}$.
\end{proposition}

The proof of Proposition~\ref{prop:3} is based on an auxiliary result, which we prove in Section~\ref{sec:proof_lem_4}. To state this result, we introduce some additional notation. Let $\s{S}(r)$, for $r \in \{0,1,\ldots,d-1\}$, be a discrete stochastic process. We define the function $h^{(\alpha)}(r)$ as follows,
   \begin{align}
        h^{(\alpha)}(r) \triangleq \s{S}(r) + \alpha f(r), \label{eqn: 41}
    \end{align}
where $f(r)$ is a bounded deterministic function, and $\alpha \in \mathbb{R}$. We further define,
    \begin{align}
        \s{M}_d(\alpha) \triangleq \max_{r} h^{(\alpha)}(r)  ,\label{eqn: 49}
    \end{align}
and
    \begin{align}
        \hat{\s{R}}(\alpha) \triangleq \argmax_{r}h^{(\alpha)}(r).   \label{eqn: 50}
    \end{align}
Note that $\s{M}_d(a)$ and $\hat{\s{R}}(a)$ are random variables. Finally, we denote $\hat{\s{R}}\triangleq\hat{\s{R}}(0)$. We have the following result, which is proved in Appendix \ref{sec:proof_lem_4}.

\begin{lem}\label{lemma:4}
The following holds,
    \begin{align}
        \mathbb{E}[f(\hat{\s{R}})] = \left.\frac{\mathrm{d}}{\mathrm{d}\alpha}\mathbb{E} [\s{M}_d(\alpha)]\right|_{\alpha=0} . 
        \label{eqn:expectedValueAndDerivativeOfMaximum}
    \end{align} 
\end{lem}
Lemma \ref{lemma:4} implies that finding the expected value of $f(\hat{\s{R}})$ is related directly to the derivative of the expected value of the maximum around zero. Thus, the problem of finding the expected value of $f(\hat{\s{R}})$ is related to finding the expected value of the maximum of the stochastic process. In our case, $\s{S}$ will be a Gaussian vector with mean given by \eqref{eqn: 14} and a covariance matrix given by \eqref{eqn:covarainceMatrix}. Thus, our goal now is to find the expected value of the maximum of $\s{S}$. For this purpose, we will recall some well-known results on the maximum of Gaussian processes. 

It is known that for an i.i.d. sequence of normally distributed random variables $\{\xi_n\}$, the asymptotic distribution of the maximum $\s{M}_n \triangleq \max\{\xi_1, \xi_2, ..., \xi_n\}$ is the Gumbel distribution, i.e., for any $x\in\mathbb{R}$,
    \begin{align}
         \mathbb{P}\pp{a_n(\s{M}_n - b_n) \leq x} \to e^{-e^{(-x)}},
    \end{align}
as $n\to\infty$, where,
    \begin{align}
         a_n \triangleq \sqrt{2\log n} \label{eqn:a_n}
    \end{align}
and,
    \begin{align}
         b_n \triangleq \sqrt{2\log n} - \frac{1}{2} \frac{\log{\log{n}} + \log {4\pi}}{\sqrt{2\log n}}.
         \label{eqn:b_n}
    \end{align}
It turns out that the above convergence result remains valid even if the sequence $\ppp{{\xi_n}}$ is not independent and normally distributed. Specifically, as shown in \cite[Theorem 6.2.1]{leadbetter2012extremes}, a similar result holds for Gaussian random variables $\ppp{{\xi_n}}$ with a covariance matrix that decays such that $\lim_{n\to\infty} \rho_n \cdot \log{n} = 0$, and with a mean vector whose maximum value decays faster than $\lim_{n\to\infty} \max_{0\leq m \leq n-1} \abs{\mu_m} \cdot \sqrt{\log{n}} = 0$. These conditions precisely match those specified in Theorem~\ref{thm:2}. 

\subsection{Proof of Lemma~\ref{lemma:4}} \label{sec:proof_lem_4}

The proof technique of Lemma~\ref{lemma:4} is similar to the technique used in \cite{pimentel2014location, lopez2016location}, but with a non-trivial adaption to the discrete case. To prove this lemma, we will first establish a deterministic counterpart of \eqref{eqn:expectedValueAndDerivativeOfMaximum}. Specifically, we define,
    \begin{align}
        h^{(\alpha)}(r) \triangleq X(r) + \alpha f(r),\label{eqn: 65}
    \end{align}
where $r \in \{0,1,\ldots, d-1\}$. The functions $X: \{0,1,\ldots, d-1\} \rightarrow \mathbb{R}$, and $f: \{0,1,\ldots, d-1\} \rightarrow \mathbb{R}$ are assumed bounded and deterministic. We further assume that $X$ is injective, i.e., for $z_i \neq z_j$, we have $X(z_i) \neq X(z_j)$. Define,
    \begin{align}
        s(\alpha) \triangleq \max_{r}\{h^{(\alpha)}(r)\}, \label{eqn: 66}
    \end{align}
and note that $s(\alpha)$ is well-defined over the supports of $X$ and $f$, and it is a continuous function of $\alpha$ around $\alpha=0$. Finally, we let,
    \begin{align}
        Z^{(\alpha)}_{\max} \triangleq \argmax_{r}\{h^{(\alpha)}(r)\}.\label{eqn: 67}
    \end{align}
We have the following result.
\begin{lem}\label{corollary:2}
The following relation holds,
    \begin{align}
        \left.\frac{\mathrm{d}}{\mathrm{d}\alpha}s(\alpha) \right|_{\alpha=0} = f(Z^{(0)}_{\max}).
    \end{align}
\end{lem}
\begin{proof} [Proof of Lemma~\ref{corollary:2}]
Note that,
    \begin{align}
        \left.\frac{\mathrm{d}}{\mathrm{d}\alpha}s(\alpha) \right|_{\alpha=0}&=\lim_{\alpha \rightarrow 0} \frac{s(\alpha) - s(0)}{\alpha}\nonumber\\
        &= \lim_{\alpha \to 0} \frac{\max_r [X(r) + \alpha f(r) ] - \max_r X(r) } {\alpha}.\label{eqn: 69}
    \end{align}
By the definition of $Z^{(\alpha)}_{\max}$, we have,
    \begin{align}
        \max_r[X(r) + \alpha f(r)] =  X(Z^{(\alpha)}_{\max}) + \alpha f(Z^{(\alpha)}_{\max}),
        \label{eqn: 70}
    \end{align}
and
    \begin{align}
        \max_r X(r) =  X(Z^{(0)}_{\max}).\label{eqn: 71}
    \end{align}
Now, the main observation here is that for a sufficiently small value of $\alpha$ around zero, we must have that ${Z^{(\alpha)}_{\max}}$ and ${Z^{(0)}_{\max}}$ equal because $Z^{(0)}_{\max}$ can take discrete values only, and it is unique. Thus, for $ \alpha \cdot\max_r{\abs{f(r)}} < \min_{i\neq j}{\abs{X(z_i) - X(z_j)}}$, we have,
    \begin{align}
        {Z^{(\alpha)}_{\max}} = {Z^{(0)}_{\max}} .\label{eqn: 72} 
    \end{align}
Combining \eqref{eqn: 69}--\eqref{eqn: 72} yields,
    \begin{align}
        \left.\frac{\mathrm{d}}{\mathrm{d}\alpha}s(\alpha) \right|_{\alpha=0} &= 
        \lim_{\alpha \rightarrow 0} \frac{X(Z^{(\alpha)}_{\max}) + \alpha f(Z^{(\alpha)}_{\max}) - X(Z^{(0)}_{\max})}{\alpha} \\
        &= f{(Z^{(0)}_{\max})},
    \end{align}
which concludes the proof.
\end{proof}

We are now in a position to prove Lemma~\ref{lemma:4}. Similarly to the deterministic case, we define the random function,
    \begin{align}
        h^{(\alpha)}(r) = \s{S}(r) + \alpha f(r),\label{eqn: 74}
    \end{align}
where $\s{S}: \{0,1,\ldots, d-1\} \to \mathbb{R}$ is a discrete stochastic process, and $f$ is a deterministic function. We assume that $\s{S}$ has a continuous probability distribution without any single point with a measure greater than 0. Using Lemma~\ref{corollary:2}, for each realization of $\s{S}(r)$, such that $\s{S}(r)$ is injective, we have,
     \begin{align}
        f(\hat{\s{R}}) = \left.\frac{\mathrm{d}}{\mathrm{d}\alpha}\ \s{M}_d(\alpha) \, \right|_{\alpha=0} \label{eqn:expeted_derivative_probabilistic_case}.
    \end{align}
Under the assumption above of $\s{S}(r)$, the measure of the set of events that $\s{S}$ is not injective is zero. Therefore, the fact that $\frac{\s{M}_d(\alpha) - \s{M}_d(0)}{\alpha}$ is bounded (see, \eqref{eqn:M_d_derivative_bound}) and \eqref{eqn:expeted_derivative_probabilistic_case}, imply that,
     \begin{align}
        \mathbb{E}[f(\hat{\s{R}})] &= \int{f(\hat{\s{R}}) \mathrm{d}\mu}\nonumber\\
        &= \int{\left.\frac{\mathrm{d}}{\mathrm{d}\alpha}\ \s{M}_d(\alpha) \, \right|_{\alpha=0} \mathrm{d}\mu}\nonumber\\
        &= \int{\lim_{\alpha\to 0} \pp{\frac{\s{M}_d(\alpha) - \s{M}_d(0)}{\alpha}} \mathrm{d}\mu} 
        \nonumber\\ & = \lim_{\alpha\to 0} \int{ \pp{\frac{\s{M}_d(\alpha) - \s{M}_d(0)}{\alpha}} \mathrm{d}\mu} \nonumber\\
        &= \left.\frac{\mathrm{d}}{\mathrm{d}\alpha}\ \mathbb{E} \s{M}_d(\alpha) \, \right|_{\alpha=0}, \label{eqn:96}
    \end{align}
which concludes the proof.

\subsection{Proof of Proposition~\ref{prop:3}}
Conditioned on $\s{N}[k]$, the Gaussian vector $\s{S}$  (see, \eqref{eqn: 14} and \eqref{eqn:covarainceMatrix}) can be represented as,
    \begin{align}
        \s{S}\vert\s{N}[k] = \s{Z} + \s{\mu},
    \end{align} 
where $\s{Z}$ is a zero mean Gaussian random vector with covariance matrix given by \eqref{eqn:covarainceMatrix} and $\s{\mu}$ is given by \eqref{eqn: 14}. Define,
    \begin{align}
        h^{(\alpha)}(r) \triangleq \s{Z}(r) + \s{\mu}(r) + \alpha f(r),
    \end{align} 
where we use the same notations as in Lemma \ref{lemma:4}. 
Then, using Lemma \ref{lemma:4},
    \begin{align}
        \mathbb{E}[f(\hat{\s{R}})] = \left.\frac{\mathrm{d}}{\mathrm{d}\alpha}\mathbb{E} \s{M}_d(\alpha)\right|_{\alpha=0},
    \end{align} 
where $\s{M}_d(\alpha) = \max_r\left\{\s{Z}(r) + \mu(r) + \alpha f(r) \right\}$. Therefore, our goal is now to find the derivative of $\mathbb{E} \s{M}_d(\alpha)$.

Using \cite[Theorem 6.2.1]{leadbetter2012extremes}, under the assumptions of Proposition~\ref{prop:3}, for a sufficiently small value of $\alpha$ such that $ \lim_{d\to\infty} \abs{\alpha} \max_r{\abs{f(r)}} \cdot \sqrt{\log d} = 0$, we have for any $x\geq0$,
    \begin{align}
         \lim_{d\to\infty}\mathbb{P}\pp{a_d(\s{M}_d(\alpha) - b_d - m_d^\star(\alpha)) \leq x} = e^{-e^{(-x)}}, \label{eqn:asymptoticGumbel}
    \end{align}
where $a_d$ and $b_d$ are given in \eqref{eqn:a_n} and \eqref{eqn:b_n}, respectively, and 
    \begin{align}
         m_d^\star(\alpha) \triangleq a_d^{-1} \log {\left( d^{-1} \sum_{i=0}^{d-1} {e^{a_d (\mu_i+\alpha f(i))}} \right)}.
    \end{align}
For brevity, we denote,
    \begin{align}
         \s{T}_d(\alpha) \triangleq a_d\cdot[\s{M}_d(\alpha) - b_d - m_d^\star(\alpha)],\label{eqn:T_d_defenition}
    \end{align}
and we note that,
    \begin{align}
         \s{T}_d(\alpha) - \s{T}_d(0) = a_d[(\s{M}_d(\alpha) - m_d^\star(\alpha)) - (\s{M}_d(0) - m_d^\star(0))] \label{eqn:T_d},
    \end{align}
and so,
    \begin{align}
         \Delta_d(\alpha)&\triangleq\frac{1}{a_d}\frac{\s{T}_d(\alpha) - \s{T}_d(0)}{\alpha} = \frac{\s{M}_d(\alpha) - \s{M}_d(0)}{\alpha} - \frac{m_d^\star(\alpha) - m_d^\star(0)}{\alpha} \label{eqn:T_d_derivative},
    \end{align}
for any $\alpha\neq0$. The following result shows $\Delta_d(\alpha)$ converges zero in the $\mathcal{L}^1$ sense.
\begin{lem}\label{corollary:1} For any $\alpha\neq0$,
    \begin{align}
         \lim_{d\to\infty}\abs{\Delta_d(\alpha)}=0\label{eqn:T_d_L1_convergence},
    \end{align}    
    i.e., $\Delta_d(\alpha)\xrightarrow[]{\mathcal{L}^1} 0$, as $d\to\infty$.
\end{lem}

\begin{proof} [Proof of Lemma~\ref{corollary:1}]
To prove \eqref{eqn:T_d_L1_convergence}, we will first show that $\Delta_d(\alpha)$ converges to zero in probability. Because $\Delta_d(\alpha)$ is uniformly integrable, this is sufficient for the desired $\mathcal{L}^1$ convergence above. Specifically, recall from \eqref{eqn:asymptoticGumbel} that $\s{T}_d(\alpha)$ converges in distribution to the Gumbel random variable $\s{Gum}$ with location zero and unit scale, i.e., $\s{T}_d(\alpha) \xrightarrow[]{\calD} \s{Gum}$, as $d\to\infty$. Furthermore, it is clear that $\frac{1}{a_d} = \frac{1}{\sqrt{2 \log d}} \to 0$, as $d \to \infty$. Thus, Slutsky's theorem \cite{slutsky1925stochastische} implies that,
     \begin{align}
         \frac{\s{T}_d(\alpha)}{a_d} \xrightarrow[]{\calD}0.
    \end{align}
It is known that convergence in distribution to a constant implies also convergence in probability to the same constant \cite{durrett2019probability}, and thus,    
    \begin{align}
         \frac{\s{T}_d(\alpha)}{a_d} \xrightarrow[]{\calP} 0 \label{eqn:T_d_convergence_in_probability}.
    \end{align}
Therefore, the above result together with the continuous mapping theorem \cite{durrett2019probability} implies that,
    \begin{align}
         \Delta_d(\alpha) \xrightarrow[]{\calP} 0 \label{eqn:T_d_derivative_convergence_in_probability},
    \end{align}
for every $\alpha\neq0$.

Next, we show that $\Delta_d(\alpha)$ is bounded with probability one. Indeed, by the definition of $\s{M}_d(\alpha)$ in \eqref{eqn: 49}, we have,
    \begin{align}
         \abs{\frac{\s{M}_d(\alpha) - \s{M}_d(0)}{\alpha}} \leq \max_{0 \leq r \leq d-1} \abs{f(r)}<\mathsf{C} < \infty, \label{eqn:M_d_derivative_bound}
    \end{align}
for some $\s{C}>0$, where we have used the fact that $f$ is bounded. Furthermore, note that,
    \begin{align}
        \frac{\mathrm{d}}{\mathrm{d}\alpha} m_d^\star(\alpha) = \frac{\sum_{i=0}^{d-1}{f(i)\exp\{a_d (\mu_{i}+\alpha f(i))}}{\sum_{i=0}^{d-1}{\exp\{a_d (\mu_{i}+\alpha f(i))}\}}
        \label{eqn: 113},
    \end{align}   
which is bounded because,
    \begin{align}
        \abs{\frac{\sum_{i=0}^{d-1}{f(i)\exp\{a_d (\mu_{i}+\alpha f(i))}}{\sum_{i=0}^{d-1}{\exp\{a_d (\mu_{i}+\alpha f(i))}\}}} \leq \max_{0 \leq r \leq d-1} \abs{f(r)}<\mathsf{C} < \infty.
        \label{eqn: 106}
    \end{align}   
Combining  \eqref{eqn:T_d_derivative}, \eqref{eqn:M_d_derivative_bound} and \eqref{eqn: 106}, leads to,
    \begin{align}
         \abs{\Delta_d(\alpha)} &\leq \abs{\frac{\s{M}_d(\alpha) - \s{M}_d(0)}{\alpha}} + \abs{\frac{m_d^\star(\alpha) - m_d^\star(0)}{\alpha}} \leq 2\max_{0 \leq r \leq d-1} \abs{f(r)} < \infty  \label{eqn:T_d_derivative_bounded}.
    \end{align}
Now, since $\Delta_d(\alpha)$ is bounded, it is also uniformly integrable, and thus when combined with \eqref{eqn:T_d_derivative_convergence_in_probability} we may conclude that,
\begin{align}
    \Delta_d(\alpha) \xrightarrow[]{\mathcal{L}^1} 0,
\end{align}
as claimed.
\end{proof}

We continue with the proof of Proposition~\ref{prop:3}. First, we show that,
    \begin{align}
         \lim_{d\to \infty}\lim_{\alpha\to 0} \mathbb{E}\pp{\Delta_d(\alpha)} = \lim_{\alpha\to 0}\lim_{d\to\infty} \mathbb{E}\pp{\Delta_d(\alpha)}\label{eqn:T_d_derivatives_order_exchange}.
    \end{align}
Indeed, note that,
    \begin{align}
         \lim_{d\to \infty}\lim_{\alpha\to 0} \mathbb{E}\pp{\Delta_d(\alpha)} = \lim_{d\to \infty} \lim_{\alpha\to 0}\int{ \pp{\frac{\s{T}_d(\alpha) - \s{T}_d(0)}{\alpha}} \mathrm{d}\mu},
    \end{align}
where $\mathrm{d}\mu$ is the probability measure associated with $\s{T}_d$. From \eqref{eqn:T_d_derivative_bounded} we know that $\Delta_d(\alpha)$ is bounded. Thus, applying the dominated convergence theorem, we obtain,
    \begin{align}
          \lim_{d\to \infty} \lim_{\alpha\to 0} \int{ \pp{\frac{\s{T}_d(\alpha) - \s{T}_d(0)}{\alpha}} \mathrm{d}\mu} = \lim_{d\to \infty} \int{\lim_{\alpha\to 0} \pp{\frac{\s{T}_d(\alpha) - \s{T}_d(0)}{\alpha}} \mathrm{d}\mu}.\label{eqn:afterDCT}
    \end{align}
Since the integral at the right-hand-side of \eqref{eqn:afterDCT} is finite and bounded for each value of $\alpha$, and for each value of $d$, the order of the limits can be exchanged, thus leading to \eqref{eqn:T_d_derivatives_order_exchange}. Therefore, from \eqref{eqn:T_d_derivative} and \eqref{eqn:T_d_derivatives_order_exchange}, we have,
    \begin{align}
         \lim_{\alpha\to 0}\lim_{d\to\infty} \mathbb{E}\pp{\Delta_d(\alpha)} &= \lim_{d\to\infty}\lim_{\alpha\to 0} \pp{\frac{\mathbb{E}[\s{M}_d(\alpha) - \s{M}_d(0)]}{\alpha} - \frac{[m_d^\star(\alpha) - m_d^\star(0)]}{\alpha}}\label{eqn:lfd} \\
         &=\lim_{d\to\infty}\left. \pp{\frac{\mathrm{d}}{\mathrm{d}\alpha}\mathbb{E} \s{M}_d(\alpha)
         - \frac{\mathrm{d}}{\mathrm{d}\alpha} m_d^\star(\alpha)}\right|_{\alpha=0}\label{eqn:T_d_derivatives}.
    \end{align}
Now, Lemma~\ref{corollary:1} implies that the left-hand-side of \eqref{eqn:lfd} nullifies, and thus,
    \begin{align}
        \lim_{d\to\infty}\left. \pp{\frac{\mathrm{d}}{\mathrm{d}\alpha}\mathbb{E} \s{M}_d(\alpha)
         - \frac{\mathrm{d}}{\mathrm{d}\alpha} m_d^\star(\alpha)}\right|_{\alpha=0} = 0 \label{eqn:maximumDifference}.
    \end{align}
Finally, combining \eqref{eqn: 113} and \eqref{eqn:maximumDifference}, we obtain \eqref{eqn:Lemma2}, which concludes the proof.

\section{Proof of Theorem~\ref{thm:2}}\label{sec:proof-of-thm-2}
\textit{Remark on notation}. In this section, we omit the dependence on $0 \leq i \leq M-1$, where this is clear from the context, e.g., $\s{N}_i = \s{N}$ and $\s{\hat{R}}_i = \s{\hat{R}}$. In addition, to streamline notation we often omit the explicit dependence on the signal length and write $x$ (and the associated quantities) in place of $x^{(d)}$ whenever the dimension $d$ is understood from the context.

\subsection{Notations and auxiliary results}
First, we introduce notation and present several auxiliary results that support Theorem~\ref{thm:2}.
Recall the definition of $\widetilde{\s{X}}_k[\ell]$ from \eqref{eqn:XtildeDefenition}. We define:
\begin{align}
    \sigma_{k}^2 = \frac{\sigma^2}{2} \cdot \sum_{\ell = 0}^{d-1} |\widetilde{\s{X}}_k[\ell]|^2, \label{eqn:E1}
\end{align}
which corresponds to the diagonal entries of $\s{\Sigma}_k[r,s]$ in \eqref{eqn:covarainceMatrix}.
In addition, recall the vector $\s{S}$ as defined in \eqref{equ:maxGaussDef}, and introduce the normalized vector:
\begin{align}
    \Tilde{\s{S}}_k = \s{S}/\sigma_k, \label{eqn:E12}
\end{align}
where $\sigma_k$ is given in \eqref{eqn:E1}. Then, by Lemma~\ref{lemma:conditioning}, we have: 
    \begin{align}
        \s{\widetilde{S}}_k\vert\s{N}[k] \sim \calN (\s{\widetilde{\mu}}_{k}, \s{\widetilde{\Sigma}}_{k}), \label{eqn:conditionalGaussianNormalized}
    \end{align}
where the mean and covariance are given by:
    \begin{align}
        \widetilde{\mu}_{k}[r] &\triangleq 2 \sigma_k^{-1}\abs{\s{X}[k]} \abs{\s{N}_i[k]} \cos \left( \frac{2\pi kr}{d} + \phi_{\s{N}_i}[k] - \phi_{\s{X}}[k] \right)
        \label{eqn:E4}
    \end{align}
    for $0\leq r \leq d-1$, and
    \begin{align}
        \s{\widetilde{\Sigma}}_{k}[r,s] & \triangleq \frac{\sum_{\ell = 0}^{d-1} |\widetilde{\s{X}}_k[\ell]|^2 \cos \left( \frac{2\pi \ell}{d}(r-s) \right)}{\sum_{\ell = 0}^{d-1} |\widetilde{\s{X}}_k[\ell]|^2}, \label{eqn:E5}
    \end{align}
    for $0\leq r,s\leq d-1$. Note in particular that normalizing $\s{S}$ by $\sigma_k$ ensures that the diagonal entries of $\s{\widetilde{\Sigma}}_k$ are equal to one.

Recall that we assume the template vector is normalized, i.e.,
$\sum_{\ell=0}^{d-1} |\s{X}[\ell]|^2 = 1$. Under this assumption, we have the following lemma.
\begin{lem} \label{lemma:E1}
    Suppose the conditions of Assumption~\ref{assump:1} hold. Then, for all $k \in \mathbb{N}$, the following limits hold:
    \begin{align}   
        \lim_{d\to\infty}\sum_{\ell=0}^{d-1} \abs{2\abs{\s{X}[\ell]}^2 - \abs{\s{\widetilde{X}}_k[\ell]}^2} = 0, \label{eqn:xTildeConvergenceToX0}
    \end{align}
    and,
    \begin{align}
        \lim_{d\to\infty} \sigma_{k}^2 = \sigma^2. \label{eqn:E7}
    \end{align}
\end{lem}

\begin{proof}[Proof of Lemma~\ref{lemma:E1}]

From the definition of $\widetilde{\s{X}}_k[\ell]$ in \eqref{eqn:XtildeDefenition}, we have:
\begin{align}
    2\abs{\s{X}[\ell]}^2 - \abs{\s{\widetilde{X}}_k[\ell]}^2 = \begin{cases}
                 2\abs{\s{X}[\ell]}^2  & \s{if}  \ell = k, d-k, \\
             \abs{\s{X}[\ell]}^2  & \s{if}  \ell = 0, d/2, \\
                0  & \s{otherwise}.
          \end{cases} \label{eqn:E14}
\end{align}
According to Assumption~\ref{assump:1}, we have,
\begin{align}
        \lim_{d\to\infty} \ppp{{\underset{0 < k \leq d-1} \max\ppp{ \abs{\s{X}[k]}}} \cdot \sqrt{\log d}}  = 0, \label{eqn:E15}
    \end{align} 
and $\s{X}[0] = 0$,
Recall that the template vector is normalized, i.e., $\sum_{\ell=0}^{d-1} |\s{X}[\ell]|^2 = 1$, which implies $\abs{\s{X}[\ell]}^2 \leq \abs{\s{X}[\ell]}$ for all $0 \leq \ell \leq d-1$.

Combining \eqref{eqn:E14}--\eqref{eqn:E15}, results,
\begin{align}
    \lim_{d \to \infty} \sum_{\ell=0}^{d-1} \pp{2\abs{\s{X}[\ell]}^2 - \abs{\s{\widetilde{X}}_k[\ell]}^2} = \lim_{d \to \infty} \pp{\abs{\s{X}[0]}^2 + \abs{\s{X}[d/2]}^2 + 4\abs{\s{X}[k]}^2} = 0,
\end{align}
which proves \eqref{eqn:xTildeConvergenceToX0}. 

To establish \eqref{eqn:E7}, observe that
\begin{align}
    \lim_{d \to \infty} \sigma_{k}^2 = \frac{\sigma^2}{2} \cdot \lim_{d \to \infty} \sum_{\ell = 0}^{d-1} |\widetilde{\s{X}}_k[\ell]|^2 = \sigma^2,
\end{align}
where the final equality follows from \eqref{eqn:xTildeConvergenceToX0}. This completes the proof of the lemma. 
\end{proof}

We now state a lemma showing that the entries of the covariance matrix $\s{\widetilde{\Sigma}}_k[r,s]$ satisfy the conditions of Proposition~\ref{prop:3}.
\begin{lem} \label{lemma:E3}
    Suppose the conditions of Assumption~\ref{assump:1} hold. Define,
\begin{align}
    \rho_{\abs{{r-s}}} \triangleq |\s{\widetilde{\Sigma}}_{k}\pp{r,s}|, \label{eqn:E8}
\end{align}
for $\s{\widetilde{\Sigma}}_{k}\pp{r,s}$ defined in \eqref{eqn:E5}. Then, $\rho_0 = 1$, and
    \begin{align}
        \rho_n \log \p{n} \to 0.
    \end{align}
    That is, the covariance matrix $\s{\widetilde{\Sigma}}_k[r,s]$ satisfies the conditions required by Proposition~\ref{prop:3}. 
\end{lem}

\begin{proof}[Proof of Lemma~\ref{lemma:E3}]
From the definition of the covariance matrix of $\s{\widetilde{S}}_k \vert \s{N}[k]$ in \eqref{eqn:E5}, we observe that it is circulant and fully characterized by its eigenvalues $|\widetilde{\s{X}}_k[\ell]|^2$ (see \eqref{eqn:XtildeDefenition}) for $0 \leq \ell \leq d-1$. Due to the normalization by $\sigma_k$, the covariance matrix is normalized such that its diagonal entries equal one, i.e., $\rho_0 = 1$.

It remains to show that the off-diagonal elements decay sufficiently fast, namely, $\rho_m \log \p{m} \to 0$, for $m \to \infty$.
    
Using the definition of $\rho$ from \eqref{eqn:E8}, we can write:
\begin{align}
    \rho_m = \s{\widetilde{\Sigma}}_{k}\pp{r,r-m} = \frac{\sum_{\ell = 0}^{d-1} |\widetilde{\s{X}}_k[\ell]|^2 \cos \left( \frac{2\pi \ell}{d}m \right)}{\sum_{\ell = 0}^{d-1} |\widetilde{\s{X}}_k[\ell]|^2}. \label{eqn:E20}
\end{align}
As $d \to \infty$, the denominator in \eqref{eqn:E20} converges to 2 by Lemma~\ref{lemma:E1}. The numerator corresponds to the DFT of the sequence $|\widetilde{\s{X}}_k[\ell]|^2$,
\begin{align}
    \sum_{\ell = 0}^{d-1} |\widetilde{\s{X}}_k[\ell]|^2 \cos \left( \frac{2\pi \ell}{d}m \right) = \mathcal{F}\ppp{|{\widetilde{\s{X}}_k}|^2} [m]. \label{eqn:E21}
\end{align}
By Lemma~\ref{lemma:E1}, for each fixed $k$ and for all $m\in\{1,\ldots,d-1\}$,
\begin{align}
    \mathcal{F}\!\big(|\widetilde{\s{X}}_k|^2\big)[m] &= 2\,\mathcal{F}\!\big(|\s{X}|^2\big)[m] + \Delta_{k,d}[m],
\label{eqn:E22_updated_decomp}
\end{align}
where the error term $\Delta_{k,d}[m]$ is supported only on the modified frequencies
$\ell\in\{0,d/2,k,d-k\}$ and satisfies the uniform bound
\begin{align}
    \sup_{1\le m\le d-1}\,|\Delta_{k,d}[m]| \le |\s{X}[0]|^2 + |\s{X}[d/2]|^2 + 4|\s{X}[k]|^2 \xrightarrow[d\to\infty]{}0,
\label{eqn:E22_updated_err}
\end{align}
using Assumption~\ref{assump:1}(2) and $|\s{X}[0]|=0$ from Assumption~\ref{assump:1}(3).

Moreover, recalling that $\s{R}_{\s{X}\s{X}} = \calF\!\big(|\s{X}|^2\big)$, Assumption~\ref{assump:1}(1) gives
\begin{align}
    \max_{1\le m\le d-1}\, \log(d)\,\big|\mathcal{F}\!\big(|\s{X}|^2\big)[m]\big| = \max_{1\le m\le d-1}\, \log(d)\,|\s{R}_{\s{X}\s{X}}[m]| \xrightarrow[d\to\infty]{}0.
\label{eqn:E23_updated_X}
\end{align}
Combining \eqref{eqn:E22_updated_decomp}--\eqref{eqn:E22_updated_err} with \eqref{eqn:E23_updated_X} yields
\begin{align}
    \max_{1\le m\le d-1}\, \log(d)\,\big|\mathcal{F}\!\big(|\widetilde{\s{X}}_k|^2\big)[m]\big| \xrightarrow[d\to\infty]{}0,
    \label{eqn:E23_updated_Xtilde}
\end{align}
and together with \eqref{eqn:E21} this implies the desired decay condition for the off-diagonal correlations $\rho_m$, completing the proof.
\end{proof}

\subsection{High-dimensional limits} \label{sec:proofOfLemmaE3}

We now present a central result, which plays a key role in the proof of Theorem~\ref{thm:2}. To state the result, we first define the functions:
\begin{align}
   f_1(r) \triangleq \abs{\s{N}[k]} \cos\left(\frac{2\pi k}{d}r + \phi_{\s{N}}[k] - \phi_{\s{X}}[k]\right),
   \label{eqn:42}
\end{align}
and
\begin{align}
   f_2(r) \triangleq\abs{\s{N}[k]}^2 \sin^2\left(\frac{2\pi k}{d}r + \phi_{\s{N}}[k] - \phi_{\s{X}}[k]\right),
   \label{eqn:43}
\end{align}
for $0\leq r\leq d-1$. Note that $f_1$ and $f_2$ correspond to the terms appearing in the expectation in the denominator and numerator of \eqref{eqn:constantCk}, respectively.

\begin{proposition}[High-dimensional limits for ${f_1(\hat{\s{R}})}$ and ${f_2(\hat{\s{R}})}$ ]\label{prop:high-dimentional-limits}
Assume the template signal $x$  satisfies Assumption~\ref{assump:1}, and that its DFT coefficients are non-vanishing, i.e., $\s{X}[k] \neq 0$, for all $0< k \leq d-1$. Let ${f_1(\hat{\s{R}})}$ and ${f_2(\hat{\s{R}})}$ be defined as in \eqref{eqn:42} and \eqref{eqn:43}, with $\hat{\s{R}}$ defined in 
\eqref{eqn:OptShiftRealSpace}. Then, as $d \to \infty$, their expected values satisfy:
    \begin{align}
        \lim_{d\to\infty} \ \frac{1}{ a_d  \cdot \abs{\s{X}[k]}}\mathbb{E} [{f_1(\hat{\s{R}})}] = \sigma,        \label{eqn:targetFunctionAsymptotic}
    \end{align}
and
    \begin{align}
        \lim_{d\to\infty} \mathbb{E} [{f_2(\hat{\s{R}})}] = \frac{\sigma^2}{2},        \label{eqn:targetFunctionAsymptotic2}
    \end{align}
where $a_d \triangleq \sqrt{2\log d}$. 
\end{proposition}

The proof of Proposition~\ref{prop:high-dimentional-limits} builds on Proposition~\ref{prop:3} and the auxiliary Lemmas~\ref{lemma:E1}--\ref{lemma:E3}.

\begin{proof}[Proof of Proposition~\ref{prop:high-dimentional-limits}]
Our goal is to prove \eqref{eqn:targetFunctionAsymptotic} and \eqref{eqn:targetFunctionAsymptotic2}. By the law of total expectation, we have,
    \begin{align}
         \frac{1}{a_d}\mathbb{E}[{f_1(\hat{\s{R}})}] = \frac{1}{a_d}\mathbb{E}\pp{\mathbb{E}\pp{\left. {f_1(\hat{\s{R}})} \right| {\s{N}[k]}}},
    \end{align}
and
    \begin{align}
         \mathbb{E}[{f_2(\hat{\s{R}})}] = \mathbb{E}\pp{\mathbb{E}\pp{\left. {f_2(\hat{\s{R}})} \right| {\s{N}[k]}}}.
    \end{align}
Accordingly, we will prove
    \begin{align}   
        \frac{\sigma_k}{a_d \cdot \abs{\s{X}[k]}}{\mathbb{E} \pp{\left.{f_1(\hat{\s{R}})} \right| \s{N}[k]}} \xrightarrow[]{\mathcal{L}^1} \abs{\s{N}[k]}^2, \label{eqn:E26}
    \end{align}
and,
    \begin{align}   
        {\mathbb{E} \pp{\left.{f_2(\hat{\s{R}})} \right| \s{N}[k]}} \xrightarrow[]{\mathcal{L}^1} \frac{1}{2} \abs{\s{N}[k]}^2, \label{eqn:E27}
    \end{align}
which would yield the desired result. 

To proceed, we apply Proposition~\ref{prop:3}.
Recall the definition of the vector $\Tilde{\s{S}}_k$ given in \eqref{eqn:E12}.
Conditioned on $\s{N}[k]$, this vector follows a Gaussian distribution with mean $\widetilde{\mu}_k[r]$ (as defined in \eqref{eqn:E4}) and covariance matrix $\s{\widetilde{\Sigma}}_{k}[r,s]$ (defined in \eqref{eqn:E5}).
We assume that both the mean and the covariance satisfy Assumption~\ref{assump:1}, and we assert that they also meet the criteria of Proposition~\ref{prop:3}. Indeed, observe the following:
\begin{enumerate}
        \item \textbf{The mean term.} By Assumption~\ref{assump:1}, we have $\abs{\s{X}[k]} \sqrt{\log d} \to 0$, as $d\to\infty$, for every $0 \leq k \leq d-1$, implying that $\sqrt{\log d}\max{\abs{\widetilde{\mu}_{k}[r]}} \to 0$, where the $\abs{\s{N}[k]}$ term in $\widetilde{\mu}_k[r]$ is finite and independent of $d$.
    \item \textbf{The covariance term.} By Lemma \ref{lemma:E3}, the covariance matrix $\s{\widetilde{\Sigma}}_{k}[r,s]$ satisfies the conditions of Proposition \ref{prop:3}.
\end{enumerate}

We apply Proposition~\ref{prop:3} and the result in \eqref{eqn:Lemma2} to the functions $f_1(\hat{\s{R}})$ and $f_2(\hat{\s{R}})$, with respect to the Gaussian vector ${\s{\widetilde{S}}_k \vert }\s{N}[k]$ \eqref{eqn:conditionalGaussianNormalized}. Observe that
    \begin{align}   
        \max_{0 \leq r \leq d-1}{\abs{\widetilde{\mu}_k[r]}} = 2\sigma_k^{-1}\abs{\s{X}[k]}\abs{\s{N}[k]}, \label{eqn: 31}
    \end{align}
and,
    \begin{align}   
        {f_1(\hat{\s{R}})}& = {\abs{\s{N}[k]}\cos\left(\frac{2\pi k}{d} \hat{\s{R}} + \phi_\s{N}[k] - \phi_\s{X}[k] \right)}= \frac{\sigma_k}{2\abs{\s{X}[k]}}{\widetilde{\mu}_k[\hat{\s{R}}]}. \label{eqn:E25}
    \end{align}

Given that the assumptions of Proposition~\ref{prop:3} are satisfied, it follows that
\begin{align}   
        \mathbb{E} \left[\left.{f_1(\hat{\s{R}})} \right| \s{N}[k] \right]  -\frac{\sum_{r=0}^{d-1}f_1(r)e^{\widetilde{\mu}_{k}[r] a_d }}{\sum_{r=0}^{d-1}e^{\widetilde{\mu}_{k}[r] a_d }} \xrightarrow[]{\s{a.s.}} 0,
        \label{eqn:f1Expectation}
    \end{align}
and,
    \begin{align}   
        \mathbb{E} \left[\left.{f_2(\hat{\s{R}})} \right| \s{N}[k] \right] - \frac{\sum_{r=0}^{d-1}f_2(r)e^{\widetilde{\mu}_{k}[r] a_d }}{\sum_{r=0}^{d-1}e^{\widetilde{\mu}_{k}[r] a_d }} \xrightarrow[]{\s{a.s.}} 0.
        \label{eqn:f2Expectation}
    \end{align}
Next, we evaluate the terms at the left-hand-side of \eqref{eqn:f1Expectation} and \eqref{eqn:f2Expectation}. 

\paragraph{Proof of \eqref{eqn:E26} and \eqref{eqn:targetFunctionAsymptotic}.} We begin by proving that
 \begin{align}   
        \frac{1}{2a_d} \frac{\sum_{r=0}^{d-1}{\widetilde{\mu}_{k}[r]\exp\{\widetilde{\mu}_{k}[r] a_d }\}}{\sum_{r=0}^{d-1}{\exp\{\widetilde{\mu}_{k}[r] a_d }\}} \frac{\sigma_k^2}{\abs{\s{X}[k]}^2} \xrightarrow[]{\s{a.s.}} \abs{\s{N}[k]}^2.\label{eqn: 30}
    \end{align}
From the definition of $f_1(r)$, it follows that
    \begin{align}   
        \sum_{r=0}^{d-1}{f_1(r)} = 0, \label{eqn:E30}
    \end{align}
almost surely. Additionally, from the definition of $\widetilde{\mu}_k[r]$,
    \begin{align}   
        {\widetilde{\mu}_{k}[r] a_d} = 2 \sigma_k^{-1} a_d \abs{\s{X}[k]}\abs{\s{N}[k]} \cos\left( {\frac{2\pi k}{d}r + \phi_{\s{N}}[k] - \phi_{\s{X}}[k]} \right). \label{eqn: 85}
    \end{align}
By Assumption \ref{assump:1}, we have $a_d \abs{\s{X}[k]} \to 0$ as $d \to \infty$. Thus, from \eqref{eqn: 85} and the continuous mapping theorem, we obtain
    \begin{align}   
        {\widetilde{\mu}_{k}[r] a_d} \xrightarrow[]{\s{a.s.}} 0.
        \label{eqn:convergenceOfMuRaDtoZero}
    \end{align}
Since $\sum_{r=0}^{d-1} \widetilde{\mu}_{k}[r] = 0$ almost surely (from \eqref{eqn:E30}), and applying the continuous mapping theorem along with \eqref{eqn:convergenceOfMuRaDtoZero}, we deduce that
    \begin{align}   
        \frac{{\sum_{r=0}^{d-1}{\widetilde{\mu}_{k}[r]\exp\{\widetilde{\mu}_{k}[r] a_d }\}}}{a_d \sum_{r=0}^{d-1} [{\widetilde{\mu}_{k}[r]}]^2} \xrightarrow[]{\s{a.s.}} 1.
        \label{eqn: 86}
    \end{align}
Similarly, applying the continuous mapping theorem using \eqref{eqn:convergenceOfMuRaDtoZero}, we get
    \begin{align}   
        \frac{{\sum_{r=0}^{d-1}{\exp\{\widetilde{\mu}_{k}[r] a_d }\}}}{d} \xrightarrow[]{\s{a.s.}} 1.
        \label{eqn: 87}
    \end{align}
Next, from the definition of $\widetilde{\mu}_k[r]$, 
    \begin{align}   
       \sum_{r=0}^{d-1} [{\widetilde{\mu}_{k}[r]}]^2 = 4 \sigma_k^{-2} \left(\abs{\s{X}[k]} \abs{\s{N}[k]}\right)^2 \sum_{r=0}^{d-1} \cos^2\left( {\frac{2\pi k}{d}r + \phi_{\s{N}}[k] - \phi_{\s{X}}[k]} \right).
        \label{eqn: 88}
    \end{align}
As $d \to \infty$, 
    \begin{align}   
       \frac{1}{d} \sum_{r=0}^{d-1} \cos^2\left( {\frac{2\pi k}{d}r + \phi_{\s{N}}[k] - \phi_{\s{X}}[k]} \right) \xrightarrow[]{\s{a.s.}} \frac{1}{2}.
        \label{eqn: 89}
    \end{align}
Combining \eqref{eqn: 88} and \eqref{eqn: 89}, we obtain
\begin{align}   
       \frac{\sigma_k^2}{2\abs{\s{X}[k]}^2}  \sum_{r=0}^{d-1} [{\widetilde{\mu}_{k}[r]}]^2
        \xrightarrow[]{\s{a.s.}} \abs{\s{N}[k]}^2 ,
        \label{eqn:E37}
    \end{align}
for $d \to \infty$.
Thus, combining \eqref{eqn: 86}--\eqref{eqn:E37}, we conclude that
\begin{align}
    \frac{\sigma_k^2}{2\abs{\s{X}[k]}^2a_d} \frac{\sum_{r=0}^{d-1} \widetilde{\mu}_k[r]e^{\widetilde{\mu}_{k}[r] a_d }}{\sum_{r=0}^{d-1}e^{\widetilde{\mu}_{k}[r] a_d }}  \xrightarrow[]{\s{a.s.}} \abs{\s{N}[k]}^2. \label{eqn:E38}
\end{align}
This proves \eqref{eqn: 30}. 
Finally, combining \eqref{eqn:E25}, \eqref{eqn:f1Expectation}, and \eqref{eqn:E38}, we arrive at
\begin{align}   
        \frac{\sigma_k^2}{2\abs{\s{X}[k]}^2a_d}  \mathbb{E} \pp{\left.{\widetilde{\mu}_k[\hat{\s{R}}]} \right| \s{N}[k]} \xrightarrow[]{\s{a.s.}} \abs{\s{N}[k]}^2. \label{eqn: 117}
    \end{align}

Let the term on the left-hand side of \eqref{eqn: 117} be denoted by $G_d$ and the term on the right-hand side by $G$, so that $G_d \xrightarrow[]{\text{a.s.}} G$. By definition, observe that $\abs{G_d} \leq \abs{\s{N}[k]}^2$, and it is evident that $\mathbb{E}[\abs{\s{N}[k]}^2] = \sigma^2 < \infty$. Thus, by the dominated convergence theorem, we conclude that $G_d \xrightarrow[]{\mathcal{L}^1} G$, and specifically,
\begin{align}   
        \mathbb{E} \left[ \frac{\sigma_k^2}{2 a_d}\frac{\mathbb{E} [\mu(\hat{\s{R}}) \vert \s{N}[k]]}{\abs{\s{X}[k]}^2} \right] \rightarrow \mathbb{E} \left[  \abs{\s{N}[k]}^2 \right]. \label{eqn:E40}
    \end{align}
By combining \eqref{eqn:E40} with \eqref{eqn:E25}, we get,
\begin{align}   
        \frac{\sigma_k}{a_d \cdot \abs{\s{X}[k]}}{\mathbb{E} \pp{\left.{f_1(\hat{\s{R}})} \right| \s{N}[k]}} \xrightarrow[]{\mathcal{L}^1} \abs{\s{N}[k]}^2, 
    \end{align}
or equivalently,
    \begin{align}
    \mathbb{E} \pp{\frac{\sigma_k}{a_d \cdot \abs{\s{X}[k]}}  \mathbb{E} \pp{\left.{f_1(\hat{\s{R}})} \right| \s{N}[k]}} \to \mathbb{E} \pp{\abs{\s{N}[k]}^2}, \label{eqn:E41}
\end{align}
which proves \eqref{eqn:E26}. 

By the law of total expectation, we then obtain
    \begin{align}   
        \frac{\sigma_k}{a_d \cdot  {\abs{\s{X}[k]}}} \mathbb{E}\pp{\abs{\s{N}[k]}\cos\left(\frac{2\pi k}{d} \hat{\s{R}} + \phi_\s{N}[k] - \phi_\s{X}[k] \right)}  \to \mathbb{E} \left[ {\abs{\s{N}[k]}^2} \right] = \sigma^2,        \label{eqn:asymptoticConvergenceOfMagnitude}
    \end{align}
as $d\to\infty$. By Lemma \ref{lemma:E1}, we know that $\sigma_k^2 \to \sigma^2$ as $d \to \infty$, so from \eqref{eqn:asymptoticConvergenceOfMagnitude}, we conclude that 
    \begin{align}   
        \frac{1}{a_d \cdot  {\abs{\s{X}[k]}}} \mathbb{E}\pp{\abs{\s{N}[k]}\cos\left(\frac{2\pi k}{d} \hat{\s{R}} + \phi_\s{N}[k] - \phi_\s{X}[k] \right)}  \to \sigma,
        \label{eqn:E46}
    \end{align}
which proves \eqref{eqn:targetFunctionAsymptotic}. 

\paragraph{Proof of \eqref{eqn:E27} and \eqref{eqn:targetFunctionAsymptotic2}.} The numerator in \eqref{eqn:constantCk} converges to,
    \begin{align}   
        \mathbb{E}{\left[ \abs{\s{N}[k]} \sin(\phi_{e}[k]) \right]}^2 = \frac{1}{2} \pp{\mathbb{E}\abs{\s{N}[k]}^2 - \mathbb{E}[\abs{\s{N}[k]}^2\cos(2\phi_{e}[k])]}  \rightarrow \frac{1}{2}\mathbb{E}\abs{\s{N}[k]}^2,
    \end{align}
as $d\to\infty$, where the last transition is because $\mathbb{E}[\cos(2\phi_{e}[k])\vert\s{N}[k]]  \xrightarrow[]{\s{a.s.}} 0$, as $d \to \infty$. Thus,
    \begin{align}   
        \frac{1}{\sigma^2}\mathbb{E}{\left[ \abs{\s{N}[k]} \sin(\phi_{e}[k]) \right]}^2 \rightarrow \frac{1}{2 \sigma^2}\mathbb{E}\abs{\s{N}[k]}^2 = \frac{1}{2},        \label{eqn:expectedValueOfSinSquare}
    \end{align}
which concludes the proof.
\end{proof}

\subsection{Proof of Theorem~\ref{thm:2}}\label{thm:proofs2}

To begin, let us summarize the notation and results from the previous sections as the foundation for the proof. Recall the definition of $\phi_{e}[k]$ in \eqref{eqn:phaseDifferenceTerm}, as well as the definitions of $f_1\p{r}$ and $f_2\p{r}$ in \eqref{eqn:42}-\eqref{eqn:43}, and let $a_d \triangleq \sqrt{2\log d}$. According to Theorem \ref{thm:1}, the convergence of the Fourier phases is given by:
    \begin{align}
       \lim_{M\to\infty} \frac{\mathbb{E} |\phi_{\s{\hat{X}}}[k] - \phi_{\s{X}}[k]|^2}{1/M} = C_k.
        \label{eqn:F1}
    \end{align}
where $C_k$ is given by \eqref{eqn:constantCk}:
\begin{align}
    C_k = \frac{{\mathbb{E}\pp{\left[ \abs{\s{N}[k]} \sin(\phi_{e}[k]) \right]^2}}}{{\p{\mathbb{E}{\left[ \abs{\s{N}[k]} \cos(\phi_{e}[k]) \right]}}}^2}. \label{eqn:F1_2}
\end{align}
The constant $C_k$ can be rewritten as: 
\begin{align}   
    a^2_{d}\cdot C_k = a_d^2 \cdot \frac{{\mathbb{E}\pp{\left[ \abs{\s{N}[k]} \sin(\phi_{e}[k]) \right]^2}}}{{\p{\mathbb{E}{\left[ \abs{\s{N}[k]} \cos(\phi_{e}[k]) \right]}}}^2} 
    = a_d^2 \cdot \frac{\mathbb{E}{[{f_2(\hat{\s{R}})}]}}{\p{\mathbb{E}{[{f_1(\hat{\s{R}})}]}}^2},
    \label{eqn:CkInF1F2representation}
\end{align}
where $f_1\p{r}$ and $f_2\p{r}$ are defined in \eqref{eqn:42}-\eqref{eqn:43}.

The signal $x$, satisfying the conditions of Theorem~\ref{thm:2}, also satisfies the assumptions of Proposition~\ref{prop:high-dimentional-limits}. By Proposition \ref{prop:high-dimentional-limits}, we have:
\begin{align}
        \lim_{d\to\infty} \ \frac{1}{ a_d  \cdot \abs{\s{X}[k]}}\mathbb{E} [{f_1(\hat{\s{R}})}] = \sigma, \label{eqn:F2}
    \end{align}
and,
    \begin{align}
        \lim_{d\to\infty} \mathbb{E} [{f_2(\hat{\s{R}})}] = \frac{\sigma^2}{2}. \label{eqn:F3}
    \end{align}
Now, we are ready to prove the results of the Theorem.

\paragraph{Convergence of the Fourier phases.} 
From \eqref{eqn:CkInF1F2representation}--\eqref{eqn:F3}, we obtain: 
    \begin{align}   
        \lim_{d\to\infty} a^2_{d} \cdot \abs{\s{X}[k]}^2 \cdot C_k  = \frac{1}{2}.\label{eqn:F6}
    \end{align}
Combining \eqref{eqn:F1}, \eqref{eqn:F6}, results,
 \begin{align}
     \lim_{d\to\infty} \lim_{M\to\infty} & \frac{\mathbb{E} |\phi_{\s{\hat{X}}}[k] - \phi_{\s{X}}[k]|^2}{1/M} \frac{1}{1/(4 \log \p{d} \abs{\s{X}[k]}^2)} = \lim_{d \to \infty} \frac{C_k}{1/(4 \log \p{d} \abs{\s{X}[k]}^2)} \label{eqn:F8}
    \\ & = \lim_{d \to \infty} \frac{a_d^2 \cdot \abs{\s{X}[k]}^2 \cdot C_k}{1/2} = 1, \label{eqn:F9}
\end{align}
where \eqref{eqn:F8} follows from \eqref{eqn:F1}, and \eqref{eqn:F9} follows from \eqref{eqn:F6}, proving \eqref{eqn:phaseConvergeneRateForAsymptoticD}.

\paragraph{Convergence of the Fourier magnitudes.} Finally, we prove \eqref{eqn:magnitudeConvergenceAsymptoticMandD}. By Theorem \ref{thm:1} and \eqref{eqn:magnitudeConvergenceAsymptoticM}, we have:
    \begin{align}
        |\s{\hat{X}}[k]| \xrightarrow[]{\s{a.s.}} 
        \mathbb{E} \left[ \abs{\s{N}[k]} \cos\left(\frac{2\pi k}{d}\hat{\s{R}} + \phi_{\s{N}}[k] - \phi_{\s{X}}[k]\right) \right] = \mathbb{E}{[{f_1(\hat{\s{R}})}]}. \label{eqn:F10}
    \end{align}
Combining \eqref{eqn:F2}, \eqref{eqn:F10} yields,
    \begin{align}   
        \frac{1}{ a_d \sigma } \frac{|\s{\hat{X}}[k]|}{\abs{\s{X}[k]}} \xrightarrow[]{\s{a.s.}} \frac{1}{a_d \sigma } \frac{\mathbb{E} [{f_1(\hat{\s{R}})}]}{\abs{\s{X}[k]}} \rightarrow 1,
    \end{align}
as $M,d\to\infty$, where the second passage follows from \eqref{eqn:targetFunctionAsymptotic}. As $a_d = \sqrt{2 \log \p{d}}$, this completes the proof of the Theorem.

\section{Proof of Proposition \ref{prop:positiveCorrelation}}\label{sec:proofOfPositiveCorrelation}

Before proving Proposition~\ref{prop:positiveCorrelation}, we first establish the following auxiliary lemma.

\begin{lem}\label{lem:GG1}
    Let $\s{A} = (A_0, A_1, \ldots, A_{d-1})$ be a $d$-dimensional random vector with $\mathbb{E}[\s{A}] = 0$. Then,
    \begin{align}
        \mathbb{E} \left[\max \{A_0, A_1, \ldots, A_{d-1}\} \right] \geq \max_{0 \leq r_1, r_2 \leq d-1} \frac{1}{2} \mathbb{E} \left[|A_{r_1} - A_{r_2}|\right].
    \end{align}
\end{lem}

\begin{proof}[Proof of Lemma~\ref{lem:GG1}]
    For any two real numbers $x$ and $y$, we have:
    \begin{align}
        \max(x,y) = \frac{1}{2}(x + y + |x - y|).
    \end{align}
    Applying this to any pair $A_{r_1}, A_{r_2}$ yields:
    \begin{align}
        \mathbb{E}[\max\{A_{r_1}, A_{r_2}\}] &= \frac{1}{2} \mathbb{E}[A_{r_1} + A_{r_2} + |A_{r_1} - A_{r_2}|] \\
        &= \frac{1}{2} \mathbb{E}[|A_{r_1} - A_{r_2}|], \label{eqn:expectedValueOfMaximum}
    \end{align}
    where we used the assumption that $\mathbb{E}[A_r] = 0$ for all $r$.

    By the convexity of the max function, it holds that:
    \begin{align}
        \mathbb{E}[\max \{A_0, A_1, \ldots, A_{d-1}\}] \geq \mathbb{E}[\max \{A_{r_1}, A_{r_2}\}], \label{eqn:app_GG5}
    \end{align}
    for every $r_1, r_2 \in \{0, 1, \ldots, d-1\}$. Combining \eqref{eqn:expectedValueOfMaximum} and \eqref{eqn:app_GG5}, we conclude:
    \begin{align}
        \mathbb{E}[\max \{A_0, A_1, \ldots, A_{d-1}\}] \geq \max_{0 \leq r_1, r_2 \leq d-1} \frac{1}{2} \mathbb{E}[|A_{r_1} - A_{r_2}|],
    \end{align}
    completing the proof.
\end{proof}

Let $n_0, n_1, \ldots, n_{M-1}$ be an i.i.d. sequence of zero-mean random vectors with covariance $\mathbb{E}[n_i n_i^\top] = \Sigma$, which by assumption $\Sigma$ is positive-definite. Recall the definition of the EfN estimator in \eqref{eqn:efnEstimatorRealSpace}:
\begin{align}
    \hat{x} \triangleq \frac{1}{M} \sum_{i=0}^{M-1} \mathcal{T}_{-\s{\hat{R}}_i} n_i,
\end{align}
where the estimated shift $\s{\hat{R}}_i$ is given by:
\begin{align}
    \s{\hat{R}}_i \triangleq \argmax_{0 \leq \ell \leq d-1} \langle n_i, \mathcal{T}_\ell x \rangle. \label{eqn:app_GG_2}
\end{align}
Using linearity of the inner product:
\begin{align}
    \langle \hat{x}, x \rangle 
    &= \left\langle \frac{1}{M} \sum_{i=0}^{M-1} \mathcal{T}_{-\s{\hat{R}}_i} n_i, x \right\rangle 
    = \frac{1}{M} \sum_{i=0}^{M-1} \langle \mathcal{T}_{-\s{\hat{R}}_i} n_i, x \rangle.
\end{align}

By SLLN, as $M \to \infty$, we have almost surely:
\begin{align}
    \frac{1}{M} \sum_{i=0}^{M-1} \langle \mathcal{T}_{-\s{\hat{R}}_i} n_i, x \rangle \xrightarrow[]{\text{a.s.}} \mathbb{E}[\langle \mathcal{T}_{-\s{\hat{R}}_1} n_1, x \rangle]. \label{eqn:app_GG4}
\end{align}
Define for $r \in \{0, 1, \ldots, d-1\}$ the random variables:
\begin{align}
    A_r \triangleq \langle n_1, \mathcal{T}_r x \rangle.
\end{align}
Then, the right-hand side of \eqref{eqn:app_GG4} becomes:
\begin{align}
    \mathbb{E}[\langle \mathcal{T}_{-\s{\hat{R}}_1} n_1, x \rangle] = \mathbb{E}[\max \{A_0, A_1, \ldots, A_{d-1}\}].
\end{align}

Applying Lemma~\ref{lem:GG1}, we get:
\begin{align}
    \mathbb{E}[\max \{A_0, A_1, \ldots, A_{d-1}\}] 
    &\geq \max_{0 \leq r_1, r_2 \leq d-1} \frac{1}{2} \mathbb{E}[|A_{r_1} - A_{r_2}|] \\
    &= \max_{0 \leq r_1, r_2 \leq d-1} \frac{1}{2} \mathbb{E}[|\langle n_1, \mathcal{T}_{r_1} x - \mathcal{T}_{r_2} x \rangle|]. \label{eqn:lowerBoundByInnerProd}
\end{align}

To complete the proof, we show that the lower bound in \eqref{eqn:lowerBoundByInnerProd} is strictly positive. Since $x \in \mathbb{R}^d$ is nonzero with non-vanishing Fourier components $\s{X}[k] \neq 0$, for every $1 \leq k \leq d-1$, and $\mathcal{T}_r$ is a cyclic shift operator, the set $\{\mathcal{T}_r x : 0 \leq r < d\}$ contains at least $d-1$ distinct vectors. Thus, there exist $r_1, r_2 \in \{0, \ldots, d-1\}$ such that
\begin{align}
    v \triangleq \mathcal{T}_{r_1} x - \mathcal{T}_{r_2} x \neq 0.
\end{align}
Then the inner product $\langle n_1, v \rangle$ is a real-valued random variable with
\begin{align}
    \mathrm{Var}(\langle n_1, v \rangle) = \mathbb{E}[\langle n_1, v \rangle^2] = v^\top \Sigma v > 0,    
\end{align}
because $v \neq 0$ and $\Sigma$ is positive definite. Hence, $\langle n_1, v \rangle$ is not almost surely zero, and
\begin{align}
    \mathbb{E}[|\langle n_1, v \rangle|] > 0.
\end{align}
This implies that
\begin{align}
    \max_{0 \leq r_1, r_2 < d} \frac{1}{2} \mathbb{E}[|\langle n_1, \mathcal{T}_{r_1} x - \mathcal{T}_{r_2} x \rangle|] > 0,
\end{align}
and consequently,
\begin{align}
    \lim_{M \to \infty} \langle \hat{x}, x \rangle = \mathbb{E}[\langle \mathcal{T}_{-\hat{\s{R}}_1} n_1, x \rangle] > 0,
\end{align}
almost surely. This completes the proof.

\section{Proof of Theorem \ref{thm:highDimentionalNoiseExtention}: High-dimensional i.i.d. noise} \label{sec:proofOfHighDimetnionalNoiseExtention}

In this section, we prove Theorem~\ref{thm:highDimentionalNoiseExtention}. The proof relies on the functional central limit theorem for the discrete Fourier transform \cite{peligrad2010central, cerovecki2017clt, brillinger2001time}, which we review in Appendix~\ref{sec:functionalCLTforDFT}. In Appendices~\ref{sec:notationsForIIDgeneral} and~\ref{sec:convergenceOfRealAndImaginaryParts}, we apply this result to analyze the real and imaginary parts of the \EfN estimator under a general i.i.d. noise model, and compare the outcome to the white Gaussian case. Finally, the proof of Theorem \ref{thm:highDimentionalNoiseExtention} is deduced in Appendix~\ref{sec:proofHighDimentionalSub}.

\subsection{The functional CLT for DFT} \label{sec:functionalCLTforDFT}
We begin by presenting a functional central limit theorem (CLT) for the DFT, which establishes that the DFT of an i.i.d. real-valued sequence converges in distribution as the dimension $d \to \infty$. This result has been studied in the literature; see, for example, \cite{peligrad2010central, cerovecki2017clt, brillinger2001time}.
To formalize this, we state the following functional CLT for DFTs of i.i.d. sequences.

\begin{thm}[Functional CLT for the DFT]\label{thm:functionalCLT}
Let $\{z_n\}_{n \in \mathbb{N}}$ be a sequence of i.i.d. real-valued random variables with zero mean $\mathbb{E}[z_0] = 0$ and finite variance $\mathbb{E}[z_0^2] = \sigma^2 < \infty$.
For each integer $d \geq 1$, define the DFT of the finite segment $\{z_0, \ldots, z_{d-1}\}$ as
\begin{align}
    \s{Z}^{(d)}[k] \triangleq \frac{1}{\sqrt{d}} \sum_{\ell=0}^{d-1} z_\ell e^{-2\pi j k \ell / d}, \quad 0 \leq k < d.
\end{align}
Extend $\s{Z}^{(d)}$ to an infinite sequence by zero-padding outside the index set $\{0, \ldots, d-1\}$:
\begin{align}
    \s{Z}^{(d)} = \big(\s{Z}^{(d)}[0], \ldots, \s{Z}^{(d)}[d-1], 0, 0, \ldots \big) \in \mathbb{C}^\mathbb{N}.
\end{align}
Then, for any fixed finite index set $\{k_1, k_2, \ldots, k_m\} \subset \mathbb{N}$, the finite-dimensional vectors $\big(\s{Z}^{(d)}[k_1], \ldots, \s{Z}^{(d)}[k_m]\big)$ converge in distribution, as $d \to \infty$,
\begin{align}
    \big(\s{Z}^{(d)}[k_1], \ldots, \s{Z}^{(d)}[k_m]\big) \xrightarrow[d \to \infty]{\mathcal{D}} \big(\s{W}_{k_1}, \ldots, \s{W}_{k_m}\big),
\end{align}
where $\s{W} = (\s{W}_k)_{k \in \mathbb{N}}$ is a sequence of i.i.d. circularly symmetric complex Gaussian random variables with $\s{W}_k \sim \mathcal{CN}(0, \sigma^2)$.
\end{thm}

This result can be obtained from the  multivariate Lindeberg--Feller CLT. This convergence holds jointly over any finite collection of indices, meaning that for every finite subset $I \subset \mathbb{N}$, the finite-dimensional vector $(\s{Z}^{(d)}_k)_{k \in I}$ converges in distribution to $(\s{W}_k)_{k \in I}$, where the $\s{W}_k$ are i.i.d. circularly symmetric complex Gaussian random variables. The collection of these finite-dimensional distributions is consistent and satisfies the compatibility conditions of Kolmogorov’s extension theorem, thereby uniquely determining a probability law on the infinite product space $\mathbb{C}^\mathbb{N}$. In this sense, the convergence of $\s{Z}^{(d)}$ to $\s{W}$ is fully characterized by convergence of finite-dimensional distributions.

\subsection{Notations} \label{sec:notationsForIIDgeneral}

\paragraph{General i.i.d noise.}
Let $z_0, z_1, \ldots$ be a sequence of i.i.d random variables with $\mathbb{E}[z_0] = 0$, and $\mathbb{E}[z_0^2] = \sigma^2 < \infty$, and $\mathbb{E}\pp{z_0^4} < \infty$. Define the (zero-padded) DFT transform of the finite segment ${z}^{(d)} = (z_0, z_1, \ldots, z_{d-1})$ as
\begin{align}
    \s{Z}^{(d)}[k] \triangleq 
    \begin{cases}
    \displaystyle \frac{1}{\sqrt{d}} \sum_{\ell=0}^{d-1} z_\ell e^{-2\pi j k \ell / d}, & 0 \leq k < d, \\
    0, & \text{otherwise}.
    \end{cases} \label{eqn:app_G1}
\end{align}
Let $x = (x_0, \ldots, x_{d-1})$ be the deterministic template signal, and denote its DFT by $\s{X}[k]$. Define the maximal correlation shift between $x$ and $z^{(d)}$ in the Fourier domain as
\begin{align}
    \hat{\s{R}}_{\s{Z}}^{(d)} \triangleq  \argmax_{0 \leq r < d} \sum_{k=0}^{d-1} \abs{\s{X}[k]} \abs{\s{Z}[k]} \cos\left( \frac{2\pi k r}{d} + \phi_{\s{Z}}[k] - \phi_{\s{X}}[k] \right),
    \label{eqn:app_G2}
\end{align}
where $\phi_{\s{Z}}[k]$ and $\phi_{\s{X}}[k]$ are the phases of $\s{Z}^{(d)}[k]$ and $\s{X}[k]$, respectively. Define the phase difference after alignment as
\begin{align}
    \phi_{e,\s{Z}}^{(d)}[k] \triangleq  \frac{2\pi k \hat{\s{R}}_{\s{Z}}^{(d)}}{d} + \phi_{\s{Z}}[k] - \phi_{\s{X}}[k].
    \label{eqn:app_G3}
\end{align}

\paragraph{Gaussian i.i.d noise.}
Let $n_0, n_1, \ldots$ be an i.i.d sequence of \textit{Gaussian random variables} with $n_\ell \sim \mathcal{N}(0, \sigma^2)$. Define the DFT of the segment $n^{(d)} = (n_0, \ldots, n_{d-1})$ as
\begin{align}
    \s{N}^{(d)}[k] \triangleq  
    \begin{cases}
    \displaystyle \frac{1}{\sqrt{d}} \sum_{\ell=0}^{d-1} n_\ell e^{-2\pi j k \ell / d}, & 0 \leq k < d, \\
    0, & \text{otherwise}.
    \end{cases}
\end{align}
Define the corresponding maximal correlation shift:
\begin{align}
    \hat{\s{R}}_{\s{N}}^{(d)} \triangleq  \argmax_{0 \leq r < d} \sum_{k=0}^{d-1} \abs{\s{X}[k]} \abs{\s{N}[k]} \cos\left( \frac{2\pi k r}{d} + \phi_{\s{N}}[k] - \phi_{\s{X}}[k] \right),
    \label{eqn:app_G5}
\end{align}
and define the aligned phase difference as
\begin{align}
    \phi_{e,\s{N}}^{(d)}[k] \triangleq  \frac{2\pi k \hat{\s{R}}_{\s{N}}^{(d)}}{d} + \phi_{\s{N}}[k] - \phi_{\s{X}}[k].
    \label{eqn:app_G6}
\end{align}

\subsection{Convergence of the real and imaginary parts of the EfN estimator} \label{sec:convergenceOfRealAndImaginaryParts}
We now present an auxiliary result that relates the real and imaginary parts of the EfN estimator under both Gaussian i.i.d. and general i.i.d. noise models. This result is a consequence of the functional central limit theorem for the DFT (Theorem~\ref{thm:functionalCLT}).

\begin{proposition} \label{prop:G1}
Let $z_0, z_1, \ldots$ be a sequence of i.i.d. real-valued random variables with zero mean, finite variance, $\mathbb{E}[z_0^2] = \sigma^2 < \infty$, and finite fourth moment.
Let $n_0, n_1, \ldots$ be an i.i.d. sequence of Gaussian random variables with $n_\ell \sim \mathcal{N}(0, \sigma^2)$.  
Let $\s{Z}^{(d)}$ and $\s{N}^{(d)}$ denote the DFTs of the sequences $\{z_i\}_{i=0}^{d-1}$ and $\{n_i\}_{i=0}^{d-1}$, respectively. Let $\phi_{e,\s{Z}}^{(d)}[k]$ and $\phi_{e,\s{N}}^{(d)}[k]$ denote the aligned phase differences defined in~\eqref{eqn:app_G3} and~\eqref{eqn:app_G6}, respectively.
Then, for each fixed frequency index $k \in \mathbb{N}$, 
\begin{align}
    \lim_{d \to \infty} \left( 
    \mathbb{E}\left[ \left| \s{Z}^{(d)}[k] \right| \sin\left( \phi_{e,\s{Z}}^{(d)}[k] \right) \right] 
    - \mathbb{E}\left[ \left| \s{N}^{(d)}[k] \right| \sin\left( \phi_{e,\s{N}}^{(d)}[k] \right) \right] 
    \right) = 0, \label{eqn:app_G7a}
\end{align}
\begin{align}
    \lim_{d \to \infty} \left( 
    \mathbb{E}\left[ \left| \s{Z}^{(d)}[k] \right| \cos\left( \phi_{e,\s{Z}}^{(d)}[k] \right) \right] 
    - \mathbb{E}\left[ \left| \s{N}^{(d)}[k] \right| \cos\left( \phi_{e,\s{N}}^{(d)}[k] \right) \right] 
    \right) = 0, \label{eqn:app_G7b}
\end{align}
and
\begin{align}
    \lim_{d \to \infty} \left( 
    \s{Var}\left[ \left| \s{Z}^{(d)}[k] \right| \sin\left( \phi_{e,\s{Z}}^{(d)}[k] \right) \right] 
    - \s{Var}\left[ \left| \s{N}^{(d)}[k] \right| \sin\left( \phi_{e,\s{N}}^{(d)}[k] \right) \right] 
    \right) = 0, \label{eqn:app_G9a}
\end{align}
\begin{align}
    \lim_{d \to \infty} \left( 
    \s{Var}\left[ \left| \s{Z}^{(d)}[k] \right| \cos\left( \phi_{e,\s{Z}}^{(d)}[k] \right) \right] 
    - \s{Var}\left[ \left| \s{N}^{(d)}[k] \right| \cos\left( \phi_{e,\s{N}}^{(d)}[k] \right) \right] 
    \right) = 0. \label{eqn:app_G9b}
\end{align}
\end{proposition}

\begin{proof}[Proof of Proposition~\ref{prop:G1}]

We begin by applying Theorem~\ref{thm:functionalCLT} to the sequence $\{z_i\}$, which satisfies its assumptions. This in turn implies that,
\begin{align}
    \s{Z}^{(d)} \xrightarrow{\mathcal{D}} \s{W},
\end{align}
where $\s{W}_k \sim \mathcal{CN}(0, \sigma^2)$ are i.i.d. complex Gaussian variables.
Next, consider the Gaussian noise sequence $n_i \sim \mathcal{N}(0, \sigma^2)$. Its discrete Fourier transform satisfies
\begin{align}
    \s{N}^{(d)} \overset{\mathcal{D}}{=} (\s{W}[0], \ldots, \s{W}[d-1]),
\end{align}
for every $d$, since the DFT of an i.i.d. Gaussian sequence remains i.i.d. in distribution with the same complex Gaussian law. Since $\mathbb{E}[z_i^2] = \sigma^2 < \infty$, we have
\begin{align}
    \mathbb{E}[|\s{Z}^{(d)}[k]|^2] < \infty,
\end{align}
for each $d$ and $k$. It follows that the sequences of random variables,
\begin{align}
    \left\{ |\s{Z}^{(d)}[k]| \sin(\phi_{e,\s{Z}}^{(d)}[k]) \right\}_{d \in \mathbb{N}^+}, \quad
    \left\{ |\s{Z}^{(d)}[k]| \cos(\phi_{e,\s{Z}}^{(d)}[k]) \right\}_{d \in \mathbb{N}^+},
\end{align}
are uniformly integrable. By Vitali’s convergence theorem, we may pass the limit inside the expectation,
\begin{align}
    \lim_{d \to \infty} \biggl( 
    \mathbb{E}\left[ |\s{Z}^{(d)}[k]| \sin(\phi_{e,\s{Z}}^{(d)}[k]) \right] 
    - \mathbb{E}\left[ |\s{W}[k]| \sin(\phi_{e,\s{W}}[k]) \right]
    \biggr) = 0. \label{eqn:app_HH16}
\end{align}
Furthermore, since $\s{N}^{(d)} \overset{\mathcal{D}}{=} \s{W}$ for all $d$, and the aligned phase differences $\phi_{e,\s{N}}^{(d)}[k]$ and $\phi_{e,\s{W}}[k]$ are identically distributed, we conclude that
\begin{align}
    \mathbb{E}\left[ |\s{N}^{(d)}[k]| \sin(\phi_{e,\s{N}}^{(d)}[k]) \right] 
    = \mathbb{E}\left[ |\s{W}[k]| \sin(\phi_{e,\s{W}}[k]) \right], \quad \forall d \in \mathbb{N}^+. \label{eqn:app_HH17}
\end{align}
Combining \eqref{eqn:app_HH16} and \eqref{eqn:app_HH17} yields
\begin{align}
    \lim_{d \to \infty} \biggl( 
    \mathbb{E}\left[ |\s{Z}^{(d)}[k]| \sin(\phi_{e,\s{Z}}^{(d)}[k]) \right] 
    - \mathbb{E}\left[ |\s{N}^{(d)}[k]| \sin(\phi_{e,\s{N}}^{(d)}[k]) \right]
    \biggr) = 0,
\end{align}
which establishes \eqref{eqn:app_G7a}. An identical argument with sine replaced by cosine proves \eqref{eqn:app_G7b}.

For the variance convergence, by assumption $\mathbb{E}[z_i^4] < \infty$, and so,
\begin{align}
    \mathbb{E}[|\s{Z}^{(d)}[k]|^4] < \infty,
\end{align}
for all $d$. Therefore, the sequences
\begin{align}
    \left\{ |\s{Z}^{(d)}[k]|^2 \sin^2(\phi_{e,\s{Z}}^{(d)}[k]) \right\}_{d \in \mathbb{N}^+}, \quad
    \left\{ |\s{Z}^{(d)}[k]|^2 \cos^2(\phi_{e,\s{Z}}^{(d)}[k]) \right\}_{d \in \mathbb{N}^+},
\end{align}
are uniformly integrable. Applying Vitali's theorem once again, we obtain,
\begin{align}
    \lim_{d \to \infty} \biggl(
    \s{Var}\left[ |\s{Z}^{(d)}[k]| \sin(\phi_{e,\s{Z}}^{(d)}[k]) \right] 
    - \s{Var}\left[ |\s{W}[k]| \sin(\phi_{e,\s{W}}[k]) \right]
    \biggr) = 0.\label{eqn:app_HH21}
\end{align}
As before, since $\s{N}^{(d)} \overset{\mathcal{D}}{=} \s{W}$ and their aligned phase differences are identically distributed, we conclude that,
\begin{align}
    \s{Var}\left[ |\s{N}^{(d)}[k]| \sin(\phi_{e,\s{N}}^{(d)}[k]) \right] 
    = \s{Var}\left[ |\s{W}[k]| \sin(\phi_{e,\s{W}}[k]) \right], \quad \forall d \in \mathbb{N}^+. \label{eqn:app_HH22}
\end{align}
Combining \eqref{eqn:app_HH21} and \eqref{eqn:app_HH22} yields,
\begin{align}
    \lim_{d \to \infty} \biggl( 
    \s{Var}\left[ |\s{Z}^{(d)}[k]| \sin(\phi_{e,\s{Z}}^{(d)}[k]) \right] 
    - \s{Var}\left[ |\s{N}^{(d)}[k]| \sin(\phi_{e,\s{N}}^{(d)}[k]) \right]
    \biggr) = 0,
\end{align}
establishing \eqref{eqn:app_G9a}. Repeating the same steps with cosine in place of sine proves \eqref{eqn:app_G9b}.

\end{proof}

\subsection{Proof of Theorem \ref{thm:highDimentionalNoiseExtention}} \label{sec:proofHighDimentionalSub}

We are now ready to prove Theorem \ref{thm:highDimentionalNoiseExtention}. As before, let $\{z_i\}_{i=0}^{M-1}$ be i.i.d. observations, where each $z_i \in \mathbb{R}^d$ has i.i.d. entries with zero mean, finite variance, and bounded fourth moment
$\mathbb{E}[(z_i[\ell])^4] < \infty$, for all $\ell \in \{0, 1, \ldots, d-1\}$.

Similarly to \eqref{eqn:strongLLN}, we analyze the EfN estimator under the noise statistics of $\{z_i\}_{i=0}^{M-1}$. Applying the SLLN, as $M \to \infty$, we obtain:
\begin{align}
   \s{\hat{X}}[k] e^{-j\phi_{\s{X}}[k]}  = & \frac{1}{M} { \sum_{i=0}^{M-1}  \abs{\s{Z}_i[k]} e^{j\phi_{e,\s{Z}_i}[k]}} 
   \\ & \xrightarrow[]{\s{a.s.}} 
    \mathbb{E} \left[ \abs{\s{Z}_1[k]} \cos\p{\phi_{e,\s{Z}_1}[k]} \right] + j \mathbb{E} \left[ \abs{\s{Z}_1[k]} \sin\p{\phi_{e,\s{Z}_1}[k]} \right], \label{eqn:app_HH25}
\end{align}
where the phase difference term is given by
\begin{align}
    \phi_{e,\s{Z}_i}^{(d)}[k] \triangleq  \frac{2\pi k \hat{\s{R}}_{\s{Z}_i}^{(d)}}{d} + \phi_{\s{Z}_i}[k] - \phi_{\s{X}}[k].
    \label{eqn:app_HH27}
\end{align}
and the corresponding maximal correlation shift $\hat{\s{R}}_{\s{Z}_i}^{(d)}$ is defined by
\begin{align}
    \hat{\s{R}}_{\s{Z}_i}^{(d)} \triangleq  \argmax_{0 \leq r < d} \sum_{k=0}^{d-1} \abs{\s{X}[k]} \abs{\s{Z}_i[k]} \cos\left( \frac{2\pi k r}{d} + \phi_{\s{Z}_i}[k] - \phi_{\s{X}}[k] \right).
    \label{eqn:app_HH26}
\end{align}

Next, we invoke Proposition \ref{prop:G1}, whose assumptions are satisfied in this setting. Let $\s{N}_1^{(d)}$ denotes a noise vector with i.i.d. Gaussian entries that match the first and second moments of the entries of $\s{Z}_1^{(d)}$ (as defined in Proposition \ref{prop:G1}). We note that the results of Theorems \ref{thm:1} and \ref{thm:2} apply to the case of i.i.d. Gaussian entries $\s{N}_1^{(d)}$.  Then, by Proposition \ref{prop:G1}, for each fixed frequency index $k \in \mathbb{N}$, the following convergence results hold,
\begin{align}
    \lim_{d \to \infty} \left( 
    \mathbb{E}\left[ \left| \s{Z}_1^{(d)}[k] \right| \sin\left( \phi_{e,\s{Z}_1}^{(d)}[k] \right) \right] 
    - \mathbb{E}\left[ \left| \s{N}_1^{(d)}[k] \right| \sin\left( \phi_{e,\s{N}_1}^{(d)}[k] \right) \right] 
    \right) = 0, \label{eqn:app_HH28}
\end{align}
\begin{align}
    \lim_{d \to \infty} \left( 
    \mathbb{E}\left[ \left| \s{Z}_1^{(d)}[k] \right| \cos\left( \phi_{e,\s{Z}_1}^{(d)}[k] \right) \right] 
    - \mathbb{E}\left[ \left| \s{N}_1^{(d)}[k] \right| \cos\left( \phi_{e,\s{N}_1}^{(d)}[k] \right) \right] 
    \right) = 0. \label{eqn:app_HH29}
\end{align}
Moreover, the variances of the corresponding expressions also converge,
\begin{align}
    \lim_{d \to \infty} \left( 
    \s{Var}\left[ \left| \s{Z}_1^{(d)}[k] \right| \sin\left( \phi_{e,\s{Z}_1}^{(d)}[k] \right) \right] 
    - \s{Var}\left[ \left| \s{N}_1^{(d)}[k] \right| \sin\left( \phi_{e,\s{N}_1}^{(d)}[k] \right) \right] 
    \right) = 0, \label{eqn:app_HH30}
\end{align}
\begin{align}
    \lim_{d \to \infty} \left( 
    \s{Var}\left[ \left| \s{Z}_1^{(d)}[k] \right| \cos\left( \phi_{e,\s{Z}_1}^{(d)}[k] \right) \right] 
    - \s{Var}\left[ \left| \s{N}^{(d)}[k] \right| \cos\left( \phi_{e,\s{N}_1}^{(d)}[k] \right) \right] 
    \right) = 0. \label{eqn:app_HH31}
\end{align}
By Theorems \ref{thm:1} and \ref{thm:2}, the convergence behavior of the estimator is governed by the variances in \eqref{eqn:app_HH30} and \eqref{eqn:app_HH31}. Therefore, the asymptotic behavior of the estimator for general i.i.d. noise $\{z_i\}$ matches that of the Gaussian i.i.d. case $\{n_i\}$. In particular for \eqref{eqn:FirstRes2}, we have,
\begin{align}
    \phi_{\s{\hat{X}}}[k] - \phi_{\s{X}}[k] & = \arctan \left( \frac{\sum_{i=0}^{M-1}\abs{\s{Z}_i^{(d)}[k]} \sin \left( \phi_{e,\s{Z}_i}[k] \right)}{\sum_{i=0}^{M-1}\abs{\s{Z}_i^{(d)}[k]} \cos \left( \phi_{e,\s{Z}_i}[k] \right)} \right )
    \\ & \xrightarrow[]{\text{a.s.}} \arctan \left( \frac{\mathbb{E} \pp{\abs{\s{Z}_1^{(d)}[k]} \sin \left( \phi_{e,\s{Z}_1}[k] \right)}}{\mathbb{E} \pp{\abs{\s{Z}_1^{(d)}[k]} \cos \left( \phi_{e,\s{Z}_1}[k] \right)}} \right ), \label{eqn:app_H33}
\end{align}
where \eqref{eqn:app_H33} follows from the SLLN as $M \to \infty$. Applying \eqref{eqn:app_HH28} and \eqref{eqn:app_HH29} into \eqref{eqn:app_H33}, yields,
\begin{align}
    \lim_{d \to \infty } \lim_{M \to \infty} \phi_{\s{\hat{X}}}[k] - \phi_{\s{X}}[k] = \lim_{d \to \infty} \arctan \left( \frac{\mathbb{E} \pp{\abs{\s{N}_1^{(d)}[k]} \sin \left( \phi_{e,\s{N}_1}[k] \right)}}{\mathbb{E} \pp{\abs{\s{N}_1^{(d)}[k]} \cos \left( \phi_{e,\s{N}_1}[k] \right)}} \right ) \label{eqn:app_HH34}.
\end{align}
By \eqref{eqn:lawOfLargeNumberPhaseRatio}, the r.h.s. of \eqref{eqn:app_HH34} vanishes for every $d$, and therefore,
\begin{align}
    \lim_{d \to \infty } \lim_{M \to \infty} \phi_{\s{\hat{X}}}[k] - \phi_{\s{X}}[k]  = 0,
\end{align}
almost surely, which proves \eqref{eqn:FirstRes2}. Similarly, for \eqref{eqn:asymptoticComnvergenceOfPhasesCirculant}, we have,
\begin{align}
   \lim_{M\to\infty} \frac{\mathbb{E} |\phi_{\s{\hat{X}}}[k] - \phi_{\s{X}}[k]|^2}{1/M} = \frac{\mathbb{E}\p{\left[ \abs{\s{Z}_1[k]} \sin(\phi_{e,\s{Z}_1}[k]) \right]^2}}{\p{\mathbb{E}{\left[ \abs{\s{Z}_1[k]} \cos(\phi_{e,\s{Z}_1}[k]) \right]}}^2}.
    \label{eqn:app_HH36}
\end{align}
which is similar to \eqref{eqn:B19}. Applying \eqref{eqn:app_HH29} and \eqref{eqn:app_HH30} into \eqref{eqn:app_HH36} yields,
\begin{align}
   \lim_{d\to\infty}\lim_{M\to\infty} \frac{\mathbb{E} |\phi_{\s{\hat{X}}}[k] - \phi_{\s{X}}[k]|^2}{1/M} = \lim_{d \to \infty} \frac{\mathbb{E}\p{\left[ \abs{\s{N}_1[k]} \sin(\phi_{e,\s{N}_1}[k]) \right]^2}}{\p{\mathbb{E}{\left[ \abs{\s{N}_1[k]} \cos(\phi_{e,\s{N}_1}[k]) \right]}}^2} = C_k < \infty.
    \label{eqn:app_HH37}
\end{align}
Finally, for \eqref{eqn:asymptoticComnvergenceOfPhases3}, under the assumption that $x$ satisfies Assumption \ref{assump:1}, then by \eqref{eqn:phaseConvergeneRateForAsymptoticD}, the r.h.s. of \eqref{eqn:app_HH37} converges to,
\begin{align}   
    \lim_{d\to\infty} a^2_{d} \cdot \abs{\s{X}[k]}^2 \cdot C_k  = \frac{1}{2},\label{eqn:app_H38}
\end{align}
where $a_d = \sqrt{2 \log (d)}$. Thus, substituting \eqref{eqn:app_H38} into the r.h.s. of \eqref{eqn:app_HH37} yields,
\begin{align}
    \lim_{d\to\infty} \lim_{M\to\infty} \frac{\mathbb{E} \left[|\phi_{\hat{\s{X}}}[k] - \phi_{\s{X}}[k]|^2\right]}{1/(M \log d)} \cdot \frac{1}{1/(4|\s{X}[k]|^2)} = 1, \label{eqn:app_H39}
\end{align}
which proves \eqref{eqn:asymptoticComnvergenceOfPhases3}.

\section{Proof of Proposition \ref{prop:circulantGauusianNoise}: Circulant Gaussian noise} \label{sec:circulantGauusianNoise}
The proof strategy for Proposition~\ref{prop:circulantGauusianNoise} closely follows that of the i.i.d. Gaussian case (Theorem \ref{thm:1}), with appropriate modifications to handle circulant noise. The necessary assumptions and notations are introduced in Appendix~\ref{subsec:prelimiaries-App-G}. Appendix~\ref{sec:convergenceOfEfNCirculant} establishes the asymptotic convergence of the EfN estimator as $M \to \infty$ under circulant Gaussian noise statistics. In Appendix~\ref{sec:conditionsOnFourierIsCirculant}, we show that conditioning the EfN process on a single Fourier noise coefficient results in a cyclo-stationary process with a cosine trend. Appendix~\ref{sec:imaginartPartVanishingCirculant} extends the vanishing imaginary part result from Appendix~\ref{sec:imaginartPartVanishing} to the setting of circulant Gaussian noise. Similarly, Appendix~\ref{sec:realPartGreaterThan0Circulant} extends the result of Appendix~\ref{sec:realPartGreaterThan0}, showing that the real part remains strictly positive in the circulant case. Finally, Appendix~\ref{sec:proofOfcirculantGauusianNoise} combines the results of the preceding sections to complete the proof of Proposition~\ref{prop:circulantGauusianNoise}.

\subsection{Preliminaries}
\label{subsec:prelimiaries-App-G}
Let $\{y_i\}_{i=0}^{M-1} \sim \mathcal{N}(0, \Sigma)$, where $\Sigma \in \mathbb{R}^{d \times d}$ is a real, symmetric, and circulant covariance matrix with strictly positive eigenvalues (i.e., $\Sigma$ is positive-definite). Let $\s{Y}_i = \mathcal{F} \ppp{y_i} \in \mathbb{C}^d$ denote the DFT of $y_i$. The random vector $\s{Y}_i$ satisfies the following properties:

\begin{enumerate}
    \item \textit{Diagonalization by the DFT.} Since $\Sigma$ is circulant, it is diagonalized by the DFT:
    $\Sigma = F^* \Lambda F$, where $F$ is the DFT matrix, $\Lambda = \mathrm{diag}(\lambda_0, \dots, \lambda_{d-1})$ contains the eigenvalues of $\Sigma$, given by the DFT of its first row. As $\Sigma$ positive-definite, all eigenvalues $\lambda_k \in \mathbb{R}$ and $\lambda_k > 0$ for all $k \in \ppp{0,1, \ldots d-1}$.
    \item \textit{Distribution of Fourier coefficients.} The vector $\s{Y}_i$ is complex Gaussian with distribution $\mathcal{CN}(0, \Lambda)$. Its entries are independent (but not identically distributed) complex Gaussian random variables, satisfying,
    $\mathbb{E}[\s{Y}_i[k]] = 0$, and $\mathbb{E}[\s{Y}_i[k] \overline{\s{Y}_i[\ell]}] = \lambda_k \delta_{k,\ell}$, for every $k, \ell \in \ppp{0,1, \ldots d/2}$.
    \item \textit{Fourier phases distribution.} For any $k$ such that $\lambda_k > 0$, the Fourier coefficient $\s{Y}_i[k]$ is a zero-mean, circularly symmetric complex Gaussian random variable. Hence, the phases $\{\phi_{\s{Y}_i}[k]\}_{k=0}^{d/2}$ are i.i.d. and uniformly distributed on $[-\pi, \pi)$ and independent of the magnitude $\{|\s{Y}_i[k]|\}_{k=0}^{d/2}$.
    \item \textit{Conjugate symmetry.} Since $y_i \in \mathbb{R}^d$, the DFT satisfies the Hermitian symmetry:
    \begin{align}
        \s{Y}_i[d - k] = \overline{\s{Y}_i[k]}, \quad \text{for } 1 \leq k \leq d - 1.
    \end{align}
    Thus, only the first $d/2 + 1$ Fourier coefficients are independent; the rest are determined by conjugate symmetry.
\end{enumerate}

\begin{remark}
    To avoid confusion with the i.i.d. Gaussian case, we use the notation $y_i$ and $\s{Y}_i$, rather than $n_i$ and $\s{N}_i$, to denote Gaussian noise with a symmetric circulant covariance matrix.
\end{remark}

\subsection{The convergence of the Einstein from Noise estimator}  \label{sec:convergenceOfEfNCirculant}
Similar to the derivation in Appendix~\ref{sec:convergenceOfEfNestimator}, the EfN estimator in the setting of circulant Gaussian noise, $\{y_i\}_{i=0}^{M-1} \sim \mathcal{N} \p{0, \Sigma}$, can be expressed explicitly as:
\begin{align}
    \hat{\s{X}}[k] &= \frac{1}{M} \sum_{i=0}^{M-1} \abs{\s{Y}_i[k]} \, e^{j\phi_{\s{Y}_i}[k]} \, e^{j\frac{2\pi k}{d} \s{\hat{R}}_i} \\
    &= \frac{e^{j\phi_{\s{X}}[k]}}{M} \sum_{i=0}^{M-1} \abs{\s{Y}_i[k]} \, e^{j\phi_{\s{Y}_i}[k]} \, e^{j\frac{2\pi k}{d} \s{\hat{R}}_i} \, e^{-j\phi_{\s{X}}[k]} \\
    &= \frac{e^{j\phi_{\s{X}}[k]}}{M} \sum_{i=0}^{M-1} \abs{\s{Y}_i[k]} \, e^{j\phi_{e,i}[k]}, \label{eqn:app_H2}
\end{align}
where the shifts $\s{\hat{R}}_i$ are given by
\begin{align}
    \s{\hat{R}}_i &\triangleq \underset{0 \leq r \leq d-1}{\argmax} {\, \langle{y_i}, \mathcal{T}_r x\rangle}
    \\ & = \underset{0\leq r\leq d-1}{\argmax} \ { \langle \mathcal{F}\ppp{y_i},  {\mathcal{F}\ppp{\mathcal{T}_r x} }\rangle}
    \\ & = \underset{0\leq r\leq d-1}{\argmax} \sum_{k=0}^{d-1} {\abs{\s{X}[k]} \abs{\s{Y}_i[k]} \, \cos \left( \frac{2\pi kr}{d} +  \phi_{\s{Y}_i}[k] - \phi_{\s{X}}[k] \right)}.\label{eqn:app_H5}
\end{align}
and the phase difference is defined as,
\begin{align}
    \phi_{e,i}[k] \triangleq \frac{2\pi k \s{\hat{R}}_i}{d} + \phi_{\s{Y}_i}[k] - \phi_{\s{X}}[k]. \label{eqn:app_H6}
\end{align}

To simplify notation, define for each $r \in \{0, 1, \ldots, d-1\}$,
\begin{align}
    \s{S}_{i}[r] \triangleq \sum_{k=0}^{d-1} \abs{\s{X}[k]} \abs{\s{Y}_i[k]} \, \cos \left( \frac{2\pi kr}{d} + \phi_{\s{Y}_i}[k] - \phi_{\s{X}}[k] \right), \label{equ:maxGaussDefCirculant}
\end{align}
so that $\s{\hat{R}}_i = \argmax_{0 \leq r \leq d-1} \s{S}_{i}[r]$.
We note that for any $0 \leq i \leq M-1$, the random vector $\s{S}_i \triangleq (\s{S}_i[0], \s{S}_i[1], \ldots, \s{S}_i[d-1])^\top$ is jointly Gaussian with zero mean and a circulant covariance matrix (as it is a Fourier transform of the convolution between $y_i$ and the template $x$). Hence, $\s{S}_i$ forms a cyclo-stationary process.
Applying the strong law of large numbers (SLLN), as $ M \to \infty $, we have,
\begin{align}
    \s{\hat{X}}[k] \, e^{-j\phi_{\s{X}}[k]} &= \frac{1}{M} \sum_{i=0}^{M-1} \abs{\s{Y}_i[k]} \, e^{j\phi_{e,i}[k]} \\
    &\xrightarrow[]{\text{a.s.}} \mathbb{E} \left[ \abs{\s{Y}_1[k]} \cos \left( \phi_{e,1}[k] \right) \right] 
    + j \, \mathbb{E} \left[ \abs{\s{Y}_1[k]} \sin \left( \phi_{e,1}[k] \right) \right], \label{eqn:app_H7}
\end{align}
where we have used the fact that the sequences of random variables $\{ \abs{\s{Y}_i[k]} \cos(\phi_{e,i}[k]) \}_{i=0}^{M-1}$ and $\{ \abs{\s{Y}_i[k]} \sin(\phi_{e,i}[k])\}_{i=0}^{M-1}$ are i.i.d. with finite means and variances.
Finally, we define for each $k$,
\begin{align}
    \mu_{\s{A},k} &\triangleq \mathbb{E} \left[ \abs{\s{Y}_1[k]} \sin \left( \phi_{e,1}[k] \right) \right], \label{eqn:app_H10} \\
    \mu_{\s{B},k} &\triangleq \mathbb{E} \left[ \abs{\s{Y}_1[k]} \cos \left( \phi_{e,1}[k] \right) \right], \label{eqn:app_H11}
\end{align}
as the asymptotic imaginary and real parts of $\s{\hat{X}}[k] e^{-j\phi_{\s{X}}[k]}$, respectively.

\subsection{Conditioning on the Fourier frequency noise component}  \label{sec:conditionsOnFourierIsCirculant}
We now extend the result of Lemma~\ref{lemma:conditioning} to the case where the noise follows a general Gaussian distribution with a real, symmetric, and circulant covariance matrix. That is, we consider observations $\{y_i\}_{i=0}^{M-1} \sim \mathcal{N}(0, \Sigma)$, where $\Sigma \in \mathbb{R}^{d \times d}$ is circulant and symmetric.
In this setting, we establish the following result.
\begin{lem} \label{lemma:H1}
    Let $\s{S}_i$ be defined as in~\eqref{equ:maxGaussDefCirculant}, and denote $\mathbb{E}\left[|\s{Y}_i[k]|^2\right] = \lambda_k > 0$ for each $k \in \{0, 1, \ldots, d-1\}$. Then, for every $k \in \left\{ 1, 2, \ldots, \tfrac{d}{2}-1, \tfrac{d}{2}+1, \ldots, d-1 \right\}$, the random vector $\s{S}_i$ conditioned on $\s{Y}_i[k]$ is Gaussian:
    \begin{align}
        \s{S}_i\vert\s{Y}_i[k] \sim \mathcal{N} (\s{\mu}_{k,i}, \s{\Sigma}_{k,i}), \label{eqn:conditionalGaussianCirculant}
    \end{align}
    with mean and covariance given by,
    \begin{align}
        \mu_{k,i}[r] &\triangleq \mathbb{E}\left[ \s{S}_i[r]\vert\s{Y}_i[k] \right]= 2\abs{\s{X}[k]} \abs{\s{Y}_i[k]} \cos \left( \frac{2\pi kr}{d} + \phi_{\s{Y}_i}[k] - \phi_{\s{X}}[k] \right),
        \label{eqn:app_H13_2}
    \end{align}
    for $0\leq r \leq d-1$, and
    \begin{align}
        \s{\Sigma}_{k,i}[r,s] & \triangleq \mathbb{E}\left[ \left(\s{S}_i[r] - \mathbb{E}\s{S}_i[r] \right) \left( \s{S}_i[s] - \mathbb{E}\s{S}_i[s] \right)\vert\s{Y}_i[k] \right] \nonumber\\ 
        & = \sum_{\ell = 0}^{d-1} \frac{\lambda_\ell}{2} \cdot |\widetilde{\s{X}}_k[\ell]|^2 \cos \left( \frac{2\pi \ell}{d}(r-s) \right), \label{eqn:covarainceMatrixCirculant}
    \end{align}
    for $0\leq r,s\leq d-1$, where $\widetilde{\s{X}}_k$ is defined by:
    \begin{align}
        \widetilde{\s{X}}_k[\ell] \triangleq  \begin{cases}
                0  & \s{if}  \ell = k, d-k, \\
             \s{X}[\ell]  & \s{if}  \ell = 0, d/2, \\
                \sqrt{2}\cdot\s{X}[\ell]  & \s{otherwise}.
          \end{cases}
           \label{eqn:XtildeDefenitionCirculant}
    \end{align}
\end{lem}

Note that the conditional process $\s{S}_i\vert\s{Y}_i[k]$ is Gaussian because it is given by a linear transform of i.i.d. Gaussian variables. Also, since its covariance matrix is circulant and depends only on the difference between the two indices, i.e., $\s{\Sigma}_{k,i}[r,s] = \sigma_{k,i}[\abs{r-s}]$, it is cycle-stationary with a cosine trend. The eigenvalues of this circulant matrix are given by the DFT of its first row, and thus its $\ell$-th eigenvalue equals $\lambda_\ell \cdot |\widetilde{\s{X}}_k[\ell]|^2$, for $0 \leq \ell\leq d-1$.

\begin{remark}
    When the noise is i.i.d. Gaussian, that is, $y_i \sim \mathcal{N}(0, \sigma^2 I_{d \times d})$, the eigenvalues of the covariance matrix satisfy $\lambda_\ell = \sigma^2$ for all $\ell \in \{0, 1, \ldots, d-1\}$. In this case, the general setting reduces to the one considered in Lemma~\ref{lemma:conditioning}, thereby recovering its result.
\end{remark}

\begin{proof}[Proof of Lemma~\ref{lemma:H1}]
We recall that if $\ppp{y_i}_{i=0}^{M-1} \sim \mathcal{N}(0, \Sigma)$, for symmetric circulant matrix $\Sigma$, then their DFT coefficients satisfy ${\{|\s{Y}_i[k]|}\}_{k=0}^{d/2}$, and $\{\phi_{\s{Y}_i}[k]\}_{k=0}^{d/2}$ are independent and $\{\phi_{\s{Y}_i}[k]\}_{k=0}^{d-1} \sim \s{Unif}[-\pi, \pi)$.
    By definition of $\s{S}_i$ \eqref{equ:maxGaussDefCirculant}, we have for every $k \neq 0, d/2$,
    \begin{align}
        \nonumber \s{S}_i\pp{r}\vert \s{Y}_i[k] = & 2\abs{\s{X}[k]} \abs{\s{Y}_i[k]} \cos \left( \frac{2\pi kr}{d} + \phi_{\s{Y}_i}[k] - \phi_{\s{X}}[k] \right) 
        \\ & +\sum_{\ell\neq k, d-k} {\abs{\s{X}[\ell]} \abs{\s{Y}_i[\ell]} \, \cos \left( \frac{2\pi \ell r}{d} +  \phi_{\s{Y}_i}[\ell] - \phi_{\s{X}}[\ell] \right)}, \label{eqn:app_H16_1}
    \end{align}
    where we have used the property of $\s{X}[k]=\overline{\s{X}[d-k]}, \  
    \s{Y}_i[k]=\overline{\s{Y}_i[d-k]}$. 
    Clearly, as $\mathbb{E} \pp{\s{Y}_i \pp{\ell}} = 0$, for every $0 \leq \ell \leq d-1$, we have,
    \begin{align}
        \mathbb{E} \pp{{\abs{\s{X}[\ell]} \abs{\s{Y}_i[\ell]} \, \cos \left( \frac{2\pi \ell r}{d} +  \phi_{\s{Y}_i}[\ell] - \phi_{\s{X}}[\ell] \right)}} = 0, \label{eqn:app_H16_2}
    \end{align}
    for every $0 \leq \ell \leq d-1$. Combining \eqref{eqn:app_H16_1} and \eqref{eqn:app_H16_2} results,
    \begin{align}
        \mu_{k,i}[r] & = \mathbb{E}\left[ \s{S}_i[r]\vert\s{Y}_i[k] \right]= 2\abs{\s{X}[k]} \abs{\s{Y}_i[k]} \cos \left( \frac{2\pi kr}{d} + \phi_{\s{Y}_i}[k] - \phi_{\s{X}}[k] \right),
    \end{align}
    proving the first result concerning the means.

    \paragraph{The covariance term.}
    In the following, we derive the covariance term,
    \begin{align}
        \s{\Sigma}_{k,i}[r,s] & \triangleq \mathbb{E}\left[ \left(\s{S}_i[r] - \mathbb{E}\s{S}_i[r] \right) \left( \s{S}_i[s] - \mathbb{E}\s{S}_i[s] \right)\vert\s{Y}_i[k] \right].
    \end{align}
    Let,
    \begin{align}
        \nonumber \rho_{k,i} \pp{r} &  \triangleq \s{S}_i[r] - \mathbb{E}\s{S}_i[r]  
        \\ & = \sum_{\ell\neq k, d-k} {\abs{\s{X}[\ell]} \abs{\s{Y}_i[\ell]} \, \cos \left( \frac{2\pi \ell r}{d} +  \phi_{\s{Y}_i}[\ell] - \phi_{\s{X}}[\ell] \right)},
        \label{eqn:app_H20}
    \end{align}
    and denote, 
    \begin{align}
        \mathcal{I} = \ppp{1,2, \dots k-1, k+1, \dots, d/2-1},
    \end{align}
    which is the indices of the Fourier coefficients, excluding $\ppp{0, k, d/2}$.  
    As the sequences $\ppp{|\s{Y}_i[\ell]|}_{\ell=0}^{d/2}$ and $\ppp{\phi_{\s{Y}_i}[\ell]}_{\ell=0}^{d/2}$ are statistically independent, and satisfy $\s{Y}_i[\ell]=\overline{\s{Y}_i[d-\ell]}$ and $\s{X}[\ell]=\overline{\s{X}[d-\ell]}$, we have,
    \begin{align}
        \nonumber  \rho_{k,i} \pp{r}= \sum_{\ell\neq k, d-k} & {\abs{\s{X}[\ell]} \abs{\s{Y}_i[\ell]} \, \cos \left( \frac{2\pi \ell r}{d} +  \phi_{\s{Y}_i}[\ell] - \phi_{\s{X}}[\ell] \right)} 
         = \\ = & \nonumber \sum_{\ell \in \ppp{0,d/2}} \abs{\s{X}[\ell]}\abs{\s{Y}_i[\ell]} \, \cos \left( \frac{2\pi \ell r}{d} +  \phi_{\s{Y}_i}[\ell] - \phi_{\s{X}}[\ell] \right) \\ & + 2 \cdot  \sum_{\ell \in \mathcal{I}} \abs{\s{X}[\ell]}\abs{\s{Y}_i[\ell]} \, \cos \left( \frac{2\pi \ell r}{d} +  \phi_{\s{Y}_i}[\ell] - \phi_{\s{X}}[\ell] \right),
        \label{eqn:app_H25}
    \end{align}
    where each one of the terms in the sum is independent.    
    Since the terms in the sum on the r.h.s. of \eqref{eqn:app_H25} are independent (i.e., $\mathbb{E} \pp{\s{Y}_i\pp{\ell_1} \overline{\s{Y}_i\pp{\ell_2}}} = \mathbb{E}\pp{\abs{\s{Y}_i\pp{\ell_1}}^2} \delta_{\ell_1,\ell_2}$), it follows that,    
    \begin{align}
        \nonumber\s{\Sigma}_{k,i}&[r,s] =  \mathbb{E} \pp{\rho_{k,i} \pp{r}\rho_{k,i} \pp{s} \vert \s{Y}_i[k] } 
        \\ & = \nonumber \mathbb{E} \pp{{\sum_{\ell \in \ppp{0, d/2}} {\abs{\s{X}[\ell]}^2 \abs{\s{Y}_i[\ell]}^2 \, \cos \left( \frac{2\pi \ell r}{d} +  \phi_{\s{Y}_i}[\ell] - \phi_{\s{X}}[\ell] \right) \cos \left( \frac{2\pi \ell s}{d} +  \phi_{\s{Y}_i}[\ell] - \phi_{\s{X}}[\ell] \right)}}}
        \\ & + 4 \cdot \mathbb{E} \pp{{\sum_{\ell \in \mathcal{I}} {\abs{\s{X}[\ell]}^2 \abs{\s{Y}_i[\ell]}^2 \, \cos \left( \frac{2\pi \ell r}{d} +  \phi_{\s{Y}_i}[\ell] - \phi_{\s{X}}[\ell] \right) \cos \left( \frac{2\pi \ell s}{d} +  \phi_{\s{Y}_i}[\ell] - \phi_{\s{X}}[\ell] \right)}}}. \label{eqn:app_H23}
    \end{align}
    The expectation value in \eqref{eqn:app_H23} is composed of the multiplications of cosines. Applying trigonometric identities, we obtain,
    \begin{align}
        \nonumber \cos & \left( \frac{2\pi \ell r}{d} + \phi_{\s{Y}_i}[\ell] - \phi_{\s{X}}[\ell] \right)\cos \left( \frac{2\pi \ell s}{d} + \phi_{\s{Y}_i}[\ell] - \phi_{\s{X}}[\ell] \right) 
        \\ & = \frac{1}{2}\cos \left( \frac{2\pi \ell (r-s)}{d} \right) + \frac{1}{2}\cos \left( \frac{2\pi \ell (r+s)}{d} + 2 \p{\phi_{\s{Y}_i}[\ell] - \phi_{\s{X}}[\ell]} \right), \label{eqn:trigonometricIdentity2}
    \end{align}
    for every $0 \leq r, s \leq d-1$. Now, since the sequences $\ppp{|\s{Y}_i[\ell]|}_{\ell=0}^{d/2}$ and $\ppp{\phi_{\s{Y}_i}[\ell]}_{\ell=0}^{d/2}$ are independent random variables, with $\mathbb{E} \pp{\abs{\s{Y}_i[k]}^2} = \lambda_k$ and phases $\phi_{\s{Y}_i}[k]$ uniformly distributed over $[-\pi, \pi)$, and by applying the trigonometric identity \eqref{eqn:trigonometricIdentity2}, it follows that,
    \begin{align}
        \nonumber \mathbb{E} & \pp{\abs{\s{Y}_i[\ell]}^2 \cos \left( \frac{2\pi \ell r}{d} + \phi_{\s{Y}_i}[\ell] - \phi_{\s{X}}[\ell] \right) \cos \left( \frac{2\pi \ell s}{d} + \phi_{\s{Y}_i}[\ell] - \phi_{\s{X}}[\ell] \right)}
        \\ & = \frac{1}{2}\mathbb{E} \pp{\abs{\s{Y}_i[\ell]}^2} \cos \left( \frac{2\pi \ell (r-s)}{d} \right) = \frac{\lambda_\ell}{2} \cos \left( \frac{2\pi \ell (r-s)}{d} \right). \label{eqn:app_H25_2}
    \end{align}
    Substituting \eqref{eqn:app_H25_2} into \eqref{eqn:app_H23} leads to,
    \begin{align}
         \nonumber \mathbb{E} & \pp{\rho_{k,i} \pp{r}\rho_{k,i} \pp{s} \vert \s{Y}_i[k] } =  {{\sum_{\ell \in \ppp{0, d/2}} \frac{\lambda_\ell}{2} \cdot  {\abs{\s{X}[\ell]}^2   \, \cos \left( \frac{2\pi \ell}{d}(r-s) \right)}}} 
        \\ & + 4 \cdot {{\sum_{\ell \in \mathcal{I}} \frac{\lambda_\ell}{2} \cdot {\abs{\s{X}[\ell]}^2  \, \cos \left( \frac{2\pi \ell}{d}(r-s) \right)}}}. \label{eqn:app_H26}
    \end{align}
    As for every $\ell \in \mathcal{I}$, $\abs{\s{X}[\ell]} = \abs{\s{X}[d-\ell]}$, we have,
    \begin{align}
        4 \cdot {{\sum_{\ell \in \mathcal{I}} { \frac{\lambda_\ell}{2} \abs{\s{X}[\ell]}^2 \, \cos \left( \frac{2\pi \ell}{d}(r-s) \right)}}} = 2 {{\sum_{\ell \neq \ppp{0,k,d/2,d-k}} { \frac{\lambda_\ell}{2} \abs{\s{X}[\ell]}^2  \, \cos \left( \frac{2\pi \ell}{d}(r-s) \right)}}}. \label{eqn:app_H27}
    \end{align}
    Substituting \eqref{eqn:app_H27} into \eqref{eqn:app_H26}, we get,
    \begin{align}
        \mathbb{E} & \pp{\rho_{k,i} \pp{r}\rho_{k,i} \pp{s} \vert \s{Y}_i[k] } =  \sum_{\ell = 0}^{d-1} \frac{\lambda_\ell}{2} \cdot |\widetilde{\s{X}}_k[\ell]|^2 \cos \left( \frac{2\pi \ell}{d}(r-s) \right),
    \end{align}
    for $\widetilde{\s{X}}_k[\ell]$ as defined in \eqref{eqn:XtildeDefenitionCirculant}, which completes the proof.
\end{proof}

\subsection{Convergence of the Fourier phases} \label{sec:imaginartPartVanishingCirculant}

Similarly to Appendix \ref{sec:imaginartPartVanishing} and Lemma \ref{lemma:A1}, we show here that the imaginary part in \eqref{eqn:app_H7} vanishes. The key observation is that ${\{|\s{Y}_i[k]|}\}_{k=0}^{d-1}$, and $\{\phi_{\s{Y}_i}[k]\}_{k=0}^{d-1}$ are statistically independent and $\{\phi_{\s{Y}_i}[k]\}_{k=0}^{d-1} \sim \s{Unif}[-\pi, \pi)$.

\begin{lem} \label{lemma:H3}
    Recall the definition of $\phi_{e,i}[k]$ in \eqref{eqn:app_H6}. Then,
    \begin{align}
        \mu_{\s{A},k} =\mathbb{E}\left[ \abs{\s{Y}_1[k]} \sin (\phi_{e,1}[k]) \right] = 0,
    \end{align}
    for every $0 \leq k \leq d-1$.
\end{lem}
\begin{proof}[Proof of Lemma~\ref{lemma:H3}]
Let $\s{D}[k]\triangleq\phi_{\s{X}}[k] - \phi_{\s{Y}_1}[k]$, and recall the definition of $\s{\hat{R}}_i$ in \eqref{eqn:OptShiftFourier}, i.e.,
  \begin{align}
    \s{\hat{R}}_i = \underset{0\leq r\leq d-1}{\argmax} \ \sum_{k=0}^{d-1} {\abs{\s{X}[k]} \abs{\s{Y}_i[k]} \, \cos \left( \frac{2\pi kr}{d} +  \phi_{\s{Y}_i}[k] - \phi_{\s{X}}[k] \right)}.
\end{align}
Note that $\s{\hat{R}}_i$ is a function of
\begin{align}
    \s{\hat{R}}_i = \s{\hat{R}}_i \p{\ppp{|\s{Y}_i[k]|}_{k=0}^{d-1}, \ppp{|\s{X}[k]|}_{k=0}^{d-1}, \ppp{\phi_{\s{Y}_i}[k]}_{k=0}^{d-1}, \ppp{\phi_{\s{X}}[k]}_{k=0}^{d-1}},
\end{align}
and it depends on $\phi_{\s{Y}_i}[k]$ and $\phi_{\s{X}}[k]$ only through $\s{D}[k]$. Accordingly, viewing $\s{\hat{R}}_1$ as a function of $\s{D}[k]$, for fixed $\ppp{|\s{Y}_i[k]|}_{k=0}^{d-1}, \ppp{|\s{X}[k]|}_{k=0}^{d-1}$,  we have,
    \begin{align}
        \s{\hat{R}}_1 \left(-\s{D}[0],-\s{D}[1],\ldots,-\s{D}[d-1]\right)  = -\s{\hat{R}}_1 \left(\s{D}[0],\s{D}[1],\ldots,\s{D}[d-1]\right).
        \label{eqn:app_H14}
    \end{align}
Namely, from symmetry arguments, by flipping the signs of all the phases, the location of the maximum flips its sign as well. Then, by the law of total expectation,
\begin{align}
         \mu_{\s{A},k} &= \mathbb{E} \left[ \abs{\s{Y}_1[k]} \sin \left(\frac{2\pi k}{d}\s{\hat{R}}_1 + \phi_{\s{Y}_1}[k] - \phi_{\s{X}}[k]  \right)  \right] \nonumber \\ 
         & = \mathbb{E} \left\{\abs{\s{Y}_1[k]}\cdot\left. \mathbb{E}\left[\sin \left(\frac{2\pi k}{d}\s{\hat{R}}_1 + \phi_{\s{Y}_1}[k] - \phi_{\s{X}}[k]  \right)\right|\{\abs{\s{Y}_1[k]}\}_{k=0}^{d-1}  \right] \right\}.
         \label{eqn:app_H15}
\end{align}
The inner expectation in \eqref{eqn:app_H15} is taken w.r.t. the uniform randomness of the phases $\{\phi_{\s{Y}_1}[k]\}_{k=0}^{d-1} \in [-\pi, \pi)$. However, due to \eqref{eqn:app_H14}, and since the sine function is odd around zero, the integration in~\eqref{eqn:app_H15} nullifies. Therefore,
    \begin{align}
        \left. \mathbb{E}\left[\sin \left(\frac{2\pi k}{d}\s{\hat{R}}_1 + \phi_{\s{Y}_1}[k] - \phi_{\s{X}}[k]  \right)\right|\{\abs{\s{Y}_1[k]}\}_{k=0}^{d-1}  \right] = 0,
        \label{eqn:muAisZeroCirculant}
    \end{align}
and thus $\mu_{\s{A},k}=0$. 
\end{proof}

\subsection{Convergence to non-vanishing signal} \label{sec:realPartGreaterThan0Circulant}

In analogy with Appendix~\ref{sec:realPartGreaterThan0} and Proposition~\ref{prop:2}, we now establish that the real part of~\eqref{eqn:app_H7} does not vanish.

\begin{proposition}\label{prop:H2}
Recall the definition of $\phi_{e,i}[k]$ in \eqref{eqn:app_H6}. Fix $d\in\mathbb{N}$, and assume that $\s{X}[k] \neq 0$ for all $0< k \leq d-1$. Then, for any $0\leq k\leq d-1$,
    \begin{align}   
        \mu_{\s{B},k} \triangleq \mathbb{E}{\left[ \abs{\s{Y}_1[k]} \cos(\phi_{e,1}[k]) \right]}>0.
        \label{eqn:app_H17}
    \end{align}
\end{proposition}

\begin{proof}[Proof of Proposition~\ref{prop:H2}]

By the law of total expectation, we have,
\begin{align}
    \mathbb{E}{\left[ \abs{\s{Y}_1[k]} \cos(\phi_{e,1}[k]) \right]} &= \mathbb{E}\pp{\left.\abs{\s{Y}_1[k]} \cdot\mathbb{E}\left( \cos(\phi_{e,1}[k])\right|\s{Y}_1[k] \right)}\nonumber\\
    & = \mathbb{E}\pp{\left.\abs{\s{Y}_1[k]} \cdot\mathbb{E}\left( \cos\p{\frac{2\pi k\s{\hat{R}}_1}{d} +  \phi_{\s{Y}_1}[k] - \phi_{\s{X}}[k]}\right|\s{Y}_1[k] \right)}.
\end{align}
More explicitly, we can write,
    \begin{align}
         &\mathbb{E}{\left[ \abs{\s{Y}_1[k]} \cos(\phi_{e,1}[k]) \right]} = \nonumber \\ 
         & \frac{1}{2\pi} \int_{0}^{\infty} \mathrm{d}y \ y f_{\abs{{\s{Y}_1[k]}}}(y)  \int_{-\pi}^{\pi} \mathrm{d}\varphi \mathbb{E} \left[\left.\cos \left( \frac{2\pi k}{d}\hat{\s{R}}_1 + \varphi \right)\right|\abs{\s{Y}_1[k]} = y, \phi_{\s{Y}_1}[k] = \phi_{\s{X}}[k] + \varphi \right].\label{eqn:app_H19}
\end{align}
Now, note that the inner integral can be written as,
\begin{align}
         &\int_{-\pi}^{\pi} \mathrm{d}\varphi \mathbb{E} \left[\left.\cos \left( \frac{2\pi k}{d}\hat{\s{R}}_1 + \varphi \right)\right|\abs{\s{Y}_1[k]} = y, \phi_{\s{Y}_1}[k] = \phi_{\s{X}}[k] + \varphi \right]\nonumber\\
         &\qquad=\int_{0}^{\pi} \mathrm{d}\varphi \ \mathbb{E} \left[\left.\cos \left( \frac{2\pi k}{d}\hat{\s{R}}_1 + \varphi \right)\right|\abs{\s{Y}_1[k]}=y, \phi_{\s{Y}_1}[k] = \phi_{\s{X}}[k] + \varphi \right] + \nonumber \\ 
         &\qquad\quad +\int_{0}^{\pi} \mathrm{d}\varphi \  \mathbb{E} \left[\left.\cos \left( \frac{2\pi k}{d}\hat{\s{R}}_1 + \varphi + \pi \right)\right|\abs{\s{Y}_1[k]}=y, \phi_{\s{Y}_1}[k] = \phi_{\s{X}}[k] + \varphi + \pi \right].     \label{eqn:lawOfTotalExpectationInnerIntegralCirculant}
    \end{align}

Next, we apply Proposition \ref{lemma:1}. Using its notation, we define the Gaussian process,
\begin{align}  
    \s{S}^{(+)} = \s{S}_1 \vert \s{Y}_1 [k],  
\end{align}  
where the r.h.s. follows from \eqref{eqn:conditionalGaussianCirculant}. By \eqref{eqn:app_H13_2}, the mean vector of $\s{S}_1 \vert \s{Y}_1 [k]$ has a cosine trend, as assumed in Proposition \ref{lemma:1} in \eqref{eqn:35}. Additionally, $\s{S}_1 \vert \s{Y}_1 [k]$ is a Gaussian cyclo-stationary process, as described in \eqref{eqn:covarainceMatrixCirculant}. The final condition to verify is~\eqref{eq:pos_prob_Cr_assumption}, which is satisfied by Lemma~\ref{lem:positive_prob_each_argmax_conditional} (and applies also to circulant Gaussian noise statistics as well).

Since the conditional distribution of $\hat{\s{R}}_1$ given $\{\abs{\s{Y}_1[k]}=y, \phi_{\s{Y}_1}[k] = \phi_{\s{X}}[k] + \varphi\}$ matches that of $\hat{\s{R}}^{(+)}$ in \eqref{eqn:27}, and similarly, given $\{\abs{\s{Y}_1[k]}=y, \phi_{\s{Y}_1}[k] = \phi_{\s{X}}[k] + \varphi+\pi\}$, it matches $\hat{\s{R}}^{(-)}$ in \eqref{eqn:28}, the sum of the integrands on the right-hand side of \eqref{eqn:lawOfTotalExpectationInnerIntegralCirculant} equals the left-hand side of \eqref{eqn:lemma1conc}. By Proposition \ref{lemma:1}, this sum is positive for all $\varphi \in [0, \pi]$. Together with \eqref{eqn:app_H19}, this completes the proof of Proposition~\ref{prop:H2}.  

\end{proof}

\subsection{Proof of Proposition \ref{prop:circulantGauusianNoise}} \label{sec:proofOfcirculantGauusianNoise}
We are now ready to prove Proposition \ref{prop:circulantGauusianNoise}.
By the definition of the phase difference between the template $x$ and the EfN estimator $\hat{x}$ (as in \eqref{eqn:estimatorPhase}), we have,
\begin{align}
    \phi_{\s{\hat{X}}}[k] - \phi_{\s{X}}[k] =  \arctan \left( \frac{\sum_{i=0}^{M-1}\abs{\s{Y}_i[k]} \sin \left( \phi_{e,i}[k] \right)}{\sum_{i=0}^{M-1}\abs{\s{Y}_i[k]} \cos \left( \phi_{e,i}[k] \right)} \right ),
\end{align}
Using the continuous mapping theorem, it is evident that it suffices to prove that,
\begin{align}
    \frac{\sum_{i=0}^{M-1}\abs{\s{Y}_i[k]} \sin \left( \phi_{e,i}[k] \right)}{\sum_{i=0}^{M-1}\abs{\s{Y}_i[k]} \cos \left( \phi_{e,i}[k] \right)}\xrightarrow[]{\s{a.s.}} 0.
\end{align}
This, however, follows by applying the SLLN,
\begin{align}
   \frac{\sum_{i=0}^{M-1}\abs{\s{Y}_i[k]} \sin \left( \phi_{e,i}[k] \right)}{\sum_{i=0}^{M-1}\abs{\s{Y}_i[k]} \cos \left( \phi_{e,i}[k] \right)}
    \xrightarrow[]{\s{a.s.}} \frac{\mu_{\s{A},k}}{\mu_{\s{B},k}},
    \label{eqn:lawOfLargeNumberPhaseRatioCirc}
\end{align}
where $\mu_{\s{A},k} \triangleq \mathbb{E}\left[ \abs{\s{Y}_1[k]} \sin (\phi_{e,1}[k]) \right]$ and $\mu_{\s{B},k} \triangleq \mathbb{E}\left[ \abs{\s{Y}_1[k]} \cos (\phi_{e,1}[k]) \right]$, defined in \eqref{eqn:app_H10}, and \eqref{eqn:app_H11}, respectively. By Lemma \ref{lemma:H1}, $\mu_{\s{A},k} = 0$, while by Proposition \ref{prop:H2}, we have that $\mu_{\s{B},k} > 0$, and thus their ratio converges a.s. to zero by the continuous mapping theorem. Thus, we proved that $\phi_{\s{\hat{X}}}[k] \xrightarrow[]{\s{a.s.}} \phi_{\s{X}}[k]$. Finally, we prove the convergence rate, given in \eqref{eqn:asymptoticComnvergenceOfPhases3}. According to Proposition \ref{prop:convergeneInExpectation}, whose assumptions apply for the case of circulant Gaussian noise statistics as well, we have,
    \begin{align}
       \lim_{M\to\infty} \frac{\mathbb{E} |\phi_{\s{\hat{X}}}[k] - \phi_{\s{X}}[k]|^2}{1/M} = \frac{\mathbb{E}\p{\left[ \abs{\s{Y}_1[k]} \sin(\phi_{e,1}[k]) \right]^2}}{{\mathbb{E}{\left[ \abs{\s{Y}_1[k]} \cos(\phi_{e,1}[k]) \right]}}^2} < \infty,
    \end{align}
which completes the proof of the Proposition.

\end{appendices}
\end{document}